\documentclass[10pt,american,twocolumnfalse]{article}
\usepackage[utf8]{inputenc}
\usepackage[a4paper]{geometry}
\usepackage{amsmath}
\usepackage{amsthm}
\usepackage{amssymb}
\usepackage{cancel}
\usepackage{stmaryrd}
\usepackage{wasysym}

\makeatletter
\theoremstyle{plain}
\newtheorem{thm}{\protect\theoremname}[section]
  \theoremstyle{plain}
  \newtheorem{lem}[thm]{\protect\lemmaname}
  \theoremstyle{remark}
  \newtheorem{rem}[thm]{\protect\remarkname}
  \theoremstyle{definition}
  \newtheorem{defn}[thm]{\protect\definitionname}
  \theoremstyle{plain}
  \newtheorem{cor}[thm]{\protect\corollaryname}

\numberwithin{equation}{section}
\usepackage{cite}

\makeatother

\usepackage{babel}
  \providecommand{\corollaryname}{Corollary}
  \providecommand{\definitionname}{Definition}
  \providecommand{\lemmaname}{Lemma}
  \providecommand{\remarkname}{Remark}
\providecommand{\theoremname}{Theorem}

\begin{document}

\title{Relative entropy for coherent states from Araki formula}

\author{Horacio Casini\thanks{e-mail: casini@cab.cnea.gov.ar}, Sergio Grillo\thanks{e-mail: sergiog@cab.cnea.gov.ar },
Diego Pontello\thanks{e-mail: diego.pontello@ib.edu.ar}}
\maketitle
\begin{center}
\textsl{Centro Atómico Bariloche and Instituto Balseiro,}\\
\textsl{ Comisión Nacional de Energía Atómica (CNEA),}\\
\textsl{ Consejo Nacional de Investigaciones Científicas y Técnicas
(CONICET),}\\
\textsl{ Universidad Nacional de Cuyo (UNCuyo), }\\
\textsl{Av. E. Bustillo 9500 R8402AGP,}\\
\textsl{San Carlos de Bariloche, Río Negro, Argentina}
\par\end{center}
\begin{abstract}
We make a rigorous computation of the relative entropy between the
vacuum state and a coherent state for a free scalar in the framework
of algebraic quantum field theory (AQFT). We study the case of the Rindler wedge. Previous calculations
including path integral methods and computations from the lattice
give a result for such relative entropy which involves integrals of
expectation values of the energy-momentum stress tensor along the
considered region. However, the stress tensor is in general non-unique.
That means that if we start with some stress tensor, then we can “improve”
it adding a conserved term without modifying the Poincaré charges.
On the other hand, the presence of such an improving term affects
the naive expectation for the relative entropy by a non-vanishing
boundary contribution along the entangling surface. In other words,
this means that there is an ambiguity in the usual formula for the
relative entropy coming from the non-uniqueness of the stress tensor.
The main motivation of this work is to solve this puzzle. We first
show that all choices of stress tensor except the canonical one are
not allowed by positivity and monotonicity of the relative entropy.
Then we fully compute the relative entropy between the vacuum and
a coherent state in the framework of AQFT using the Araki formula
and the techniques of modular theory. After all, both results coincide
and give the usual expression for the relative entropy calculated
with the canonical stress tensor.
\end{abstract}

\section{Introduction\label{sec:Intro}}

The algebraic description of quantum field theory (AQFT) focuses on
the local algebras of operators generated by fields in regions of
the space rather than the field operators themselves. This gives a
“basis independent\char`\"{} formulation which does not depend on
the particular fields used for the description of the theory. Statistical
properties of the state in these local algebras have been the subject
of much recent interest in different areas of physics ranging from
holography to condensed matter. Given one or more states and algebras,
several entropic quantities can be defined which give natural measures
of the statistics of fluctuations. In a certain sense, these assignations
of numbers to algebras in AQFT is analogous to the study of correlators
in the approach based on field operators.

In actual computations in specific models, it is customary and useful
to assume a cutoff model, such as a lattice, and proceed to the computation
taking the continuum limit as a final step. In general, we expect that the quantity computed belongs to the continuous theory as far as the
result does not depend on the regularization. In the cutoff model,
given a global pure state $\Phi\in\mathcal{H}$ one can consider the
reduced density matrix $\rho_{\Phi}^{R}$ in a space region $R$ of
a lattice and compute its von Neumann (vN) entropy
\begin{equation}
S_{\Phi}^{R}=-\textrm{tr}\rho_{\Phi}^{R}\log\rho_{\Phi}^{R}\,.
\end{equation}
This is divergent and not well-defined in the continuum due to a large
amount of entanglement of UV modes between both sides of the region
boundary. However, given two states $\Psi$ and $\Omega$ we can also
compute the relative entropy
\begin{equation}
S_{R}\left(\Phi\mid\Omega\right)=\textrm{tr}\rho_{\Phi}^{R}\left(\log\rho_{\Phi}^{R}-\log\rho_{\Omega}^{R}\right)\,,\label{rela1}
\end{equation}
which is much better behaved than the entropy (see for example \cite{rela_holo,Lashkari:2014yva,Sarosi_2016,Sarosi_2016b,Ruggiero:2016khg,Hollands:2017mlk,Longo_xu,Hollands:2018wzi,scalar_nuestro,Murciano:2018cfp,Xu,Longo_2018,Longo_2019,Hollands:2019czd}
for actual calculations). In fact, the relative entropy has an expression
directly in the continuous theory for type III algebras in terms of the Araki formula \cite{Araki_entropy}. This shows it is free from ambiguities.
Relative entropy is an important quantity in quantum information that
measures distinguishability between states. It is always positive
and increasing for fixed states under increasing algebras. It has
recently been very useful in holography to understand the bulk-boundary
map \cite{Jafferis_2015,Dong_2016,Faulkner_2017,Faulkner_2018} and
in the proof of the quantum null energy condition \cite{Ceyhan:2018zfg}.

Another object that has a nice continuum limit is the following one
parameter group of unitaries
\begin{equation}
(\rho_{\Omega}^{R})^{is}\otimes(\rho_{\Omega}^{R'})^{-is}\,,
\end{equation}
where $R'$ is the complement of $R$ and we are assuming there is
a decomposition of the full operator algebra as a tensor product of
the algebras in $R$ and $R'$. This one-parametric group is called
the modular group. The generator,
\begin{equation}
K_{\Omega}=-K_{R}\otimes1+1\otimes K_{R'}\,,\hspace{1cm}K_{R}=-\log\rho_{\Omega}^{R}\,,\label{decom}
\end{equation}
is called the modular Hamiltonian. A well-known case where the modular
Hamiltonian can be computed exactly is the case when $R$ is the Rindler
wedge corresponding to a spatial slice $x^{1}>0$ at $x^{0}=0$, and
the state is the vacuum. In this case, $K_{\Omega}=2\pi K_{1}$ with
$K_{1}$ being the boost generator. In terms of the energy density operator
we can write
\begin{equation}
K_{\Omega}=2\pi\int d^{d-1}x\,x^{1}\,T_{00}(x)\,.\label{rin}
\end{equation}
Returning to the relative entropy, it is useful to write \eqref{rela1}
as
\begin{equation}
S_{R}\left(\Phi\mid\Omega\right)=\Delta\langle K_{R}\rangle-\Delta S_{R}\label{iii}
\end{equation}
where
\begin{eqnarray}
\Delta\langle K_{R}\rangle & = & \textrm{tr}\rho_{\Phi}^{R}K_{R}-\textrm{tr}\rho_{\Omega}^{R}K_{R}\,,\\
\Delta S & = & S_{\Phi}^{R}-S_{\Omega}^{R}\,.
\end{eqnarray}
Written in this way, the relative entropy is the variation in expectation
value of an operator minus the variation in the entropy between the
two states. The positivity of relative entropy means that $\Delta\langle K_{R}\rangle\ge\Delta S_{R}$.
In this form, when $R$ is the Rindler wedge, this inequality has
been related to Bekenstein's bound on entropy \cite{Casini_bek}.

Even if the relative entropy is well defined in the continuum, a mathematically
rigorous definition of the continuum limit of the two terms in \eqref{iii}
has not been worked out in the literature yet. One difficulty is that
the operator $K_{R}$ is only half of the modular Hamiltonian \eqref{decom}.
Even if the modular Hamiltonian has a good operator limit in the continuum,
its half part $K_{R}$ is at most a sesquilinear form. If we focus
for simplicity on the case of the half space and where $\Omega$ is
the vacuum, we could induce from \eqref{rin} that
\begin{equation}
K_{R}=2\pi\int_{x^{1}>0}\negthickspace d^{d-1}x\,x^{1}\,T_{00}(x)\,.\label{half}
\end{equation}
This is not a well-defined operator in Hilbert space because its fluctuation
$\langle\Omega,K_{R}^{2}\Omega\rangle$ diverges. However, expectation
values as in $\Delta\langle K_{R}\rangle$ can still be computed.
Another more important issue is that the act of cutting the modular
Hamiltonian in two pieces generate ambiguities. We are allowed for
example to add field operators localized at the boundary such that
$K_{R}$ has still the same localization and commutation relations
with operators inside $R$. Another view of the same problem is that
hidden in expression \eqref{half}, there is an ambiguity related
to the non-uniqueness of the stress tensor. For example, for the free
Hermitian scalar field, starting from the canonical stress tensor
\begin{equation}
T_{\mu\nu}^{can}=:\partial_{\mu}\phi\partial_{\nu}\phi-\frac{1}{2}\eta_{\mu\nu}\left(\partial_{\sigma}\phi\partial^{\sigma}\phi-m^{2}\phi^{2}\right):\,,\label{eq:stress}
\end{equation}
we can add an ``improving term'' to obtain a new stress tensor
\begin{equation}
T_{\mu\nu}=T_{\mu\nu}^{can}+\frac{\lambda}{2\pi}\,(\partial_{\mu}\partial_{\nu}-g_{\mu\nu}\partial^{2}):\phi^{2}:\,.\label{eq: stress_imp}
\end{equation}
The Poincaré generators obtained from \eqref{eq: stress_imp} equal
the ones obtained from \eqref{eq:stress}, since both expressions
differ in a boundary term which vanishes when the integration region
is the whole space. However, the expression \eqref{half} for $K_{R}$
involves an integration in a semi-infinite region, and hence the presence
of an improving term adds a non zero extra boundary term to the result,
\begin{equation}
K_{R}\rightarrow K_{R}+\lambda\int_{x^{1}=0}\negthickspace d^{d-2}x\,:\phi^{2}\left(x\right):\,.\label{eq:hamil_improved}
\end{equation}
This is essentially the only boundary term we can add with the correct
dimensions and that does not require a dimensionful coefficient with
negative dimensions. This can have nonzero expectation values for
certain states and makes the definition of $\Delta\langle K_{R}\rangle$
ambiguous.

Since the relative entropy is well defined, this ambiguity must be
compensated by another one in the definition of $\Delta S_{R}$ in
\eqref{iii}. This is the subtraction of two divergent quantities
and again we do not have a mathematically rigorous definition in the
continuum. We can make this definition unambiguous in a natural way
by using a particular regularization of entropy that has been proposed
in the literature \cite{Casini_2006,chmy}. The idea is to associate
the entropy (for a pure state) with half the mutual information $I(R_{\epsilon}^{+},R_{\epsilon}^{-})$
between two non-intersecting regions on both sides of the boundary
of $R$. The regions $R_{\epsilon}^{\pm}$ are displaced a distance
$\epsilon$ from the boundary of $R$. For the case of the Rindler
wedge we can take $R_{\epsilon}^{+}$ formed by points with $x^{1}>\epsilon$
and $R_{\epsilon}^{-}$ formed by points with $x^{1}<-\epsilon$.
The mutual information is also a relative entropy and is well defined
in the continuum. Then, a well-defined $\Delta S_{R}$ is given by
\begin{equation}
\Delta S_{R}=\frac{1}{2}\lim_{\epsilon\rightarrow0}(I_{\Phi}(R_{\epsilon}^{+},R_{\epsilon}^{-})-I_{\Omega}(R_{\epsilon}^{+},R_{\epsilon}^{-}))\,.\label{rigo}
\end{equation}
When it is computed in the lattice, it coincides with the usual $\Delta S_{R}$.

Defining $\Delta S_{R}$ rigorously through \eqref{rigo}, then $\Delta\langle K_{R}\rangle$
is also well defined through 
\begin{equation}
\Delta\langle K_{R}\rangle=S_{R}\left(\Phi\mid\Omega\right)+\Delta S_{R}\,.
\end{equation}
Then the question that motivates this paper is whether this definition
agrees with the expectation value of \eqref{half}. In such a case,
boundary terms in this expression should be automatically fixed. In
particular, we should be able to study which value of the improvement
term is the correct one for a scalar field in \eqref{eq: stress_imp}.

In order to (partially) settle this issue, in this paper we analyze
the relative entropy between a coherent state for a free scalar field
and the vacuum in the Rindler wedge. Coherent states are states formed
out by acting on the vacuum with a unitary operator that is the exponential
of the smeared field, i.e. 
\begin{equation}
\Phi=\mathrm{e}^{i\int d^{d-1}x\,\left[\varphi\left(\bar{x}\right)f_{\varphi}\left(\bar{x}\right)+\pi\left(\bar{x}\right)f_{\pi}\left(\bar{x}\right)\right]}\Omega\,,\label{coh_intro}
\end{equation}
where $\varphi\left(\bar{x}\right):=\phi\left(0,\bar{x}\right)$ and
$\pi\left(\bar{x}\right):=\partial_{0}\phi\left(0,\bar{x}\right)$.
For the purpose of the definition \eqref{rigo}, we can represent
the same state with a different vector $\tilde{\Phi}=UU'\Omega$, where
$U$ is a unitary belonging to the region $R$ and $U'$ is a unitary
belonging to its complementary region $R'$. Indeed, we can replace
each of the smooth functions $f_{\varphi}\left(\bar{x}\right)$ and
$f_{\pi}\left(\bar{x}\right)$ in \eqref{coh_intro} by the sum of
two new smooth functions,
\begin{equation}
f_{\varphi}\rightarrow f_{\varphi,R}+f_{\varphi,R'}\,,\quad f_{\pi}\rightarrow f_{\pi,R}+f_{\pi,R'}\,,
\end{equation}
such that $f_{\varphi,R},f_{\pi,R}$ vanish inside $R'$ and $f_{\varphi,R'},f_{\pi,R'}$
vanish inside $R$. We must also require that $f_{\varphi,R}\equiv f_{\varphi}$
inside $R_{\epsilon}^{+}$ and $f_{\varphi,R'}\equiv f_{\varphi}$
inside $R_{\epsilon}^{-}$ (idem for $\pi$). Under this assumptions,
the new state $\tilde{\Phi}=UU'\Omega$, defined through
\begin{equation}
U=\mathrm{e}^{i\int d^{d-1}x\,\left[\varphi\left(\bar{x}\right)f_{\varphi,R}\left(\bar{x}\right)+\pi\left(\bar{x}\right)f_{\pi,R}\left(\bar{x}\right)\right]}\quad\mathrm{and}\quad U'=\mathrm{e}^{i\int d^{d-1}x\,\left[\varphi\left(\bar{x}\right)f_{\varphi,R'}\left(\bar{x}\right)+\pi\left(\bar{x}\right)f_{\pi,R'}\left(\bar{x}\right)\right]},
\end{equation}
represents the same state as $\Phi$ in the algebra of the region
$R_{\epsilon}^{+}\cup R_{\epsilon}^{-}$. In fact, the above computation
can be done because of the presence of the finite corridor of width
$2\epsilon$. Moreover, we have that the operator $U$ (respectively, $U'$)
acts, by adjoint action, as an automorphism of the algebra of the
region $R_{\epsilon}^{+}$ (respectively, $R_{\epsilon}^{-}$), and as the
identity transformation over the algebra of the region $R_{\epsilon}^{-}$
(resp. $R_{\epsilon}^{+}$). Such automorphisms do not change the
mutual information, and with our definition \eqref{rigo}, we automatically
have $\Delta S_{R}=0$ for these states. 

Hence, the question simplifies to see whether for coherent states
\begin{equation}
S_{R}\left(\Phi\mid\Omega\right)=2\pi\int_{x^{1}>0}\negthickspace d^{d-1}x\,x^{1}\,\left\langle \Phi,T_{00}\left(\bar{x}\right)\Phi\right\rangle \,,\label{cali}
\end{equation}
and which is the right improvement term. Notice that coherent states
can change the expectation value of $:\phi^{2}:$.

In section \ref{sec:lamda_bounds}, assuming that \eqref{cali} is
correct for some improvement, we show that the only possibility
is the canonical stress tensor, i.e. $\lambda=0$. We show this by
imposing bounds which come from the positivity and monotonicity of
the relative entropy.

In the rest of the paper, we actually compute the relative entropy
using Araki formula and show the result \eqref{cali} is correct for
the canonical stress tensor. We note that, while this paper was being
prepared, a similar calculation by R. Longo has appeared in the literature
\cite{Longo_2019}. A simpler case where the unitary has support inside
the wedge has previously appeared in \cite{nima}. Our paper differs
from the one by Longo in motivation, scope, and several details, while
there is an overlap in the main technical ideas. 

To make this article as self-contained as possible, in section
\ref{sec:Wightman-theory-of} we briefly review the algebraic formulation
of the free scalar field. Because of a forthcoming necessity, we consider
two different approaches. The first one is the usual approach where
we define the net of algebras associated to spacetime regions. The
second one consists in defining the local algebras associated to spatial
sets belonging to a common Cauchy surface. We also explain how these
two approaches are related. In section \ref{sec:Modular-theory} we
review the basic concepts of the modular theory of von Neumann algebras.
In particular, we introduce the modular operator used to derive the
modular Hamiltonian and the modular flow. We also discuss the theorems
of Tomita-Takesaki and Bisognano-Whichmann. And finally, we introduce
the relative modular operator used in the definition of the relative
entropy for general von Neumann algebras. The reader who is familiar
with these concepts may skip these sections and go directly to \ref{sec:Explicit-calculation},
where we explicitly compute the proposed relative entropy. We study
separately the (trivial) case when the coherent state belongs to the
wedge algebra, and the more interesting (and also more difficult)
case when the coherent state has a non-vanishing density along the
entangling surface. In this section, we also first study some general
aspects concerning the relative entropy for coherent states which
applies to any region. We provide a complete mathematical rigorous
proof of all the results. For a better reading of the article, the
proof of some theorems and some tedious but straightforward calculations
were placed into the appendixes.

\section{Boundary terms in the relative entropy\label{sec:lamda_bounds}}

According to the discussion above, there is an ambiguity on the expression
\eqref{cali} for the relative entropy of a coherent state coming
from the different possible choices of an improving term for the stress
energy-momentum tensor. According to \eqref{eq:hamil_improved}, the
relative entropy could be written as the usual contribution with the
canonical stress tensor plus a boundary term coming from the improving
\begin{equation}
S_{R}\left(\Phi\mid\Omega\right)=\lambda\int_{x_{1}=0}\negthickspace d^{d-2}x\,\left\langle \Phi,\phi^{2}\left(\bar{x}\right)\Phi\right\rangle +2\pi\int_{x_{1}>0}\negthickspace d^{d-1}x\,x^{1}\left\langle \Phi,T_{00}^{can}\left(\bar{x}\right)\Phi\right\rangle \,.\label{eq:re_lambda}
\end{equation}
In this section we assume this formula is correct and show that the
only consistent choice is $\lambda=0$.

A general coherent state can be written as in \eqref{coh_intro} with
$f_{\varphi},f_{\pi}\in\mathcal{S}\left(\mathbb{R}^{d-1},\mathbb{R}\right)$.\footnote{$\mathcal{S}\left(\mathbb{R}^{n},\mathbb{R}\right)$ denotes the Schwartz
space of real, smooth and exponentially decreasing functions at infinity.} In this case, a straightforward computation from \eqref{eq:re_lambda}
gives 
\begin{equation}
S_{R}\left(\Phi\mid\Omega\right)=\lambda\int_{x_{1}=0}\negthickspace d^{d-2}x\,f_{\pi}\left(\bar{x}\right)^{2}+2\pi\int_{x_{1}>0}\negthickspace d^{d-1}x\,\frac{1}{2}\left(f_{\varphi}\left(\bar{x}\right)^{2}+\left(\nabla f_{\pi}\left(\bar{x}\right)\right)^{2}+m^{2}f_{\pi}\left(\bar{x}\right)^{2}\right)\,.\label{eq:re_coherent}
\end{equation}
Regardless of what should be the true value for $\lambda$, if we want \eqref{eq:re_lambda} and \eqref{eq:re_coherent} represent real
expressions for a relative entropy, they must satisfy all the properties
known for a relative entropy. In particular we concentrate on
the positivity 
\begin{equation}
S_{R}\left(\Phi\mid\Omega\right)\geq0\,,\label{eq:pos_re}
\end{equation}
and the monotonicity, that for the case of wedges implies 
\begin{equation}
\left.S_{R}\left(\Phi\mid\Omega\right)\right|_{\mathcal{W}_{y}}\geq\left.S_{R}\left(\Phi\mid\Omega\right)\right|_{\mathcal{W}_{y'}}\,,\quad\textrm{for any }y'\geq y\,,\label{eq:mon_re}
\end{equation}
where $\left.S_{R}\left(\Phi\mid\Omega\right)\right|_{\mathcal{W}_{y}}$
is the relative entropy for the states $\Psi,\Omega$ but associated
to the algebra of the translated Rindler wedge $\mathcal{W}_{y}:=\left\{ x\in\mathbb{R}^{d}\,:\,x^{1}-y>\left|x^{0}\right|\right\} $.
In fact, $\mathcal{W}_{y}$ is obtained applying a translation of
amount $y$, in the $x^{1}$ positive direction, to the original Rindler
wedge $\mathcal{W}$. From now on, we denote $S_{R}\left(y\right):=\left.S_{R}\left(\Phi\mid\Omega\right)\right|_{\mathcal{W}_{y}}$.

Therefore, the strategy we adopt is to choose conveniently functions
$f_{\varphi}$ and $f_{\pi}$ and impose \eqref{eq:pos_re} and \eqref{eq:mon_re}
on \eqref{eq:re_coherent} in order to bound the allowed values for
$\lambda$. In fact, we show that from positivity we obtain $\lambda\geq0$
and from the monotonicity we obtain $\lambda\leq0$, an hence it must
be 
\begin{equation}
\lambda=0\,.
\end{equation}
Then we conclude that, if we assume that \eqref{cali} is the correct
result for the relative entropy, such an expression holds for the canonical
stress-energy-momentum tensor \eqref{eq:stress}.

Before we start, we make two simplifications. The first one, which
is obvious, is to take $f_{\varphi}\equiv0$ and denote $f:=f_{\pi}$.
The second one is to work in $d=1+1$ dimensions. The general result
for any dimensions could be obtained easily from the former.

\subsection{Lower bound from positivity}

We start with the expression

\begin{equation}
S_{R}\left(\Phi\mid\Omega\right)=\lambda f\left(0\right)^{2}+\pi\int_{0}^{+\infty}dx\,x\left(f'\left(x\right)^{2}+m^{2}f\left(x\right)^{2}\right)\,,
\end{equation}
where $f$ is a real-valued function belonging to $\mathcal{S}\left(\mathbb{R}\right)$.
Then, the positivity of the relative entropy means that

\begin{equation}
0\leq\lambda f\left(0\right)^{2}+\pi\int_{0}^{+\infty}dx\,xf'\left(x\right)^{2}+\pi m^{2}\int_{0}^{+\infty}dx\,xf\left(x\right)^{2}\,.
\end{equation}
By scaling the function $f(x)\rightarrow f(x/\beta)$ the first two
terms of the right-hand side are constant while the last one gets
multiplies by $\beta^{2}$. Hence, we can make the last term as small
as we want and simply take $m=0$ in the following. Taking $f$ such
that $f\left(0\right)\neq0$ we get 
\begin{equation}
0\leq\lambda+\pi\frac{\int_{0}^{+\infty}dx\,xf'\left(x\right)^{2}}{f\left(0\right)^{2}}\,.\label{eq:pos_ineq_2}
\end{equation}
Now, we introduce a convenient family of real functions $\left(f_{a}\right)_{a>0}\in\mathcal{S}\left(\mathbb{R}\right)$
given by 
\begin{equation}
f_{a}\left(x\right):=\log\left(\frac{x}{L}+a\right)\mathrm{e}^{-\frac{x}{L}}\,,\quad x\geq0\,,
\end{equation}
and where $L>0$ is a dimensionful fixed constant.\footnote{The functions $f_{a}$ are smoothly extended to the whole real line.
Such an extension is guaranteed by a theorem due to Seeley \cite{seeley}.\label{fn:c_inf_ext}} A straightforward computation shows that the integral in equation
\eqref{eq:pos_ineq_2} behaves as

\begin{equation}
\int_{0}^{+\infty}dx\,xf_{a}'\left(x\right)^{2}=-\log\left(a\right)+\mathcal{O}\left(1\right)\,,\quad a\apprge0\,.\label{eq:bound1}
\end{equation}
Then replacing \eqref{eq:bound1} into \eqref{eq:pos_ineq_2} we get
\begin{equation}
0\leq\lambda-\frac{L^{2}}{4}\pi\frac{\log\left(a\right)+\mathcal{O}\left(1\right)}{\log^{2}\left(a\right)}\,.
\end{equation}
Finally, taking the limit $a\rightarrow0^{+}$ we get the desired
result 
\begin{equation}
\lambda\geq0.\label{eq: lower_bound}
\end{equation}

\subsection{Upper bound from monotonicity}

We start with the expressions 
\begin{eqnarray}
S_{R}\left(0\right)\!\!\!\! & = & \!\!\!\!\lambda f\left(0\right)^{2}+\pi\int_{0}^{+\infty}dx\,x\left(f'\left(x\right)^{2}+m^{2}f\left(x\right)^{2}\right)\,,\\
S_{R}\left(y\right)\!\!\!\! & = & \!\!\!\!\lambda f\left(y\right)^{2}+\pi\int_{y}^{+\infty}dx\,\left(x-y\right)\left(f'\left(x\right)^{2}+m^{2}f\left(x\right)^{2}\right)\,,
\end{eqnarray}
where $f$ is a real-valued function belonging to $\mathcal{S}\left(\mathbb{R}\right)$.
We can eliminate the mass terms by scaling as in the previous section.
The monotonicity $S_{R}\left(0\right)\geq S_{R}\left(y\right)$ for
$y\geq0$ reads 
\begin{eqnarray}
\lambda\left(f\left(y\right)^{2}-f\left(0\right)^{2}\right) & \leq & \pi\int_{0}^{y}dx\,xf'\left(x\right)^{2}+\pi y\int_{y}^{+\infty}dx\,f'\left(x\right)^{2}\,.\label{eq: ineq_mon}
\end{eqnarray}
Now, we introduce a convenient family of functions parametrized with
the constants $\alpha\in\left(0,\frac{1}{2}\right),\delta\in\left(0,1\right),y>0,\epsilon>0$
given by 
\begin{equation}
f_{\alpha,\delta,y,\epsilon}\left(x\right):=g_{\alpha,\delta,y}\left(x\right)\Theta_{y,\epsilon}\left(x\right)\,,\quad\textrm{for }x\geq0\,,\label{eq:func_mon}
\end{equation}
where 
\begin{equation}
g_{\alpha,\delta,y}\left(x\right):=\left(\frac{x}{y}\left(1-\delta\right)+\delta\right)^{\alpha}\,,
\end{equation}
and $\Theta_{y,\epsilon}$ is a smooth step function with the condition
\begin{equation}
\Theta_{y,\epsilon}\left(x\right)=\begin{cases}
1 & x\leq y\,,\\
0 & x\geq y+\epsilon\,.
\end{cases}
\end{equation}
We introduce such a step function to ensure that $f_{\alpha,\delta,y,\epsilon}\in\mathcal{S}\left(\mathbb{R}\right)$
for any values of $\left(\alpha,\delta,y,\epsilon\right)$ in the
set specified above. The functions $f_{\alpha,\delta,y,\epsilon}$
are smoothly extended to the whole real line. In particular we use
\begin{equation}
\Theta_{y,\epsilon}\left(x\right):=\left[1+\exp\left(-\frac{2\epsilon\left(x-y-\frac{\epsilon}{2}\right)}{\left(x-y-\frac{\epsilon}{2}\right)^{2}-\frac{\epsilon^{2}}{4}}\right)\right]^{-1}\,,\quad\textrm{if }y<x<y+\epsilon\,,
\end{equation}
which has the useful property $\max_{x\in\mathbb{R}}\left|\Theta_{y,\epsilon}'\left(x\right)\right|=\frac{2}{\epsilon}$.
From now on, we do not write the cumbersome subindices of the above
functions. For the different terms of \eqref{eq: ineq_mon} we have
that 
\begin{eqnarray}
f\left(y\right)^{2}-f\left(0\right)^{2}\!\!\!\! & = & \!\!\!\!1-\delta^{2\alpha}\,,\\
\pi\int_{0}^{y}dx\,xf'\left(x\right)^{2}\!\!\!\! & = & \!\!\!\!\pi\frac{\alpha}{2}\left(\frac{2\alpha\delta-\delta^{2\alpha}}{1-2\alpha}+1\right)\,,\\
\pi y\int_{y}^{+\infty}dx\,f'\left(x\right)^{2}\!\!\!\! & \leq & \!\!\!\!\pi y\int_{y}^{y+\epsilon}dx\left|g'\left(x\right)^{2}\Theta\left(x\right)^{2}\right|+\pi y\int_{y}^{y+\epsilon}dx\left|g\left(x\right)^{2}\Theta'\left(x\right)^{2}\right|\\
 &  & \!\!\!\!+\pi y\int_{y}^{y+\epsilon}dx\left|2g'\left(x\right)g\left(x\right)\Theta\left(x\right)\Theta'\left(x\right)\right|\,.\label{eq: term_feo}
\end{eqnarray}
We deal with each term of \eqref{eq: term_feo} separately 
\begin{eqnarray}
\pi y\int_{y}^{y+\epsilon}dx\,\left|g'\left(x\right)^{2}\Theta\left(x\right)^{2}\right|\!\!\!\! & \leq & \!\!\!\!\pi y\int_{y}^{y+\epsilon}dx\,g'\left(x\right)^{2}\nonumber \\
 & = & \!\!\!\!\frac{\pi\alpha^{2}\left(1-\delta\right)}{1-2\alpha}\left[1-\left(1+\frac{(1-\delta)\epsilon}{y}\right)^{2\alpha-1}\right]\underset{\epsilon\rightarrow+\infty}{\longrightarrow}\frac{\pi\alpha^{2}\left(1-\delta\right)}{1-2\alpha}\,,\\
\pi y\int_{y}^{y+\epsilon}dx\,\left|2g'\left(x\right)g\left(x\right)\Theta\left(x\right)\Theta'\left(x\right)\right|\!\!\!\! & \leq & \!\!\!\!\pi y\frac{2}{\epsilon}\int_{y}^{y+\epsilon}dx\,2g'\left(x\right)g\left(x\right)=2\pi\frac{y}{\epsilon}\left[g\left(y+\epsilon\right)^{2}-g\left(y\right)^{2}\right]\nonumber \\
 & = & \!\!\!\!\frac{2\pi y}{\epsilon}\left[\left(1+\frac{\left(1-\delta\right)\epsilon}{y}\right)^{2\alpha}-1\right]\underset{\epsilon\rightarrow+\infty}{\longrightarrow}0\,,
\end{eqnarray}
\begin{eqnarray}
\pi y\int_{y}^{y+\epsilon}dx\,\left|g\left(x\right)^{2}\Theta'\left(x\right)^{2}\right|\!\!\!\! & \leq & \!\!\!\!\pi y\frac{4}{\epsilon^{2}}\int_{y}^{y+\epsilon}dx\,g\left(x\right)^{2}\nonumber \\
 & = & \!\!\!\!\frac{4\pi y^{2}}{\left(1+2\alpha\right)\left(1-\delta\right)\epsilon^{2}}\left(\left(1+\frac{\left(1-\delta\right)\epsilon}{y}\right)^{2\alpha+1}-1\right)\\
 & \underset{\epsilon\rightarrow+\infty}{\longrightarrow} & \!\!\!\!0\,.\nonumber 
\end{eqnarray}
where in the last steps of each computation we take the limit $\epsilon\rightarrow+\infty$.
It is valid to take this limit in the inequality since it must hold
for all $\epsilon>0$. Replacing these partial results on \eqref{eq: ineq_mon}
we arrive at
\begin{equation}
\lambda\left(1-\delta^{2\alpha}\right)\leq\pi\frac{\alpha}{2}\left(\frac{2\alpha\delta-\delta^{2\alpha}}{1-2\alpha}+1\right)+\frac{\pi\alpha^{2}\left(1-\delta\right)}{1-2\alpha}\,.
\end{equation}
Then, taking the limit $\delta\rightarrow0^{+}$ we get

\begin{equation}
\lambda\leq\pi\frac{\alpha}{2}+\pi\frac{\alpha^{2}}{1-2\alpha}\,,
\end{equation}
and finally, taking $\alpha\rightarrow0^{+}$ we arrive at the desired
result 
\begin{equation}
\lambda\leq0.\label{eq: upper_bound}
\end{equation}

\section{Algebraic theory of the free hermitian scalar field\label{sec:Wightman-theory-of}}

\subsection{Axioms of AQFT\label{subsec:Axioms-of-AQFT}}

In the algebraic approach to quantum field theory (AQFT), we associate
for each region of the spacetime a $C^{*}$ or von Neumann algebra
which encodes the algebraic relations between the quantum fields.
Such an assignment must satisfy a set of axioms that encode the physical
conditions in the algebraic framework. Unless the specific set of
axioms considered could depend on the underlying theory (especially
on the spacetime considered), the assumptions listed below are very
standard for the treatment of QFT's on Minkowski spacetime.

To start we call a double cone to any open region $\mathcal{O}\subset\mathbb{R}^{d}$
of Minkowski spacetime defined by the intersection of the future open
null cone of some point $x\in\mathbb{R}^{d}$ with the past open null
cone of other point $y\in\mathbb{R}^{d}$.\footnote{In particular, if $y$ is not in the timelike future of $x$, then
$\mathcal{O}=\emptyset$.} In AQFT, we start with a $C^{*}$-algebra $\mathfrak{A}$, called
the \textit{quasilocal algebra}, and we assign to each (nonempty)
double cone $\mathcal{O}\subset\mathbb{R}^{d}$ a $C^{*}$-subalgebra
$\mathfrak{A}\left(\mathcal{O}\right)\subset\mathfrak{A}$, which
are called the \textit{local algebra}s. This collection (net\footnote{Mathematically, due to axiom 1, the collection of local algebras forms
a net indexed by the set of double cones. The set of double cones
forms a direct set when it is ordered by the usual set inclusion.}) of local algebras must satisfy the following: 
\begin{enumerate}
\item Generating property: $\mathfrak{A=\overline{\bigcup_{\mathcal{O}}\mathfrak{A}\left(\mathcal{O}\right)}}^{\left\Vert .\right\Vert }$
, where the union runs over the set of all double cones.
\item Isotony: for any pair of double cones $\mathcal{O}_{1}\subset\mathcal{O}_{2}$,
then $\mathfrak{A}\left(\mathcal{O}_{1}\right)\subset\mathfrak{A}\left(\mathcal{O}_{2}\right)$. 
\item Causality: if $\mathcal{O}_{1}$ and $\mathcal{O}_{2}$ are spacelike
separated (i.e. $\mathcal{O}_{1}\sim\mathcal{O}_{2}$) then $\left[\mathfrak{A}\left(\mathcal{O}_{1}\right),\mathfrak{A}\left(\mathcal{O}_{2}\right)\right]=\left\{ 0\right\} $. 
\item Poincaré covariance: there is a (norm) continuous linear representation
$\alpha_{g}$ of $\mathcal{P}_{+}^{\uparrow}$ in $\mathfrak{U}$,
such that $\alpha_{g}\left(\mathfrak{A}\left(\mathcal{O}\right)\right)=\mathfrak{A}\left(g\mathcal{O}\right)$
for any open bounded region $\mathcal{O}$ and all $g\in\mathcal{P}_{+}^{\uparrow}$,
where the action of $g\in\mathcal{P}_{+}^{\uparrow}$ over a region
$\mathcal{O}$ is given by $g\mathcal{O}:=\left\{ \Lambda x+a\,:\,x\in\mathcal{O}\right\} $.
\item Vacuum: there is a pure state $\omega$ in $\mathfrak{A}$ invariant
under all $\alpha_{g}$. Then, in its GNS representation $\left(\pi,\mathcal{H},\Omega\right)$
the linear representation $\alpha_{g}$ is implemented by a positive
energy unitary representation of $\mathcal{P}_{+}^{\uparrow}$ in
$\mathcal{H}$ in the sense that $U\left(g\right)\pi\left(A\right)U\left(g\right)^{*}=\pi\left(\alpha_{g}\left(A\right)\right)$
for all $A\in\mathfrak{A}$ and all $g\in\mathcal{P}_{+}^{\uparrow}$.
Positive energy means that the representation is strongly continuous
and the infinitesimal generators $P^{\mu}$ of the translation subgroup
(i.e. $U\left(0,a\right)=\mathrm{e}^{iP^{\mu}a_{\mu}}$) have their
spectral projections on the \textit{closed forward light cone} $\overline{V}_{+}:=\left\{ p\in\mathbb{R}^{d}\,:\,p\cdot p>0\textrm{ and }p^{0}>0\right\} $. 
\end{enumerate}
For a general open region (possibly unbounded) $\mathcal{\mathcal{O}}\subset\mathbb{R}^{d}$,
we define $\mathfrak{A}\left(\mathcal{O}\right):=\overline{\bigcup_{\mathcal{\tilde{O}}\subset\mathcal{O}}\mathfrak{A}\left(\tilde{\mathcal{O}}\right)}^{\left\Vert .\right\Vert }$
where the union runs over the set of all double cones $\mathcal{\tilde{O}}\subset\mathcal{O}$.

When we want to study states which are constructed by local perturbations
around the vacuum state $\omega$, we often work directly by the collection
of concrete $C^{*}$-algebras $\pi\left(\mathfrak{A}\left(\mathcal{O}\right)\right)\subset\mathcal{B}\left(\mathcal{H}\right)$
acting on the vacuum Hilbert space $\mathcal{H}$. For technical reasons,
we usually work with the net of von Neumann algebras $\mathcal{R}\left(\mathcal{O}\right):=\pi\left(\mathfrak{A}\left(\mathcal{O}\right)\right)''$,
where $''$ denotes the double commutant which coincides with the
weak closure. Moreover, when we want to construct a concrete example
of a QFT satisfying the axioms above, it is, in general, easier to
construct a net of von Neumann algebras $\mathcal{O}\rightarrow\mathcal{R}\left(\mathcal{O}\right)$
acting on a given Hilbert space.

One immediate consequence of the axioms is the Reeh-Schlieder theorem.
Before we state it we need to introduce some definitions. For any
open region $\mathcal{O}\subset\mathbb{R}^{d}$, we define its (open)
\textit{spacelike complement} as
\begin{equation}
\mathcal{O}':=\mathrm{Int}\left\{ x\in\mathbb{R}^{d}\::\:\left(x-y\right)^{2}<0,\;\forall y\in\mathcal{O}\right\} \,.
\end{equation}
Let $\mathcal{R}\subset\mathcal{B}\left(\mathcal{H}\right)$ be a
von Neumann algebra. We say that a vector $\Phi\in\mathcal{H}$ is
cyclic iff $\overline{\mathcal{R}\Phi}=\mathcal{H}$, and separating
iff $A\Phi=0$ with $A\in\mathcal{R}$ implies $A=0$.
\begin{thm}
\label{par:Theorem_rs}(Reeh-Schlieder \cite{reeh})\textup{ In any
QFT satisfying the axioms 1. to 5. above, the vacuum vector $\Omega$
is cyclic for any algebra $\pi\left(\mathfrak{A}\left(\mathcal{O}\right)\right)''$
corresponding to any (non-empty) open region. Moreover, if $\mathcal{O}'$
is also open and non-empty, then $\Omega$ is also separating for
$\pi\left(\mathfrak{A}\left(\mathcal{O}\right)\right)''$. }
\end{thm}
In the following subsections, we concretely define the net of
algebras associated to a free Hermitian scalar field satisfying the
axioms listed above.

\subsection{Local algebras for spacetime regions\label{subsec:algebras-st}}

The algebraic theory of the real scalar field is defined as a net
of von Neumann algebras acting in the Fock Hilbert space $\mathcal{H}$.
This space is constructed as the (symmetric) tensor product of the
one-particle Hilbert space. To describe it properly, we introduce
the following three vector spaces.

\paragraph*{The space of test functions.}

The space of test functions is the Schwartz space $\mathcal{S}\left(\mathbb{R}^{d},\mathbb{R}\right)$
of real, smooth and exponentially decreasing functions at infinity.
This space carries naturally a representation of the restricted Poincaré
group $\mathcal{P}_{+}^{\uparrow}$ given by $f\mapsto f_{\left(\Lambda,a\right)}$,
with $f_{\left(\Lambda,a\right)}\left(x\right):=f\left(\Lambda\left(x-a\right)\right)$
for any $\left(\Lambda,a\right)\in\mathcal{P}_{+}^{\uparrow}$.

\paragraph*{The one-particle Hilbert space.}

The Hilbert space $\mathfrak{H}$ of one particle states of mass $m>0$
and zero spin is made up of the square-integrable functions on the
mass shell hyperboloid $H_{m}:=\left\{ p\in\mathbb{R}^{d}\,:\,p^{2}=m^{2}\:,\,p^{0}>0\right\} $
with the Poincaré invariant measure $d\mu(p):=\varTheta\left(p^{0}\right)\delta\left(p^{2}-m^{2}\right)d^{d}p$.
It can be realized as 
\begin{eqnarray}
\mathfrak{H}\!\!\!\! & = & \!\!\!\!L^{2}\left(\mathbb{R}^{d-1},\frac{d^{d-1}p}{2\omega\left(\bar{p}\right)}\right)\:\textrm{,}\\
\left\langle \boldsymbol{f},\boldsymbol{g}\right\rangle _{\mathfrak{H}}\!\!\!\! & = & \!\!\!\!\int_{\mathbb{R}^{d-1}}\frac{d^{d-1}p}{2\omega\left(\bar{p}\right)}\boldsymbol{f}\left(\bar{p}\right)^{*}\boldsymbol{g}\left(\bar{p}\right)\,,
\end{eqnarray}
where $p^{0}=\sqrt{\bar{p}^{2}+m^{2}}=:\omega\left(\bar{p}\right)$.
Such a space carries a unitary representation of $\mathcal{P}_{+}^{\uparrow}$
given by $\left(u\left(\Lambda,a\right)\boldsymbol{f}\right)\left(p\right)=\mathrm{e}^{ip\cdot a}\boldsymbol{f}\left(\Lambda^{-1}p\right)$
for any $\boldsymbol{f}\in\mathfrak{H}$ and $\left(\Lambda,a\right)\in\mathcal{P}_{+}^{\uparrow}$.

\paragraph*{The Fock Hilbert space.}

The Fock Hilbert space $\mathcal{H}$ is the direct sum of the symmetric
tensor powers of the one particle Hilbert space $\mathfrak{H}$, i.e.
\begin{equation}
\mathcal{H}=\bigoplus_{n=0}^{\infty}\mathfrak{H}^{\otimes n,sym}\,.
\end{equation}
For each $\boldsymbol{h}\in\mathfrak{H}$, the creation and annihilation
operators $A^{*}\left(\boldsymbol{h}\right)$ and $A\left(\boldsymbol{h}\right)$
act over $\mathcal{H}$ as usual. The Fock space naturally inherits
from $\mathfrak{H}$ a unitary representation of $\mathcal{P}_{+}^{\uparrow}$
which is denoted by $U\left(\Lambda,a\right)$. According to that
there is a unique (up to a phase) $\mathcal{P}_{+}^{\uparrow}$-invariant
vector denoted by $\Omega=\boldsymbol{1}\in\mathfrak{H}^{\otimes0}$,
which is called the \textit{vacuum vector.} For each $\boldsymbol{h}\in\mathfrak{H}$,
the normalized vector 
\[
\mathrm{e}^{\boldsymbol{h}}:=\mathrm{e}^{-\frac{1}{2}\left\Vert \boldsymbol{h}\right\Vert _{\mathfrak{H}}^{2}}\sum_{n=0}^{\infty}\frac{\boldsymbol{h}^{\otimes n}}{\sqrt{n!}}\in\mathcal{H}\,,
\]
is called \textit{coherent vector}, and it satisfies the relations
$\mathrm{e}^{\boldsymbol{0}}=\Omega$ and $\left\langle \Omega,\mathrm{e}^{\boldsymbol{h}}\right\rangle _{\mathcal{H}}=\mathrm{e}^{-\frac{1}{2}\left\Vert \boldsymbol{h}\right\Vert _{\mathfrak{H}}^{2}}$.

It is a very well-known fact that the structure of a free QFT is completely
determined by the underlying one-particle quantum theory. More concretely,
the assignment $\mathcal{O}\rightarrow\mathcal{R}\left(\mathcal{O}\right)$
is determined by the composition of two different maps 
\begin{eqnarray}
\mathcal{O}\subset\mathbb{R}^{d} & \longrightarrow & K\left(\mathcal{O}\right)\subset\mathfrak{H}\,,\\
K\subset\mathfrak{H} & \longrightarrow & \mathcal{R}\left(K\right)\subset\mathcal{B}\left(\mathcal{H}\right)\,,
\end{eqnarray}
which are called 1st and 2nd quantization maps respectively. We treat each map separately.

\subsubsection{First quantization map}

Given any open region $\mathcal{O}\subset\mathbb{R}^{d}$, we remember
that $\mathcal{O}'$ denotes its spacelike complement, and then we
define its \textit{causal completation} as $\mathcal{O}''$.\footnote{It is always true that $\mathcal{O}\subset\mathcal{O}''$.}
A region $\mathcal{O}\subset\mathbb{R}^{d}$ is called \textit{causally
complete} iff $\mathcal{O}\equiv\mathcal{O}''$. In particular, any
double cone is causally complete. From now on, we work with causally
complete regions.

For any closed linear subspace $K\subset\mathfrak{H}$ we define its
the \textit{symplectic complemen}t as
\begin{equation}
K':=\left\{ h\in\mathfrak{H}\,:\,\mathrm{Im}\left\langle h,k\right\rangle _{\mathfrak{H}}=0\,,\textrm{ for all }k\in K\right\} \,.
\end{equation}
 Now, we consider the following real dense embedding $E:\,\mathcal{S}\left(\mathbb{R}^{d},\mathbb{R}\right)\rightarrow\mathfrak{H}$
\begin{equation}
\left(Ef\right)\left(\bar{p}\right):=\sqrt{2\pi}\,\hat{f}\mid_{H_{m}}\left(\bar{p}\right)=\sqrt{2\pi}\,\hat{f}\left(\omega\left(\bar{p}\right),\bar{p}\right)\:\textrm{,}
\end{equation}
where $\hat{f}\left(p\right):=\left(2\pi\right)^{-\frac{d}{2}}\int_{\mathbb{R}^{d}}f\left(x\right)\mathrm{e}^{ip\cdot x}d^{d}x$
is the usual Fourier transform. Such embedding is Poincaré covariant,
i.e. $E\left(f_{\left(\Lambda,a\right)}\right)=u\left(\Lambda,a\right)E\left(f\right)$.
From now on, we naturally identified functions on $\mathcal{S}\left(\mathbb{R}^{d},\mathbb{R}\right)$
with vectors on $\mathfrak{H}$ through the above embedding.

The \textit{1st quantization map} is assignment $\mathcal{O}\subset\mathbb{R}^{d}\rightarrow K\left(\mathcal{O}\right)\subset\mathfrak{H}$,
where $K\left(\mathcal{O}\right)$ is a real closed linear subspace.
It is defined by 
\begin{eqnarray}
\mathcal{O}\subset\mathbb{R}^{d} & \longrightarrow & K\left(\mathcal{O}\right):=\overline{\left\{ E\left(f\right)\,:\,f\in\mathcal{S}\left(\mathbb{R}^{d},\mathbb{R}\right)\textrm{ and }supp\left(f\right)\subset\mathcal{O}\right\} }\subset\mathfrak{H}\,.
\end{eqnarray}
It is not difficult to see that this satisfies the duality $K\left(\mathcal{O}'\right)=K\left(\mathcal{O}\right)'$.

\subsubsection{Second quantization map}

We define the embedding $W:\mathfrak{H}\rightarrow\mathcal{B}\left(\mathcal{H}\right)$
\begin{equation}
W\left(\boldsymbol{h}\right):=\mathrm{e}^{i\left(A\left(\boldsymbol{h}\right)+A^{*}\left(\boldsymbol{h}\right)\right)}\:.\label{weyl_map}
\end{equation}
The operators $W\left(\boldsymbol{h}\right)$ are called \textit{Weyl
unitaries}. These operators satisfy the canonical commutation relations
(CCR) in Segal's form \cite{araki63} 
\begin{eqnarray}
W\left(\boldsymbol{h}_{1}\right)W\left(\boldsymbol{h}_{2}\right)\!\!\!\! & = & \!\!\!\!\mathrm{e}^{-i\,\mathrm{Im}\left\langle \boldsymbol{h}_{1},\boldsymbol{h}_{2}\right\rangle _{\mathfrak{H}}}W\left(\boldsymbol{h}_{1}+\boldsymbol{h}_{2}\right)\,,\\
W\left(\boldsymbol{h}\right)^{*}\!\!\!\! & = & \!\!\!\!W\left(-\boldsymbol{h}\right)\:.
\end{eqnarray}
A Poincaré unitary $U\left(\Lambda,a\right)$ acts covariantly on
a Weyl operator according to 
\begin{eqnarray}
U\left(\Lambda,a\right)W\left(\boldsymbol{h}\right)U\left(\Lambda,a\right)^{*}\!\!\!\! & = & \!\!\!\!W\left(u\left(\Lambda,a\right)\boldsymbol{h}\right)\,,\\
W\left(\boldsymbol{h}\right)\Omega\!\!\!\! & = & \!\!\!\!\mathrm{e}^{i\boldsymbol{h}}\:\textrm{.}
\end{eqnarray}

The \textit{2nd quantization map} is an assignment $K\subset\mathfrak{H}\rightarrow\mathcal{R}\left(K\right)\subset\mathcal{B}\left(\mathcal{H}\right)$
, from the set of real closed linear subspace of $\mathfrak{H}$ to
the set of von Neumann subalgebras of $\mathcal{B}\left(\mathcal{H}\right)$.
It is defined as 
\begin{eqnarray}
K\in\mathfrak{H} & \longrightarrow & \mathcal{R}\left(K\right):=\left\{ W\left(\boldsymbol{k}\right)\,:\,\boldsymbol{k}\in K\right\} ''\subset\mathcal{B}\left(\mathcal{H}\right)\:.
\end{eqnarray}
This map satisfies the duality $\mathcal{R}\left(K'\right)=\mathcal{R}\left(K\right)'$.

\subsubsection{Net of local algebras}

According to the above discussion, the net of local algebras $\mathcal{O}\subset\mathbb{R}^{d}\rightarrow\mathcal{R}\left(\mathcal{O}\right)\subset\mathcal{B}\left(\mathcal{H}\right)$
of the free Hermitian scalar field is defined as the composition of
the 1st and 2nd quantization maps, i.e. 
\begin{equation}
\mathcal{R}\left(\mathcal{O}\right):=\mathcal{R}\left(K\left(\mathcal{O}\right)\right)\:\textrm{.}\label{eq:net_spacetime}
\end{equation}
This net satisfies all the axioms listed above, including the Haag
duality. For $f\in\mathcal{S}\left(\mathbb{R}^{d},\mathbb{R}\right)$,
the field operator $\phi\left(f\right)$ is defined through the relation

\[
W\left(E\left(f\right)\right)=\mathrm{e}^{i\phi\left(f\right)}=:W\left(f\right)\,.
\]

\subsection{Local algebras at a fixed time\label{subsec:algebras-ft}}

In this subsection, we discuss the theory of the von Neumann algebras
for the real scalar free field at a fixed time $x^{0}=0$. Naively
speaking, they are the local algebras generated by the field operator
at a fixed time $\varphi\left(\bar{x}\right)$ and its canonical conjugate
momentum field $\pi\left(\bar{x}\right)$. This theory is very
useful for the computation of the relative entropy in section \ref{sec:Explicit-calculation}.

We can decompose $\mathfrak{H=\mathfrak{H}_{\varphi}\oplus_{\mathbb{R}}\mathfrak{H}_{\pi}}$
into two $\mathbb{R}$-linear closed subspaces 
\begin{eqnarray}
\mathfrak{H}_{\varphi}\!\!\!\! & := & \!\!\!\!\left\{ \boldsymbol{h}\in\mathfrak{H}\,:\,\boldsymbol{h}\left(\bar{p}\right)=\boldsymbol{h}\left(-\bar{p}\right)^{*}\:\textrm{a.e.}\right\} \,,\\
\mathfrak{H}_{\pi}\!\!\!\! & := & \!\!\!\!\left\{ \boldsymbol{h}\in\mathfrak{H}\,:\,\boldsymbol{h}\left(\bar{p}\right)=-\boldsymbol{h}\left(-\bar{p}\right)^{*}\:\textrm{a.e.}\right\} \,.
\end{eqnarray}
 Each $\boldsymbol{h}\in\mathfrak{H}$ can be uniquely written as
$\boldsymbol{h}=\boldsymbol{h}_{\varphi}+\boldsymbol{h}_{\pi}$, where
\begin{eqnarray}
\boldsymbol{h}_{\varphi}\left(\bar{p}\right)=\frac{\boldsymbol{h}\left(\bar{p}\right)+\boldsymbol{h}\left(-\bar{p}\right)^{*}}{2} & \textrm{ and } & \boldsymbol{h}_{\pi}\left(\bar{p}\right)=\frac{\boldsymbol{h}\left(\bar{p}\right)-\boldsymbol{h}\left(-\bar{p}\right)^{*}}{2}\,.
\end{eqnarray}
We also have the useful relations
\begin{equation}
\mathrm{Im}\left\langle \boldsymbol{h}_{\varphi},\boldsymbol{h}'_{\varphi}\right\rangle =\mathrm{Im}\left\langle \boldsymbol{h}_{\pi},\boldsymbol{h}'_{\pi}\right\rangle =\mathrm{Re}\left\langle \boldsymbol{h}_{\varphi},\boldsymbol{h}'_{\pi}\right\rangle =0\:\textrm{,}\label{eq:rel_prod_int}
\end{equation}
for all $\boldsymbol{h}_{\varphi},\,\boldsymbol{h}'_{\varphi}\in\mathfrak{H}_{\varphi}$
and $\boldsymbol{h}_{\pi},\,\boldsymbol{h}'_{\pi}\in\mathfrak{H}_{\pi}$. 

Now, we consider the following real dense embeddings $E_{\varphi,\pi}:\,\mathcal{S}\left(\mathbb{R}^{d-1},\mathbb{R}\right)\rightarrow\mathfrak{H}_{\varphi,\pi}$,
\begin{equation}
\left(E_{\varphi}f\right)\left(\bar{p}\right):=\hat{f}\left(\bar{p}\right)\quad\mathrm{and}\quad\left(E_{\pi}f\right)\left(\bar{p}\right):=i\omega\left(\bar{p}\right)\,\hat{f}\left(\bar{p}\right)\,,
\end{equation}
where $\hat{f}\left(\bar{p}\right):=\left(2\pi\right)^{-\frac{d-1}{2}}\int_{\mathbb{R}^{d-1}}f\left(\bar{x}\right)\mathrm{e}^{-i\bar{p}\cdot\bar{x}}d^{d-1}x$.
From now on, we naturally identify functions on $\mathcal{S}\left(\mathbb{R}^{d-1},\mathbb{R}\right)$
with vectors on $\mathfrak{H}_{\varphi},\mathfrak{H}_{\pi}$ through
these embeddings. The map $E_{\varphi}$ (respectively, $E_{\pi}$) is actually
defined on a bigger class of test functions, namely $H^{-\frac{1}{2}}\left(\mathbb{R}^{d-1},\mathbb{R}\right)$
(respectively, $H^{\frac{1}{2}}\left(\mathbb{R}^{d-1},\mathbb{R}\right)$),
i.e.
\begin{equation}
E_{\varphi}:\,H^{-\frac{1}{2}}\left(\mathbb{R}^{d-1},\mathbb{R}\right)\rightarrow\mathfrak{H}_{\varphi}\quad\mathrm{and}\quad E_{\pi}:\,H^{\frac{1}{2}}\left(\mathbb{R}^{d-1},\mathbb{R}\right)\rightarrow\mathfrak{H}_{\pi}\,,
\end{equation}
where $H^{\alpha}\left(\mathbb{R}^{d-1},\mathbb{R}\right)$ is the
real Sobolev space of order $\alpha$ (see Appendix \ref{subsec:appendix_Sobolev}).
We have that $E_{\varphi}\left(H^{\frac{1}{2}}\left(\mathbb{R}^{d-1},\mathbb{R}\right)\right)=\mathfrak{H}_{\varphi}$
and $E_{\pi}\left(H^{-\frac{1}{2}}\left(\mathbb{R}^{d-1},\mathbb{R}\right)\right)=\mathfrak{H}_{\pi}$.
For each $\boldsymbol{h}_{\varphi}\in\mathfrak{H}_{\varphi}$ and
$\boldsymbol{h}_{\pi}\in\mathfrak{H}_{\pi}$ and using the map \eqref{weyl_map},
we define the Weyl unitaries
\begin{equation}
W_{\varphi}\left(\boldsymbol{h}_{\varphi}\right):=W\left(\boldsymbol{h}_{\varphi}\right)\quad\textrm{and}\quad W_{\pi}\left(\boldsymbol{h}_{\pi}\right):=W\left(\boldsymbol{h}_{\pi}\right)\,,
\end{equation}
which satisfy the CCR in the Weyl form \cite{araki63}
\begin{eqnarray}
W_{\varphi}\left(\boldsymbol{h}_{\varphi}\right)W_{\pi}\left(\boldsymbol{h}_{\pi}\right)W_{\varphi}\left(\boldsymbol{h}'_{\varphi}\right)W_{\pi}\left(\boldsymbol{h}'_{\pi}\right)\!\!\!\! & = & \!\!\!\!W_{\varphi}\left(\boldsymbol{h}_{\varphi}+\boldsymbol{h}'_{\varphi}\right)W_{\pi}\left(\boldsymbol{h}_{\pi}+\boldsymbol{h}'_{\pi}\right)\mathrm{e}^{2i\,\mathrm{Im}\left\langle \boldsymbol{h}'_{\varphi},\boldsymbol{h}_{\pi}\right\rangle _{\mathfrak{H}}}\\
W_{\varphi}\left(\boldsymbol{h}_{\varphi}\right)^{*}\!\!\!\! & = & \!\!\!\!W_{\varphi}\left(-\boldsymbol{h}_{\varphi}\right)\\
W_{\pi}\left(\boldsymbol{h}_{\pi}\right)^{*}\!\!\!\! & = & \!\!\!\!W_{\pi}\left(-\boldsymbol{h}_{\pi}\right)\:\textrm{,}
\end{eqnarray}
The \textit{field operator at a fixed time} $\varphi\left(f_{\varphi}\right)$
and its \textit{canonical conjugate momentum field} $\pi\left(f_{\pi}\right)$
are defined through the formulas 
\begin{equation}
W_{\varphi}\left(E_{\varphi}\left(f_{\varphi}\right)\right)=\mathrm{e}^{i\varphi\left(f_{\varphi}\right)}=:W_{\varphi}\left(f_{\varphi}\right)\quad\mathrm{and}\quad W_{\pi}\left(E_{\pi}\left(f_{\pi}\right)\right)=:\mathrm{e}^{i\pi\left(f_{\pi}\right)}=:W_{\pi}\left(f_{\pi}\right)\,.
\end{equation}
Here again, the local algebras at a fixed time are also defined through
the 1st and 2nd quantization map.

\paragraph{First quantization map.}

We say that $\mathcal{C}\subset\mathbb{R}^{d-1}$ is a \textit{spatially
complete} region iff it is open, regular\footnote{An open set $U\subset\mathbb{R}^{n}$ is regular iff $U\equiv Int\left(\overline{U}\right)$.}
and with regular boundary.\footnote{The boundary $\partial\mathcal{C}\subset\mathbb{R}^{d-1}$ is a smooth
submanifold of dimension $d-2$, or several manifolds joined together
along smooth manifolds of dimension $d-3$ \cite{araki64}.} Here we work with this kind of regions. Given any such region $\mathcal{C}\subset\mathbb{R}^{d-1}$,
we define its (open) \textit{space complement} as $\mathcal{C}':=\mathbb{R}^{d-1}-\overline{\mathcal{C}}$
. 

Then the \textit{1st quantization map} is defined as
\begin{eqnarray}
\mathcal{C}\subset\mathbb{R}^{d-1} & \rightarrow & K_{\varphi}\left(\mathcal{C}\right):=\overline{\left\{ E_{\varphi}\left(f\right)\,:\,f\in\mathcal{S}\left(\mathbb{R}^{d-1},\mathbb{R}\right)\textrm{ and }supp\left(f\right)\subset\mathcal{C}\right\} }\subset\mathfrak{H}_{\varphi}\,,\\
\mathcal{C}\subset\mathbb{R}^{d-1} & \rightarrow & K_{\pi}\left(\mathcal{C}\right):=\overline{\left\{ E_{\pi}\left(f\right)\,:\,f\in\mathcal{S}_{\mathbb{}}\left(\mathbb{R}^{d-1},\mathbb{R}\right)\textrm{ and }supp\left(f\right)\subset\mathcal{C}\right\} }\:\subset\mathfrak{H}_{\pi}\textrm{.}
\end{eqnarray}
It can be shown that
\begin{eqnarray}
K_{\varphi}\left(\mathcal{C}\right)\!\!\!\! & = & \!\!\!\!\left\{ E_{\varphi}\left(f\right)\,:\,f\in H^{-\frac{1}{2}}\left(\mathbb{R}^{d-1},\mathbb{R}\right)\textrm{ and }supp\left(f\right)\subset\mathcal{\overline{C}}\:\textrm{a.e.}\right\} \,,\\
K_{\pi}\left(\mathcal{C}\right)\!\!\!\! & = & \!\!\!\!\left\{ E_{\pi}\left(f\right)\,:\,f\in H^{\frac{1}{2}}\left(\mathbb{R}^{d-1},\mathbb{R}\right)\textrm{ and }supp\left(f\right)\subset\mathcal{\mathcal{\overline{C}}}\:\textrm{a.e.}\right\} \:\textrm{.}
\end{eqnarray}

\paragraph{Second quantization map.}

For each pair $K_{\varphi}\subset\mathfrak{H}_{\varphi}$ and $K_{\pi}\subset\mathfrak{H}_{\pi}$
of $\mathbb{R}$-linear closed subspaces, we define the von Neumann
algebra
\begin{eqnarray}
\left(K_{\varphi},K_{\pi}\right) & \rightarrow & \mathcal{R}_{0}\left(K_{\varphi},K_{\pi}\right):=\left\{ W_{\varphi}\left(k_{\varphi}\right)W_{\pi}\left(k_{\pi}\right)\,:\,k_{\varphi}\in K_{\varphi},\,k_{\pi}\in K_{\pi}\right\} ''\subset\mathcal{B}\left(\mathcal{H}\right)\:\textrm{.}
\end{eqnarray}

\paragraph{Net of local algebras at a fixed time. }

The net of local algebras $\mathcal{C}\subset\mathbb{R}^{d-1}\rightarrow\mathcal{R}_{0}\left(\mathcal{C}\right)\subset\mathcal{B}\left(\mathcal{H}\right)$
of the free Hermitian scalar field at a fixed time is then defined
as the composition of the 1st and 2nd quantization maps, i.e.
\begin{equation}
\mathcal{R}_{0}\left(\mathcal{C}\right):=\mathcal{R}_{0}\left(K_{\varphi}\left(\mathcal{C}\right),K_{\pi}\left(\mathcal{C}\right)\right)\:\textrm{.}\label{eq:net_fixed_time}
\end{equation}
The above net satisfies the following expected properties \cite{araki64}:
\begin{align}
 & \mathcal{R}_{0}\left(\mathcal{C}_{1}\right)\subset\mathcal{R}_{0}\left(\mathcal{C}_{2}\right)\quad\textrm{if }\mathcal{C}_{1}\subset\mathcal{C}_{2}\,,\label{eq:axioms_ft_1st}\\
 & \mathcal{R}_{0}\left(\mathcal{C}_{1}\cup\mathcal{C}_{2}\right)=\mathcal{R}_{0}\left(\mathcal{C}_{1}\right)\lor\mathcal{R}_{0}\left(\mathcal{C}_{2}\right)\,,\\
 & \mathcal{R}_{0}\left(\mathcal{C}'\right)=\mathcal{R}_{0}\left(\mathcal{C}\right)'\,,\\
 & \mathcal{R}_{0}\left(\mathbb{R}^{d-1}\right)=\mathcal{B}\left(\mathcal{H}\right)\,,\label{eq:axioms_ft_last}
\end{align}
where the $\mathcal{R}_{1}\lor\mathcal{R}_{2}$ is the von Neumann
algebra generated by $\mathcal{R}_{1}\cup\mathcal{R}_{2}$ .

\subsection{Relation between the two approaches\label{subsec:Relation-between-algebras}}

In this subsection, we explain the relation existing between the two
approaches of sections \ref{subsec:algebras-st} and \ref{subsec:algebras-ft}.

\paragraph{The relation between nets.}

Given any spatially complete region $\mathcal{C}\subset\mathbb{R}^{d-1}$,
we define its \textit{domain of dependence} $\mathcal{O_{C}}\subset\mathbb{R}^{d}$
as 
\begin{equation}
\mathcal{O_{C}}:=\left\{ x\in\mathbb{R}^{d}\::\:\left(x-\left(0,\bar{y}\right)\right)^{2}<0\textrm{ for all }\bar{y}\in\mathcal{C}'\right\} \,.\label{eq:double_light_cone}
\end{equation}
Then the following relation holds \cite{Camassa},
\begin{eqnarray}
K\left(\mathcal{O_{C}}\right) & = & K_{\varphi}\left(\mathcal{C}\right)\oplus_{\mathbb{R}}K_{\pi}\left(\mathcal{C}\right)\,,
\end{eqnarray}
and hence we have the equality between the von Neumann algebras
\begin{equation}
\mathcal{R}_{0}\left(\mathcal{C}\right)=\mathcal{R}\left(\mathcal{O_{C}}\right)\,.\label{eq:rel_algebras}
\end{equation}
The relations developed along the above subsections can be summarized
in the following schematic diagram
\begin{equation}
\begin{array}{ccccc}
\mathcal{O}_{\mathcal{C}}\subset\mathbb{R}^{d} & \overset{E}{\longrightarrow} & K\subset\mathfrak{H} & \overset{W}{\longrightarrow} & \mathcal{R}\subset\mathcal{B}\left(\mathcal{H}\right)\\
\uparrow &  & \uparrow\oplus_{\mathbb{R}} &  & \shortparallel\qquad\qquad\\
\mathcal{C}\subset\mathbb{R}^{d-1} & \overset{E_{\varphi,\pi}}{\longrightarrow} & \left(K_{\varphi},K_{\pi}\right)\subset\mathfrak{H}_{\varphi}\oplus_{\mathbb{R}}\mathfrak{H}_{\pi} & \overset{W_{\varphi,\pi}}{\longrightarrow} & \;\mathcal{R}_{0}\subset\mathcal{B}\left(\mathcal{H}\right)\,.
\end{array}
\end{equation}

\paragraph{The relation between test functions.}

Given $f\in\mathcal{S}\left(\mathbb{R}^{d},\mathbb{R}\right)$ we
can define
\begin{equation}
F\left(x\right):=\int_{\mathbb{R}^{d}}\Delta\left(x-y\right)\,f\left(y\right)d^{d}y\,,\label{eq:convolucion}
\end{equation}
where $\Delta\left(x\right):=-i\left(2\pi\right)^{-\left(d-1\right)}\int_{\mathbb{R}^{d}}\mathrm{e}^{-ip\cdot x}\delta\left(p^{2}-m^{2}\right)sgn\left(p^{0}\right)\,d^{d}p$.
Indeed $\left(\boxempty+m^{2}\right)F=0$ and we can take its initial
Cauchy data at $x^{0}=0$ through
\begin{equation}
f_{\varphi}\left(\bar{x}\right)=-\frac{\partial F}{\partial x^{0}}\left(0,\bar{x}\right)\;\textrm{ and }\;f_{\pi}\left(\bar{x}\right)=F\left(0,\bar{x}\right)\:\textrm{.}\label{eq:f_fixed_time}
\end{equation}
Finally, it can be shown $f_{\varphi},\,f_{\pi}\in\mathcal{S}\left(\mathbb{R}^{d-1},\mathbb{R}\right)$
and
\begin{equation}
E\left(f\right)=E_{\varphi}\left(f_{\varphi}\right)+E_{\pi}\left(f_{\pi}\right)\:.\label{eq:segal_vs_weyl_1p}
\end{equation}
Moreover, since $F\left(x\right)=0$ if $x\in supp\left(f\right)'$,
then we have $supp\left(f\right)\subset\mathcal{O}_{\mathcal{C}}\Rightarrow supp\left(f_{\varphi}\right),\,supp\left(f_{\pi}\right)\subset\mathcal{C}\:\textrm{.}$

\paragraph{The relation between Weyl unitaries.}

For the particular case of Weyl unitaries, it follows that
\begin{equation}
W\left(f\right)=\mathrm{e}^{i\mathrm{Im}\left\langle f_{\varphi},f_{\pi}\right\rangle _{\mathfrak{H}}}W_{\varphi}\left(f_{\varphi}\right)W_{\pi}\left(f_{\pi}\right)\,,\label{eq:segal_vs_weyl}
\end{equation}
where 
\begin{equation}
\mathrm{Im}\left\langle f_{\varphi},f_{\pi}\right\rangle _{\mathfrak{H}}=\frac{1}{2}\int_{\mathbb{R}^{d-1}}f_{\varphi}\left(\bar{x}\right)\,f_{\pi}\left(\bar{x}\right)d^{d-1}x\,.\label{eq:img_pe}
\end{equation}

\section{Modular theory\label{sec:Modular-theory}}

In this section, we discuss the key points of the modular theory in
the framework of von Neumann (vN) algebras. The main purpose of this
section is to introduce the Araki formula for the relative entropy.
For more details about the content of this section, see for example
\cite{haag,takesaki,Brattelli,bisognano,Araki_entropy}.

\subsection{Modular Hamiltonian and modular flow}
\begin{lem}
\textup{Let $\mathcal{R}\subset\mathcal{B}\left(\mathcal{H}\right)$
be a vN algebra and $\Omega\in\mathcal{H}$ be a cyclic and separating
vector. Then there exists a unique closed antilinear (generally unbounded)
operator $S_{\Omega}$ such that
\begin{equation}
S_{\Omega}\,A\Omega=A^{*}\Omega\,,\quad\forall A\in\mathcal{R}\,.\label{eq:mod_op_def}
\end{equation}
The operator $S_{\Omega}$ is called the }modular involution\textup{
associated to the pair }$\left\{ \mathcal{R},\Omega\right\} $.
\end{lem}
Let $S_{\Omega}=J_{\Omega}\Delta_{\Omega}^{\frac{1}{2}}$ be the polar
decomposition of $S_{\Omega}$. Then, $\Delta_{\Omega}$ (positive
self-adjoint and generally unbounded) is called the \textit{modular
operator} and $J_{\Omega}$ (antiunitary) is called the \textit{modular
conjugation}. Finally, the \textit{modular Hamiltonian} is defined
as 
\begin{equation}
K_{\Omega}:=-\log\left(\Delta_{\Omega}\right)\,,
\end{equation}
and the 1-parameter (strongly continuous) group of unitaries $\Delta_{\Omega}^{it}$
is called the \textit{modular flow}.
\begin{thm}
(Tomita-Takesaki) \textup{Let $\mathcal{R}\subset\mathcal{B}\left(\mathcal{H}\right)$
be a vN algebra, $\Omega\in\mathcal{H}$ be a cyclic and separating vector
and $S_{\Omega}=J_{\Omega}\varDelta_{\Omega}^{\frac{1}{2}}$ be the operator
defined above. The one-parameter (strongly continuous) group of unitaries
$\Delta_{\Omega}^{it}$ is called modular group or modular flow. The
Tomita-Takesaki theorem states that
\begin{equation}
J_{\Omega}\,\mathcal{R}\,J_{\Omega}=\mathcal{R}'
\end{equation}
\begin{equation}
\Delta_{\Omega}^{it}\,\mathcal{R}\,\Delta_{\Omega}^{-it}=\mathcal{R}\quad\mathrm{and}\quad\Delta_{\Omega}^{it}\,\mathcal{R}'\,\Delta_{\Omega}^{-it}=\mathcal{R}'\:\textrm{,}
\end{equation}
for all $t\in\mathbb{R}$. }
\end{thm}
\begin{rem}
In general, the modular flow $\Delta_{\Omega}^{it}$ does not belong
to $\mathcal{R}$ or $\mathcal{R}'$.
\end{rem}
Before we state the Bisognano-Whichmann theorem, we need to introduce
some conventions. Let $\mathcal{W}:=\left\{ x\in\mathbb{R}^{d}\,:\,x^{1}>\left|x^{0}\right|\right\} $
be the \textit{right Rindler wedge} and $\Sigma:=\left\{ \bar{x}\in\mathbb{R}^{d-1}\,:\,x^{1}\geq0\right\} $.
Then by \eqref{eq:double_light_cone} we have that $\mathcal{O}_{\Sigma}=\mathcal{W}$.
From now on, we will denote the orthogonal coordinates to the Rindler
wedge as $\bar{x}_{\bot}:=\left(x^{2},\ldots,x^{d-1}\right)$ and
hence any spacetime point can be expressed as $x=\left(x^{0},x^{1},\bar{x}_{\bot}\right)$.
We also denote the following vN algebras simply as
\begin{align}
\mathcal{R}_{\mathcal{W}} & :=\mathcal{R}\left(\mathcal{W}\right)=\mathcal{R}_{0}\left(\Sigma\right)\,,\\
\mathcal{R}_{\mathcal{W}'} & :=\mathcal{R}\left(\mathcal{W}'\right)=\mathcal{R}_{0}\left(\Sigma'\right)\,,
\end{align}
From relations (\ref{eq:axioms_ft_1st}-\ref{eq:axioms_ft_last}) we have
\begin{equation}
\mathcal{R}'_{\mathcal{W}}=\mathcal{R}{}_{\mathcal{W}'}\quad\mathrm{and}\quad\mathcal{R}{}_{\mathcal{W}}\vee\mathcal{R}_{\mathcal{W}'}=\mathcal{B}\left(\mathcal{H}\right)\,\textrm{.}
\end{equation}
Reeh-Schlieder theorem \ref{par:Theorem_rs} asserts that the vacuum
vector $\Omega$ is cyclic and separating for $\mathcal{R}{}_{\mathcal{W}}$.
\begin{thm}
(Bisognano-Wichmann \cite{bisognano}) \textup{The modular operator
$\Delta_{\Omega}$ and the modular conjugation $J_{\Omega}$ for the
pair $\left\{ \mathcal{R}_{\mathcal{W}},\Omega\right\} $ are
\begin{equation}
J_{\Omega}=\Theta\,U\left(R_{1}\left(\pi\right)\right)\quad\textrm{and}\quad\Delta_{\Omega}=\mathrm{e}^{-2\pi K_{1}}\:\textrm{,}
\end{equation}
where $\Theta$ is the CPT operator, $U\left(R_{1}\left(\pi\right)\right)$
is the Lorentz unitary operator representing a space rotation of angle
$\pi$ along the $x^{1}$ axes and $K_{1}$ is the infinitesimal generator
of the one-parameter group of boost symmetries in the plane $\left(x^{0},x^{1}\right)$,
i.e.
\begin{equation}
U\left(\Lambda_{1}^{s},0\right)=\mathrm{e}^{iK_{1}s}\,,\quad\textrm{with }\Lambda_{1}^{s}:=\left(\begin{array}{ccc}
\cosh\left(s\right) & \sinh\left(s\right) & \boldsymbol{0}\\
\sinh\left(s\right) & \cosh\left(s\right) & \boldsymbol{0}\\
\boldsymbol{0} & \boldsymbol{0} & \boldsymbol{1}
\end{array}\right)\,.\label{eq:boost}
\end{equation}
}
\end{thm}
\begin{rem}
Althpugh we are working with the net of local algebras for the real
scalar field, the above result holds for any relativistic QFT which
satisfies the Wightman axioms.
\end{rem}

\subsection{Relative modular Hamiltonian and relative modular flow\label{subsec:Relative-modular-Hamiltonian}}
\begin{lem}
\textup{Let $\mathcal{R}\subset\mathcal{B}\left(\mathcal{H}\right)$
be a vN algebra and two cyclic and separating vectors $\Omega,\Phi\in\mathcal{H}$.
Then there exists a unique (generally unbounded) closed antilinear
operator such that
\begin{equation}
S_{\Phi,\Omega}\,A\Omega=A^{*}\Phi\,,\quad\forall A\in\mathcal{R}\,.\label{eq:rel_mod}
\end{equation}
The operator $S_{\Omega}$ is called the }relative modular involution\textup{
associated to the pair }$\left\{ \mathcal{R},\Omega,\Phi\right\} $.
\end{lem}
Let $S_{\Phi,\Omega}=J_{\Phi,\Omega}\Delta_{\Phi,\Omega}^{\frac{1}{2}}$
be the polar decomposition of $S_{\Phi,\Omega}$. Then, $\Delta_{\Phi,\Omega}$
is called the \textit{relative modular operator} and $J_{\Phi,\Omega}$
(antiunitary) is called the \textit{relative modular conjugation.}
Then the \textit{relative modular Hamiltonian} is defined
\begin{equation}
K_{\Phi,\Omega}:=-\log\left(\Delta_{\Phi,\Omega}\right)\:\textrm{.}
\end{equation}
The \textit{relative modular flow} $\Delta_{\Phi,\Omega}^{it}$ acts
as the modular flow $\Delta_{\Phi}^{it}$ for the algebra $\mathcal{R}$
and as $\Delta_{\Omega}^{it}$ for the algebra $\mathcal{R}'$, i.e.
\begin{eqnarray}
\Delta_{\Phi,\Omega}^{it}\,A\,\Delta_{\Phi,\Omega}^{-it}=\Delta_{\Phi}^{it}\,A\,\Delta_{\Phi}^{-it} &  & A\in\mathcal{R}\,,\label{eq:rel_mod_flow}\\
\Delta_{\Phi,\Omega}^{it}\,A'\,\Delta_{\Phi,\Omega}^{-it}=\Delta_{\Omega}^{it}\,A'\,\Delta_{\Omega}^{-it} &  & A'\in\mathcal{R}'\,.
\end{eqnarray}
The following theorem summarize the analytics properties of the the
relative modular flow.
\begin{thm}
\label{thm_kms}(KMS condition \cite{Bertozzini})\textup{ Under the
hypothesis of the previous lemma, given any $A,B\in\mathcal{R}$,
there exits a unique continuous function $G_{A,B}:\mathbb{R}+i\left[-1,0\right]\rightarrow\mathbb{C}$,
analytic on $\mathbb{R}+i\left(-1,0\right)$ such that
\begin{eqnarray}
G_{A,B}\left(t\right)\!\!\!\! & = & \!\!\!\!\left\langle \Omega,\Delta_{\Phi,\Omega}^{it}A\Delta_{\Omega}^{-it}B\Omega\right\rangle _{\mathcal{H}}\,,\label{kms_func}\\
G_{A,B}\left(t-i\right)\!\!\!\! & = & \!\!\!\!\left\langle \Phi,B\Delta_{\Phi,\Omega}^{it}A\Delta_{\Omega}^{-it}\Phi\right\rangle _{\mathcal{H}}\,,
\end{eqnarray}
for all $t\in\mathbb{R}$. Moreover, the function above is uniquely
determined by one of its boundary values.}
\end{thm}
As it happens for the modular flow, $\Delta_{\Phi,\Omega}^{it}\notin\mathcal{R}\cup\mathcal{R}'$
in general. However, we can define the one-parameter family of unitaries\footnote{This one-parameter family of operators is not a one-parameter group.}
\begin{equation}
u_{\Phi,\Omega}\left(t\right)=\Delta_{\Phi,\Omega}^{it}\Delta_{\Omega}^{-it}\,,\label{crn}
\end{equation}
whose belong $u_{\Phi,\Omega}\left(t\right)\in\mathcal{R}$ for all
$t\in\mathbb{R}$. This family of unitaries is best known as Connes
Radon-Nikodym cocycle.

\subsection{Araki formula for relative entropy}

The definition of the relative entropy for a general von Neumann algebra
is due to Araki \cite{Araki_entropy}. 
\begin{defn}
Let $\mathcal{R}\subset\mathcal{B}\left(\mathcal{H}\right)$ be a
vN algebra in standard form. For any given two $\omega,\phi$ two
faithful normal states, there exists cyclic and separating vector
representatives $\Omega,\Phi\in\mathcal{H}$.\footnote{In particular, choose them in the natural cone of the standard vector
of $\mathcal{R}$.} Then the relative entropy $S_{R}\left(\phi\mid\omega\right)$ is
defined using the relative modular Hamiltonian $K_{\Omega,\Phi}$
as\footnote{Contrary to the notation employed in sections \ref{sec:Intro} and
\ref{sec:lamda_bounds}, on the l.h.s. expression \eqref{eq:rel_ent}
we emphasize that the relative entropy depends on the states rather
than the vector representatives used to define it. We use this new
notation in the rest of the paper.}
\begin{equation}
S_{R}\left(\phi\mid\omega\right):=\left\langle \Phi,K_{\Omega,\Phi}\Phi\right\rangle _{\mathcal{H}}\,.\label{eq:rel_ent}
\end{equation}
\end{defn}
It can be shown that the above formula is independent of the choice
of the vector representatives for the given states, and it also satisfies
the well-known properties of strict-positivity, monotonicity, convexity
and lower semi-continuity. All these are well discussed in Araki's
original work \cite{Araki_entropy}. When the relative entropy is
finite (in particular, when $\Omega$ belongs to the domain of $K_{\Omega,\Phi}$),
the following useful expression holds
\begin{equation}
S_{R}\left(\phi\mid\omega\right)=i\lim_{t\rightarrow0}\frac{\left\langle \Phi,\Delta_{\Omega,\Phi}^{it}\Phi\right\rangle _{\mathcal{H}}-1}{t}\,.\label{eq:rel_ent_f}
\end{equation}

\section{Relative entropy for coherent states\label{sec:Explicit-calculation}}

In this section, we compute the relative entropy between a coherent
state and the vacuum for the Rindler wedge. Before doing that, we
study some relations concerning the relative entropy which are valid
for any kind of regions. These are explained in the following subsection.

\subsection{Generalities}

Coherent states come from acting with a Weyl operator to the vacuum
vector. Weyl unitaries have the very interesting property that implements,
by adjoint action, automorphism for any local algebra $\mathcal{R}\left(\mathcal{O}\right)$.
Indeed, for any $\boldsymbol{h}\in\mathfrak{H}$ and any Weyl operator
$W\left(f\right)\in\mathcal{R}\left(\mathcal{O}\right)$ ($supp\left(f\right)\subset\mathcal{O}$)
we have that
\begin{equation}
W\left(\boldsymbol{h}\right)^{*}W\left(f\right)W\left(\boldsymbol{h}\right)=\mathrm{e}^{2i\mathrm{Im}\left\langle f,\boldsymbol{h}\right\rangle _{\mathfrak{H}}}W\left(f\right)\in\mathcal{R}\left(\mathcal{O}\right)\,,\label{eq:case_3}
\end{equation}
which implies that $W\left(\boldsymbol{h}\right)^{*}\mathcal{R}\left(\mathcal{O}\right)W\left(\boldsymbol{h}\right)=\mathcal{R}\left(\mathcal{O}\right)$.
This property has an interesting implication for the relative entropy
itself. Indeed, it implies that the relative entropy between a coherent
state and the vacuum is symmetric. In order to justify this property,
we prove the following lemmas.
\begin{lem}
\textup{\label{lema_auto}Let $\mathcal{R}\subset\mathcal{B}\left(\mathcal{H}\right)$
be a vN algebra and $\Omega,\Phi\in\mathcal{H}$ a cyclic and separating
vectors and $U\in\mathcal{B}\left(\mathcal{H}\right)$ unitary such
$U^{*}\mathcal{R}U=\mathcal{R}$. Then,}
\end{lem}
\begin{enumerate}
\item $U\Omega$ and $U\Phi$ are cyclic and separating.
\item $S_{U\Omega}=US_{\Omega}U^{*}$ $\Rightarrow$ $\Delta_{U\Omega}=U\Delta_{\Omega}U^{\text{*}}$.
\item $S_{U\Omega,U\Phi}=US_{\Omega,\Phi}U^{*}$ $\Rightarrow$ $\Delta_{U\Omega,U\Phi}=U\Delta_{\Omega,\Phi}U^{*}.$
\end{enumerate}
\begin{proof}
(1) $\overline{\mathcal{R}U\Omega}=\overline{U\mathcal{R}\Omega}=U\overline{\mathcal{R}\Omega}=\mathcal{H}$
implies that $U\Omega$ is cyclic. $AU\Omega=0\Leftrightarrow U^{*}AU\Omega=0\Leftrightarrow U^{*}AU=0\Leftrightarrow A=0$
implies that $U\Omega$ is separating. Idem for $U\Phi$. (2) For
any $A\in\mathcal{R}$, we have $\left(US_{\Omega}U^{*}\right)AU\Omega=US_{\Omega}\left(U^{*}AU\right)\Omega=U\left(U^{*}AU\right)^{*}\Omega=A^{*}U\Omega$.
Then, applying the polar decomposition we have $\Delta_{U\Omega}=U\Delta_{\Omega}U^{\text{*}}$.
(3) For any $A\in\mathcal{R}$, we have $\left(US_{\Omega,\Phi}U^{*}\right)AU\Phi=US_{\Omega,\Phi}\left(U^{*}AU\right)\Phi=U\left(U^{*}AU\right)^{*}\Omega=A^{*}U\Omega\,.$
Then $\Delta_{U\Omega,U\Phi}=U\Delta_{\Omega,\Phi}U^{*}$ follows
from the polar decomposition.
\end{proof}
Given a state $\omega$ of a vN algebra $\mathcal{R}\subset\mathcal{B}\left(\mathcal{H}\right)$
and a unitary $U\in\mathcal{B}\left(\mathbb{\mathcal{H}}\right)$,
we denote by $\omega_{U}$ the state defined through $\omega_{U}\left(\cdot\right):=\omega\left(U^{*}\cdot U\right)$.
\begin{lem}
\textup{\label{sim_ent_conj}Given $\mathcal{R}\subset\mathcal{B}\left(\mathcal{H}\right)$
a vN algebra in standard form, $\omega$ a faithful normal state and
$U\in\mathcal{B}\left(\mathcal{H}\right)$ unitary such that $U^{*}\mathcal{R}U=\mathcal{R}$, then
\begin{equation}
S_{R}\left(\omega_{U}\mid\omega\right)=S_{R}\left(\omega\mid\omega_{U^{*}}\right)\,.
\end{equation}
}
\end{lem}
\begin{proof}
Let $\Omega$ be the cyclic and separating vector representative of
$\omega$. Then $U\Omega$, $U^{*}\Omega$ are the vector representatives
of $\omega_{U}$,$\omega_{U^{*}}$ and they are cyclic and separating
because of 1. in lemma \ref{lema_auto}. Using 3. of the same lemma
we have $S_{\Omega,U\Omega}=S_{UU^{*}\Omega,U\Omega}=US_{U^{*}\Omega,\Omega}U^{*}$.
Then $S_{R}\left(\omega_{U}\mid\omega\right)=\left\langle U\Omega,K_{\Omega,U\Omega}U\Omega\right\rangle _{\mathcal{H}}=\left\langle U\Omega,UK_{U^{*}\Omega,\Omega}U^{*}U\Omega\right\rangle _{\mathcal{H}}=\left\langle \Omega,K_{U^{*}\Omega,\Omega}\Omega\right\rangle _{\mathcal{H}}=S_{R}\left(\omega\mid\omega_{U^{*}}\right)$.
\end{proof}
Now, we come back to coherent states. From now on $\omega\left(\cdot\right)=\left\langle \Omega,\cdot\,\Omega\right\rangle _{\mathcal{H}}$
denotes the vacuum state. And given any $f\in\mathcal{S}\left(\mathbb{R}^{d},\mathbb{R}\right)$
we define the coherent state $\omega_{f}\left(\cdot\right)=\left\langle \Omega,W\left(f\right)^{*}\cdot\,W\left(f\right)\Omega\right\rangle _{\mathcal{H}}$. The Reeh-Schlieder theorem asserts that the vacuum vector $\Omega$
is cyclic and separating for any local algebra $\mathcal{R}\left(\mathcal{O}\right)$,
and lemma \ref{lema_auto} ensures the same for the coherent vector
$W\left(f\right)\Omega$. Then, lemma \ref{sim_ent_conj} implies
\begin{equation}
S_{R}\left(\omega_{f}\mid\omega\right)=S_{R}\left(\omega\mid\omega_{-f}\right)\,,\label{sr-antisym}
\end{equation}
for any coherent state $\omega_{f}$ and any local algebra $\mathcal{R}\left(\mathcal{O}\right)$.
Moreover, the net algebra of the free scalar field has a global $\mathbb{Z}_{2}$-symmetry
implemented by an operator $z=z^{-1}=z^{*}$ such that\footnote{In the Lagrangian approach to QFT, this is the usual symmetry $\phi\left(x\right)\rightarrow-\phi\left(x\right)$.}
\begin{equation}
zW\left(f\right)z=W\left(-f\right)=W\left(f\right)^{*}\,,\quad\mathrm{and}\quad z\Omega=\Omega\,.\label{z2-sym}
\end{equation}
This motivates the following lemma.
\begin{lem}
\textup{For any local algebra $\mathcal{R}\left(\mathcal{O}\right)$,
the relative entropy between a coherent state $\omega_{f}$ and the
vacuum state $\omega$ satisfies 
\begin{equation}
S_{R}\left(\omega_{f}\mid\omega\right)=S_{R}\left(\omega_{-f}\mid\omega\right)\,.\label{z2-sym-sr}
\end{equation}
}
\end{lem}
\begin{proof}
Let $\Omega$, $W\left(f\right)\Omega$ and $W\left(f\right)^{*}\Omega$
be the vector representatives of the states $\omega$, $\omega_{f}$
and $\omega_{-f}$. If $S_{\Omega,f}$ is the relative modular involution
associated to $\left\{ \mathcal{R}\left(\mathcal{O}\right),W\left(f\right)\Omega,\Omega\right\} $
and employing the $\mathbb{Z}_{2}$-symmetry \eqref{z2-sym}, we have
that
\begin{equation}
\left(zS_{\Omega,f}z\right)W\left(g\right)W\left(f\right)^{*}\Omega=zS_{\Omega,f}W\left(g\right)^{*}W\left(f\right)\Omega=zW\left(g\right)\Omega=W\left(g\right)^{\text{*}}\Omega\,,
\end{equation}
for all $W\left(g\right)\in\mathcal{R}\left(\mathcal{O}\right)$.
Then $S_{\Omega,-f}=zS_{\Omega,f}z$ and hence $K_{\Omega,-f}=zK_{\Omega,f}z$.
Finally, $S_{R}\left(\omega_{f}\mid\omega\right)=\left\langle \Omega,K_{\Omega,f}\Omega\right\rangle _{\mathcal{H}}=\left\langle \Omega,K_{\Omega,-f}\Omega\right\rangle _{\mathcal{H}}=S_{R}\left(\omega_{-f}\mid\omega\right)$.
\end{proof}
\begin{rem}
The above lemma should apply to any scalar theory with $\mathbb{Z}_{2}$-symmetry
as above, satisfying the Wightman axioms.
\end{rem}
Finally, combining \eqref{sr-antisym} and \eqref{z2-sym-sr} we have
the following theorem concerning the symmetry for the relative entropy
between coherent states. 
\begin{thm}
\textup{\label{thm_sym_sr}For any local algebra $\mathcal{R}\left(\mathcal{O}\right)$,
the relative entropy between a coherent state $\omega_{f}$ and the
vacuum state $\omega$ is symmetric, i.e. 
\begin{equation}
S_{R}\left(\omega_{f}\mid\omega\right)=S_{R}\left(\omega\mid\omega_{f}\right)\,.\label{sym-sr}
\end{equation}
}
\end{thm}
To end, we have the following theorem concerning the relative entropy
between two coherent states.\footnote{This result has been found in the past using other methods. For example,
see \cite{Lashkari:2015dia} for a derivation using the replica trick
for $2d$ free CFTs.}
\begin{thm}
\textup{For any local algebra $\mathcal{R}\left(\mathcal{O}\right)$,
the relative entropy between two coherent states $\omega_{f}$ and
$\omega_{g}$ satisfies
\begin{equation}
S_{R}\left(\omega_{f}\mid\omega_{g}\right)=S_{R}\left(\omega_{f-g}\mid\omega\right)\,.
\end{equation}
}
\end{thm}
\begin{proof}
Let $\Omega$, $W\left(f\right)\Omega$, $W\left(g\right)\Omega$
and the vector representatives of the states $\omega$, $\omega_{f}$
and $\omega_{g}$. If $S_{g,f}$ is the relative modular involution
associated to $\left\{ \mathcal{R}\left(\mathcal{O}\right),W\left(f\right)\Omega,W\left(g\right)\Omega\right\} $,
then because of 3. in lemma\ref{lema_auto} we have that $S_{U\Omega,U\Psi}=W\left(g\right)^{*}S_{g,f}W\left(g\right)$
is the relative modular involution associated to $\left\{ \mathcal{R}\left(\mathcal{O}\right),W\left(g\right)^{*}W\left(f\right)\Omega,\Omega\right\} $.
Since $W\left(g\right)^{*}W\left(f\right)\Omega$ is a vector representative
of $\omega_{f-g}$ , we have $S_{R}\left(\omega_{f}\mid\omega_{g}\right)=\left\langle W\left(f\right)\Omega,S_{g,f}W\left(f\right)\Omega\right\rangle _{\mathcal{H}}=S_{R}\left(\omega_{f-g}\mid\omega\right)\,.$
\end{proof}

\subsection{Relative entropy for the Rindler wedge}

Let $\mathcal{R}_{\mathcal{W}}$ be the right Rindler wedge algebra,
$\omega$ the vacuum state and $\omega_{f}$ a coherent state with
$f\in\mathcal{S}\left(\mathbb{R}^{d},\mathbb{R}\right)$. Let call
$\Omega$ and $\Phi:=W\left(f\right)\Omega$ its vector representatives.
The aim of this subsection is to compute the relative entropy $S_{R}\left(\omega_{f}\mid\omega\right)$,
and for that we need to calculate the relative modular Hamiltonian
$K_{\Omega,\Phi}$ (or $K_{\Phi,\Omega}$ according to theorem \ref{thm_sym_sr}).
As we explained in the last subsection, the vectors $\Omega$ and
$\Phi:=W\left(f\right)\Omega$ are cyclic and separating. We distinguish
between two cases,
\begin{eqnarray}
\textrm{easy case} & : & f=f_{L}+f_{R}\,,\label{eq:case_prod}\\
\textrm{hard case} & : & f\neq f_{L}+f_{R}\,,\label{eq:case_no_prod}
\end{eqnarray}
where $supp\left(f_{L}\right)\in\mathcal{W}'$ and $supp\left(f_{R}\right)\in\mathcal{W}$.\footnote{In particular, the easy case includes the cases when $W\left(f\right)\in\mathcal{R}{}_{\mathcal{W}}$
or $W\left(f\right)\in\mathcal{R}{}_{\mathcal{W}'}$.} In the following subsections, we deal with each case \eqref{eq:case_prod}
and \eqref{eq:case_no_prod} separately. 

\subsubsection{Easy case: $f=f_{L}+f_{R}$}

In this case, we have that the coherent vector can be written as $W\left(f\right)=W\left(f_{L}\right)W\left(f_{R}\right)$
with $W\left(f_{L}\right)\in\mathcal{R}_{\mathcal{W}'}$ and $W\left(f_{R}\right)\in\mathcal{R}_{\mathcal{W}}$.
This case can be solved in general using the following lemma.
\begin{lem}
\textup{\label{lema_easy_case}Given $\mathcal{R}\subset\mathcal{B}\left(\mathcal{H}\right)$
a vN algebra, $\Omega$ a cyclic and separating and $U\in\mathcal{R}$
and $U'\in\mathcal{R}'$ unitaries. Then $\Phi=U'U\Omega$ is cyclic
and separating and
\begin{equation}
S_{\Omega,\Phi}=US_{\Omega}U'^{*}\,,
\end{equation}
and by polar decomposition we have $J_{\Omega,\Phi}=UJ_{\Omega}U'^{*}$,
$\Delta_{\Omega,\Phi}=U'\Delta_{\Omega}U'^{*}$ and $K_{\Omega,\Phi}=U'K_{\Omega}U'^{*}$.}
\end{lem}
\begin{proof}
$\overline{\mathcal{R}\Phi}\mathcal{H}=\overline{\mathcal{R}U'U\Omega}=U'\overline{\mathcal{R}U\Omega}=U'\overline{\mathcal{R}{}_{\mathcal{W}}\Omega}=\overline{\mathcal{R}\Omega}=\mathcal{H}$
implies $\Phi$ is cyclic. Since the same argument holds for $\mathcal{R}'$,
then $\Phi$ is separating for $\mathcal{R}$. For any $A\in\mathcal{R}$,
we have that $\left(US_{\Omega}U'^{*}\right)A\Phi=US_{\Omega}U'^{*}AU'U\Phi=U_{R}S_{\Omega}\left(AU\right)\Omega=U\left(AU\right)^{*}\Omega=A^{*}\Omega$
$\Rightarrow\,S_{\Omega,\Phi}=US_{\Omega}U'^{*}$.
\end{proof}
\begin{cor}
\textup{\label{cor_re_easy}In the context of the above lemma, if
$\Omega$ and $\Phi$ are vector representatives of the states $\omega$
and $\phi$, then $S_{R}\left(\phi\mid\omega\right)=\left\langle \Phi,U'K_{\Omega}U'^{*}\Phi\right\rangle _{\mathcal{H}}=\left\langle \Omega,U^{*}K_{\Omega}U\Omega\right\rangle _{\mathcal{H}}$.}
\end{cor}
The above corollary shows explicitly that the relative entropy does
not depend on the unitary $U'$. This is expected because the relative
entropy is a measure of indistinguishability of the states in $\mathcal{R}$,
and indeed has to be invariant under changes of the states outside
$\mathcal{R}$.

Now we apply the corollary \ref{cor_re_easy} to the case of a coherent
state, i.e. $U=W\left(f_{R}\right)$ with $supp\left(f_{R}\right)\subset\mathcal{W}$.
Remembering that the second quantized Poincaré unitary operator $U\left(\Lambda_{1}^{s},0\right)=\mathrm{e}^{iK_{1}s}$,
acting on the the Fock space $\mathcal{H}$, is constructed from the
Poincaré unitary operator $u\left(\Lambda_{1}^{s},0\right)=\mathrm{e}^{ik_{1}s}$,
acting on the one-particle Hilbert space $\mathfrak{H}$, then we have
that
\begin{equation}
S_{R}\left(\phi\mid\omega\right)=\left\langle \Omega,U^{*}K_{\Omega}U\Omega\right\rangle _{\mathcal{H}}=2\pi\left\langle \Omega,W\left(f_{R}\right)^{*}K_{1}W\left(f_{R}\right)\Omega\right\rangle _{\mathcal{H}}=2\pi\left\langle f_{R},k_{1}f_{R}\right\rangle _{\mathfrak{H}}\,,\label{eq:rel_ent_easy}
\end{equation}
where the last equality is fully calculated in Appendix \eqref{subsec:appendix_coherent}.
Thus, the relative entropy between the coherent state and the vacuum,
can be expressed, in the one-particle Hilbert space $\mathfrak{H}$,
in terms of the expectation value of the boost operator $k_{1}$ in
the vector $E\left(f\right)\in\mathfrak{H}$ which generates the coherent
state. At the end, following from \ref{eq:rel_ent_easy} we get the
following theorem.
\begin{thm}
\textup{\label{thm_sr_coh_par}Let $f_{L},f_{R}\in\mathcal{S}\left(\mathbb{R}^{d},\mathbb{R}\right)$
with $supp\left(f_{L}\right)\in\mathcal{W}$ and $supp\left(f_{R}\right)\in\mathcal{W}'$,
and $f=f_{L}+f_{R}$. Then the relative entropy between the coherent
state $\omega_{f}$ and the vacuum $\omega$, for the right Rindler
wedge algebra $\mathcal{R}_{\mathcal{W}}$, is
\begin{equation}
S_{R}\left(\omega_{f}\mid\omega\right)=2\pi\int_{x^{1}>0}\negthickspace d^{d-1}x\,x^{1}\left.\frac{1}{2}\left(\left(\frac{\partial F}{\partial x^{0}}\right)^{2}+\left|\nabla F\right|^{2}+m^{2}F^{2}\right)\right|_{x^{0}=0},\label{eq:rel_ent_easy2}
\end{equation}
where $F\left(x\right)=\int_{\mathbb{R}^{d}}\Delta\left(x-y\right)\,f\left(y\right)d^{d}y=\int_{\mathbb{R}^{d}}\Delta\left(x-y\right)\,\left[f_{L}\left(y\right)+f_{R}\left(y\right)\right]d^{d}y$
. In addition, formula \ref{eq:rel_ent_easy2} does not depend in
the function $f_{L}$ (with support in $\mathcal{W}'$) chosen.}
\end{thm}
\begin{proof}
A straightforward calculation explained in appendix \ref{subsec:appendix_2}
allows us to rewrite the expression \eqref{eq:rel_ent_easy} as equation
\eqref{apend_sr_easy}. However, there are already two differences
between \eqref{apend_sr_easy} and \eqref{eq:rel_ent_easy2} (beyond
the obvious $2\pi$ in front of the expression). The first one is
that in \eqref{apend_sr_easy} the integral is along the whole space
$\mathbb{R}^{d-1}$, and the second one is that the function $F$
in \eqref{apend_sr_easy} is computed using only $f_{R}$. To finally
pass from \eqref{apend_sr_easy} to \eqref{eq:rel_ent_easy2} we have
to make the following two changes. First notice that because $supp\left(f_{R}\right)\subset\mathcal{W}\;\Rightarrow supp\left(\left.F\right|_{x^{0}=0}\right)\subset\Sigma$, this allows us to replace the integration region in \eqref{apend_sr_easy}
by $\Sigma$. Similarly, because $supp\left(f_{L}\right)\subset\mathcal{W}'\;\Rightarrow$
the function $F_{L}\left(x\right):=\int_{\mathbb{R}^{d}}\Delta\left(x-y\right)\,f_{L}\left(y\right)d^{d}y$
vanishes along $\Sigma$ and hence \eqref{eq:rel_ent_easy2} holds.
This also implies that \eqref{eq:rel_ent_easy2} does not depend
on $f_{L}$.
\end{proof}
As a remark, the outcome of the above theorem coincides with \eqref{cali}
for the canonical stress tensor \eqref{eq:stress}.

\subsubsection{Hard case: $f\protect\neq f_{L}+f_{R}$ \label{subsec:Hard-case}}

In this section, we assume that the function $f\in\mathcal{S}\left(\mathbb{R}^{d},\mathbb{R}\right)$
has $supp\left(f\right)\not\subset\mathcal{W},\mathcal{W}'$. Moreover,
we assume that $supp\left(f\right)$ is compact in order to avoid
some possible complications coming from integrals along regions of
infinite size. At the end, we are interested in the behavior of the
relative entropy around the boundary of the wedge region $\partial\Sigma=\left\{ \bar{x}\in\mathbb{R}^{d-1}\,:\,x^{1}=0\right\} $,
which can be captured with a compactly supported coherent state.

Before we continue, we remark that, in this case, the relative
entropy must be finite. The proof is as follows. Since $supp\left(f\right)$
is compact, there exists a ``bigger'' right wedge $\tilde{\mathcal{W}}_{R}\supset\mathcal{W}$
such that $W\left(f\right)\in\tilde{\mathcal{W}}_{R}$. Then the relative
entropy between this coherent and the vacuum in the algebra $\mathcal{R}\left(\tilde{\mathcal{W}}_{R}\right)$
is as the one computed in the previous section, which is finite because
the generating function $f$ is smooth. Then by monotonicity, the
relative entropy for the original wedge $\mathcal{W}$ must be finite.
In particular, we are allowed to use expression \eqref{eq:rel_ent_f}.

The first question which arises is whether we could split the unitary into
two unitaries, one belonging to the right wedge $\mathcal{W}$ and
the other to the left wedge $\mathcal{W}'$. In other words, if there
exists unitaries $U_{R}\in\mathcal{R}_{\mathcal{W}}$ and $U_{L}\in\mathcal{R}_{\mathcal{W}'}$
unitaries such that $W\left(f\right)=U_{L}U_{R}$. Unfortunately the
answer is no, almost for the most general interesting case. This fact
arises when we try to explicitly split $W\left(f\right)$. To begin,
it seems natural to split the function $f$ simply as
\begin{eqnarray}
f_{R}\left(x\right)\!\!\!\! & := & \!\!\!\!\Theta_{\mathcal{W}}\left(x\right)f\left(x\right)\,,\label{eq:partir_mal}\\
f_{L}\left(x\right)\!\!\!\! & := & \!\!\!\!\Theta_{\mathcal{W}'}\left(x\right)f\left(x\right)\,,
\end{eqnarray}
where $\Theta_{\mathcal{W}}$ is the characteristic function of the
right Rindler wedge (equivalently for $\Theta_{\mathcal{W}'}$). However,
it leads to a wrong result, since $f_{R}+f_{L}\neq f$. Moreover,
if for example we start with a function $f$ supported in the
upper light cone $V^{+}:=\left\{ x\in\mathbb{R}^{d}\,:\,x^{0}>\left|\bar{x}\right|\right\} $,
then equation \eqref{eq:partir_mal} implies that $f_{R}\equiv0$
and hence we obtain $S_{R}\left(\phi\mid\omega\right)=0$, which is
obviously the wrong result. To make a consistent splitting, we must
use the relations explained in subsection \ref{subsec:Relation-between-algebras}.
Given the spacetime function $f\in\mathcal{S}\left(\mathbb{R}^{d},\mathbb{R}\right)$
we can construct $f_{\varphi},\,f_{\pi}\in\mathcal{S}\left(\mathbb{R}^{d-1},\mathbb{R}\right)$
satisfying the relation \eqref{eq:segal_vs_weyl}. The correct result
is to split these functions $f_{\varphi},\,f_{\pi}$, which are the
initial data at $x^{0}=0$ of the Klein-Gordon solution generated
by $f$. The assumption $supp\left(f\right)\not\subset\mathcal{W},\mathcal{W}'$
implies that an open neighborhood of the origin $x=0$ is included
in the supports of $f_{\varphi}$ and $f_{\pi}$. Now, we write
\begin{eqnarray}
f_{\varphi}=f_{\varphi,L}+f_{\varphi,R} & \mathrm{and} & f_{\pi}=f_{\pi,L}+f_{\pi,R}\:\textrm{,}
\end{eqnarray}
with $supp\left(f_{\nu,L}\right)\in\Sigma'$ and $supp\left(f_{\nu,R}\right)\in\Sigma$
($\nu=\varphi,\,\pi$). The right way to do this is taking
\begin{eqnarray}
f_{\nu,L}\left(\bar{x}\right):=f_{\nu}\left(\bar{x}\right)\cdot\Theta\left(-x^{1}\right) & \mathrm{and} & f_{\nu,R}\left(\bar{x}\right):=f_{\nu}\left(\bar{x}\right)\cdot\Theta\left(x^{1}\right)\,,
\end{eqnarray}
where $\Theta$ is the usual step Heaviside function. The problem
is that $f_{\nu,L}$ and $f_{\nu,R}$ are no longer smooth, and nothing
guarantees that $E_{\nu}\left(f_{\nu,R}\right)\in\mathfrak{H}_{\nu}$
(the same problem occurs for $f_{\nu,L}$). More precisely, since
$f_{\nu,R}\in L^{2}\left(\mathbb{R}^{d-1},\mathbb{R}\right)=H^{0}\left(\mathbb{R}^{d-1},\mathbb{R}\right)$,
and because of the inclusions (see Appendix \ref{subsec:appendix_Sobolev})
\begin{equation}
H^{\frac{1}{2}}\left(\mathbb{R}^{d-1},\mathbb{R}\right)\subset H^{0}\left(\mathbb{R}^{d-1},\mathbb{R}\right)\subset H^{-\frac{1}{2}}\left(\mathbb{R}^{d-1},\mathbb{R}\right)\,,
\end{equation}
we have that $f_{\varphi,R}\in\mathfrak{H}_{\varphi}$ but $f_{\pi,R}\notin\mathfrak{H}_{\pi}$.
In other words, $f_{\pi,R}$ is not an appropriate smear function
for the canonical conjugate field $\pi\left(\bar{x}\right)$. This
problem does not arise because the test function is no longer smooth,
it is just because $f_{\pi,R}$ is no longer continuous. On the other
hand, if $f_{\pi,R}$ is continuous, the problem can be solved due
to the following lemma. 
\begin{thm}
\label{par:lemma-1}\textup{Let $f\in L^{2}\left(\mathbb{R}^{n}\right)\cap C^{0}\left(\mathbb{R}^{n}\right)\cap C_{t}^{1}\left(\mathbb{R}^{n}\right)$
and $\partial_{j}f\in L^{2}\left(\mathbb{R}^{n}\right)$ for $j=1,\ldots,n$.}\footnote{$C_{t}^{1}\left(\mathbb{R}^{n}\right)$ is the set of piecewise differentiable
functions. See Appendix \ref{subsec:appendix_Sobolev} for a proper
definition.}\textup{ Then $f\in H^{1}\left(\mathbb{R}^{n}\right)$.} 
\end{thm}
\begin{proof}
See appendix \ref{subsec:appendix_Sobolev}. 
\end{proof}
Then, having this in mind, the strategy we adopt below is to make
a splitting for some other smear function which, by construction,
we know is continuous. 

\subsubsection{A lemma for the relative modular flow}

In this subsection, we prove a lemma that gives a general expression
for the relative modular flow, under the assumption that some non-local
operator can be written as a product of two new operators, one belonging
to $\mathcal{R}$ and another to $\mathcal{R}'$. In the following subsection,
we prove that this assumption is already valid for the free hermitian
scalar field. For simplicity and due to the symmetry relation \eqref{sym-sr},
in the following we work with the modular operator $\Delta_{\Phi,\Omega}$
instead of $\Delta_{\Omega,\Phi}$.

As a motivation, we remember that, contrary to the modular flow $\Delta_{\Omega}^{it}$
and the relative modular flow $\Delta_{\Phi,\Omega}^{it}$, the Connes
Radon-Nikodym cocycle $u_{\Phi,\Omega}\left(t\right)=\Delta_{\Phi,\Omega}^{it}\Delta_{\Omega}^{it}$
belongs to the algebra $\mathcal{R}$. This makes us think that the
computation of $u_{\Phi,\Omega}\left(t\right)$ may involve the splitting
of some test function, which at the end, will lead to a well-defined
operator. To gain some intuition, using lemmas \ref{lema_auto} and
\ref{lema_easy_case}, we know that
\begin{equation}
u_{\Phi,\Omega}\left(t\right)=U^{*}\Delta_{\Omega}^{it}U\Delta_{\Omega}^{-it}\quad\textrm{when }\Phi=U'U\Omega\textrm{ with }U\in\mathcal{R},\,U'\in\mathcal{R}'\,.
\end{equation}
This expression motivates the following lemma.
\begin{lem}
\textup{\label{l_par_uni}Let $\mathcal{R}\subset\mathcal{B}\left(\mathcal{H}\right)$
be a vN factor,}\footnote{A vN algebra $\mathcal{R}\subset\mathcal{B}\left(\mathcal{H}\right)$
is said to be a \textit{factor} if its center is trivial, i.e. $\mathcal{R}\cap\mathcal{R}'=\left\{ \lambda\cdot\mathbf{1}\right\} $.}\textup{ $\Omega$ a cyclic and separating vector, $U\in\mathcal{B}\left(\mathcal{H}\right)$
a unitary such $U^{*}\mathcal{R}U=\mathcal{R}$ and $\Phi=U\Omega$.
If there exists a family of unitaries}\footnote{They are not necessarily one-parameter groups for $t\in\mathbb{R}$.}\textup{
$V\left(t\right)\in\mathcal{R},\:V'\left(t\right)\in\mathcal{R}'$
such that
\begin{align}
 & \begin{cases}
U^{*}\Delta_{\Omega}^{it}U\Delta_{\Omega}^{-it}=V\left(t\right)V'\left(t\right)\,,\quad\forall t\in\mathbb{R}\,,\\
V\left(0\right)=V'\left(0\right)=\mathbf{1}\,.
\end{cases}\label{eq:partir_unitario}
\end{align}
Then there exists a real function $\alpha:\mathbb{R}\rightarrow\mathbb{R}$
with $\alpha\left(0\right)=0$ such
\begin{equation}
\Delta_{\Phi,\Omega}^{it}=\mathrm{e}^{-i\alpha\left(t\right)}V\left(t\right)\Delta_{\Omega}^{it}\,\textrm{.}\label{eq:rel_mod_decom_2}
\end{equation}
}
\end{lem}
\begin{proof}
We first see that $V\left(t\right)\Delta_{\Omega}^{it}$ has the same
action as $\Delta_{\Phi,\Omega}^{it}$ over every$A\in\mathcal{R}$
and $A'\in\mathcal{R}'$. Indeed
\begin{eqnarray}
\mathcal{R}\ni V\left(t\right)\Delta^{it}A\Delta^{-it}V\left(t\right)^{*}\!\!\!\! & = & \!\!\!\!V\left(t\right)V'\left(t\right)\Delta_{\Omega}^{it}A\Delta_{\Omega}^{-it}V\left(t\right)^{*}V'\left(t\right)^{*}=U\Delta_{\Omega}^{it}U^{*}\Delta_{\Omega}^{-it}\,\Delta_{\Omega}^{it}A\Delta_{\Omega}^{-it}\,\Delta_{\Omega}^{it}U\Delta_{\Omega}^{-it}U^{*}\nonumber \\
 & = & \!\!\!\!U\Delta_{\Omega}^{it}U^{*}AU\Delta_{\Omega}^{-it}U^{*}=\Delta_{\Phi}^{it}A\Delta_{\Phi}^{-it}=\Delta_{\Phi,\Omega}^{it}A\Delta_{\Phi,\Omega}^{-it}\,\textrm{,}
\end{eqnarray}
where we have used 2. in lemma \ref{lema_auto}. Similarly,
\begin{equation}
V\left(t\right)\underset{\in\mathcal{R}{}_{\mathcal{W}'}}{\underbrace{\Delta_{\Omega}^{it}A'\Delta_{\Omega}^{-it}}}V\left(t\right)^{*}=V\left(t\right)V\left(t\right)^{*}\,\Delta_{\Omega}^{it}A'\Delta_{\Omega}^{-it}=\Delta_{\Omega}^{it}A'\Delta_{\Omega}^{-it}=\Delta_{\Phi,\Omega}^{it}A'\Delta_{\Phi,\Omega}^{-it}\,.
\end{equation}
Then for all $B\in\mathcal{R}\cup\mathcal{R}'$ we have
\begin{equation}
\left(V\left(t\right)\Delta_{\Omega}^{it}\right)B\left(V\left(t\right)\Delta_{\Omega}^{it}\right)^{*}=\Delta_{\Phi,\Omega}^{it}B\Delta_{\Phi,\Omega}^{-it}\:\Rightarrow\:\left[B,\left(V\left(t\right)\Delta_{\Omega}^{it}\right)^{*}\Delta_{\Phi,\Omega}^{it}\right]=0\,,
\end{equation}
and hence $\left(V\left(t\right)\Delta_{\Omega}^{it}\right)^{*}\Delta_{\Phi,\Omega}^{it}$
belongs to the center $\left(\mathcal{R}\cup\mathcal{R}'\right)'=\mathcal{R}\cap\mathcal{R}'=\left\{ \lambda\cdot\mathbf{1}\right\} $,
which is trivial since $\mathcal{R}$ is a factor. This means that
there exists a function $\lambda:\mathbb{R}\rightarrow\mathbb{C}$
such that
\begin{equation}
\Delta_{\Phi,\Omega}^{it}=\lambda\left(t\right)V\left(t\right)\Delta_{\Omega}^{it}\,.\label{eq:rel_mod_decom}
\end{equation}
Moreover, evaluating the above expression at $t=0$ we get that $\lambda\left(0\right)=1$.
Finally, since all operators in \eqref{eq:rel_mod_decom} are unitaries, we have that $\lambda\left(t\right)=\mathrm{e}^{-i\alpha\left(t\right)}$
for some function $\alpha:\mathbb{R}\rightarrow\mathbb{R}$ with $\alpha\left(0\right)=0$,
and then \eqref{eq:rel_mod_decom_2} holds.
\end{proof}
Under the hypothesis of the above lemma, we obtain the relative
modular Hamiltonian deriving \eqref{eq:rel_mod_decom_2} at $t=0$,
\begin{equation}
K_{\Phi,\Omega}=i\,\left.\frac{\mathrm{d}}{\mathrm{d}t}\right|_{t=0}\Delta_{\Phi,\Omega}^{it}=i\,\left.\frac{\mathrm{d}}{\mathrm{d}t}\right|_{t=0}\mathrm{e}^{-i\alpha\left(t\right)}V\left(t\right)\Delta_{\Omega}^{it}=\alpha'\left(0\right)\mathbf{1}+i\dot{V}\left(0\right)+K_{\Omega}\,,\label{eq:rel_mod_ham_p}
\end{equation}
where the derivative in \eqref{eq:rel_mod_ham_p} has to be understood
as a limit in the strong operator topology of $\mathcal{H}$. This
formula gives a well-defined expression for the relative modular Hamiltonian
up to a constant. One way to determine such a constant is using that
$\Delta_{\Phi,\Omega}^{it}$ is a one-parameter group of unitaries and
must fulfil the concatenation equation
\begin{eqnarray}
\Delta_{\Phi,\Omega}^{it_{1}}\Delta_{\Phi,\Omega}^{it_{2}}=\Delta_{\Phi,\Omega}^{i\left(t_{1}+t_{2}\right)}\,, &  & \forall t_{1},t_{2}\in\mathbb{R}\,.\label{eq:concatenar}
\end{eqnarray}
We discuss the computation to determine $\alpha'\left(0\right)$
in subsection \ref{subsec:determ_alpha}.

\subsubsection{Relative modular flow for coherent states\label{subsec:Computation-the-splitting}}

In this subsection, we apply lemma \ref{l_par_uni} for the theory
of a real scalar field. More concretely, we show that the splitting
of such lemma can be done for a general coherent state. Indeed we
have the following theorem.
\begin{thm}
\textup{\label{th_mod_rel_coh}For the algebra of the Rindler wedge
$\mathcal{R}_{\mathcal{W}}$, a Weyl unitary $U=W\left(f\right)$
with $f\in\mathcal{S}\left(\mathbb{R}^{d},\mathbb{R}\right)$, the
vacuum vector $\Omega$ and $\Phi=U\Omega$, the hypothesis in lemma
\eqref{l_par_uni} holds. In particular we have that
\begin{equation}
\Delta_{\Phi,\Omega}^{it}=\mathrm{e}^{i\alpha\left(s\right)}W_{\varphi}\left(g_{\varphi,R}^{s}\right)W_{\pi}\left(g_{\pi,R}^{s}\right)\Delta_{\Omega}^{it}=\mathrm{e}^{i\alpha\left(s\right)}\mathrm{e}^{i\varphi\left(g_{\varphi,R}^{s}\right)}\mathrm{e}^{i\pi\left(g_{\pi,R}^{s}\right)}\mathrm{e}^{isK_{1}}\:\textrm{,}\label{rel_mod_flow_coh}
\end{equation}
where we have denoted $s:=-2\pi t$ and
\begin{eqnarray}
g_{\varphi,R}^{s}\left(\bar{x}\right)\!\!\!\! & = & \!\!\!\!-\frac{\partial G^{s}}{\partial x^{0}}\left(0,\bar{x}\right)\Theta\left(x^{1}\right)\,,\\
g_{\pi,R}^{s}\left(\bar{x}\right)\!\!\!\! & = & \!\!\!\!G^{s}\left(0,\bar{x}\right)\Theta\left(x^{1}\right)\,,\\
G^{s}\left(x\right)\!\!\!\! & = & \!\!\!\!\int_{\mathbb{R}^{d}}\Delta\left(x-y\right)\,\left[f\left(\Lambda_{1}^{-s}y\right)-f\left(y\right)\right]d^{d}y\,.
\end{eqnarray}
}
\end{thm}
\begin{proof}
From relations (\ref{eq:axioms_ft_1st}-\ref{eq:axioms_ft_last})
we have that $\mathcal{R}_{\mathcal{W}}$ is a vN factor. From the Reeh-Schlieder
theorem, we have that the vacuum vector $\Omega$ is cyclic and separating.
And we have already discussed that any Weyl unitary satisfies $W\left(\boldsymbol{h}\right)^{*}\mathcal{R}_{\mathcal{W}}W\left(\boldsymbol{h}\right)=\mathcal{R}_{\mathcal{W}}$.
From now on, we set $s=-2\pi t$ and we replace $U=W\left(f\right)$
in \eqref{eq:partir_unitario}
\begin{eqnarray}
W\left(f\right)^{*}\Delta_{\Omega}^{it}W\left(-f\right)\Delta_{\Omega}^{-it}\!\!\!\! & = & \!\!\!\!W\left(-f\right)\mathrm{e}^{isK_{1}}W\left(f\right)\mathrm{e}^{-isK_{1}}=W\left(-f\right)W\left(f_{\left(\Lambda_{1}^{s},0\right)}\right)\nonumber \\
 & = & \!\!\!\!\mathrm{e}^{i\mathrm{Im}\left\langle f,f_{\left(\Lambda_{1}^{s},0\right)}\right\rangle _{\mathfrak{H}}}W\left(f_{\left(\Lambda_{1}^{s},0\right)}-f\right)=\mathrm{e}^{i\mathrm{Im}\left\langle f,f^{s}\right\rangle _{\mathfrak{H}}}W\left(f^{s}-f\right)\,,\label{eq:partir_unitario_2}
\end{eqnarray}
where we have defined $f^{s}:=f_{\left(\Lambda_{1}^{s},0\right)}$.
Applying the decomposition \eqref{eq:segal_vs_weyl_1p} to $g^{s}:=f^{s}-f$
we have
\begin{eqnarray}
W\left(f\right)^{*}\Delta_{\Omega}^{it}W\left(f\right)\Delta_{\Omega}^{-it}\!\!\!\! & = & \!\!\!\!\mathrm{e}^{i\mathrm{Im}\left\langle f,f^{s}\right\rangle _{\mathfrak{H}}}W\left(g^{s}\right)=\mathrm{e}^{i\mathrm{Im}\left\langle f,f^{s}\right\rangle _{\mathfrak{H}}}\mathrm{e}^{i\mathrm{Im}\left\langle g_{\varphi}^{s},g_{\pi}^{s}\right\rangle _{\mathfrak{H}}}W_{\varphi}\left(g_{\varphi}^{s}\right)W_{\pi}\left(g_{\pi}^{s}\right)\,,\label{eq:partir_unitario_3}
\end{eqnarray}
with
\begin{eqnarray}
g_{\varphi}^{s}\left(\bar{x}\right)\!\!\!\! & = & \!\!\!\!-\frac{\partial G^{s}}{\partial x^{0}}\left(0,\bar{x}\right)=-\cosh\left(s\right)\frac{\partial F}{\partial x^{0}}\left(\bar{x}^{s}\right)+\sinh\left(s\right)\frac{\partial F}{\partial x^{1}}\left(\bar{x}^{s}\right)+\frac{\partial F}{\partial x^{0}}\left(0,\bar{x}\right)\,,\label{eq:fpi_partible}\\
g_{\pi}^{s}\left(\bar{x}\right)\!\!\!\! & = & \!\!\!\!G^{s}\left(0,\bar{x}\right)=F\left(\bar{x}^{s}\right)-F\left(0,\bar{x}\right)\,,
\end{eqnarray}
where $\bar{x}^{s}:=\left(\Lambda_{1}^{-s}x\right)_{x^{0}=0}=\left(-x^{1}\sinh\left(s\right),x^{1}\cosh\left(s\right),\bar{x}_{\bot}\right)$
and
\begin{eqnarray}
G^{s}\left(x\right)\!\!\!\! & = & \!\!\!\!\int_{\mathbb{R}^{d}}\Delta\left(x-y\right)\,\left[f^{s}\left(y\right)-f\left(y\right)\right]d^{d}y\,.\label{eq:Gs}
\end{eqnarray}
Now, we explicitly split the unitaries $W_{\varphi}\left(g_{\varphi}^{s}\right)$
and $W_{\pi}\left(g_{\pi}^{s}\right)$ in equation \eqref{eq:partir_unitario_3}
defining
\begin{eqnarray}
g_{\varphi,R}^{s}\left(\bar{x}\right):=g_{\varphi}^{s}\left(\bar{x}\right)\Theta\left(x^{1}\right) & \mathrm{and} & g_{\varphi,L}^{s}\left(\bar{x}\right):=g_{\varphi}^{s}\left(\bar{x}\right)\Theta\left(-x^{1}\right)\,,\label{eq:split_phi}\\
g_{\pi,R}^{s}\left(\bar{x}\right):=g_{\pi}^{s}\left(\bar{x}\right)\Theta\left(x^{1}\right) & \mathrm{and} & g_{\pi,L}^{s}\left(\bar{x}\right):=g_{\pi}^{s}\left(\bar{x}\right)\Theta\left(-x^{1}\right)\,,\label{eq:split_pi}
\end{eqnarray}
which clearly implies that $g_{\varphi,L}^{s}+g_{\varphi,R}^{s}=g_{\varphi}^{s}$
and $g_{\pi,L}^{s}+g_{\pi,R}^{s}=g_{\pi}^{s}$. Moreover
\begin{equation}
\begin{array}{cc}
\left.\begin{array}{c}
g_{\varphi,R}^{s},g_{\varphi,L}^{s}\in L^{2}\left(\mathbb{R}^{d-1},\mathbb{R}\right)\subset H^{-\frac{1}{2}}\left(\mathbb{R}^{d-1},\mathbb{R}\right)\\
supp\left(g_{\varphi,R}^{s}\right)\subset\Sigma\textrm{ and }supp\left(g_{\pi,L}^{s}\right)\subset\Sigma'
\end{array}\right\}  & \Rightarrow\end{array}g_{\varphi,R}^{s}\in K_{\varphi}\left(\Sigma\right)\textrm{ and }g_{\varphi,L}^{s}\in K_{\varphi}\left(\Sigma'\right)\,.
\end{equation}
Furthermore, we have that $g_{\pi,R}^{s},\,g_{\pi,L}^{s}$ are real-valued
functions and they satisfy the hypothesis in lemma \eqref{par:lemma-1}.
Then
\begin{equation}
\begin{array}{cc}
\left.\begin{array}{c}
g_{\pi,R}^{s},g_{\pi,L}^{s}\in H^{1}\left(\mathbb{R}^{d-1},\mathbb{R}\right)\subset H^{\frac{1}{2}}\left(\mathbb{R}^{d-1},\mathbb{R}\right)\\
supp\left(g_{\pi,R}^{s}\right)\subset\Sigma\textrm{ and }supp\left(g_{\pi,L}^{s}\right)\subset\Sigma'
\end{array}\right\}  & \Rightarrow\end{array}g_{\pi,R}^{s}\in K_{\pi}\left(\Sigma\right)\textrm{ and }g_{\pi,L}^{s}\in K_{\pi}\left(\Sigma'\right)\,,
\end{equation}
which means that the splits (\ref{eq:split_phi}-\ref{eq:split_pi})
work. Coming back to \eqref{eq:partir_unitario_3}, we have that
\begin{eqnarray}
W\left(f\right)\Delta_{\Omega}^{it}W\left(f\right)^{*}\Delta_{\Omega}^{-it}\!\!\!\! & = & \!\!\!\!\mathrm{e}^{i\mathrm{Im}\left\langle f,f^{s}\right\rangle _{\mathfrak{H}}}\mathrm{e}^{i\mathrm{Im}\left\langle g_{\varphi}^{s},g_{\pi}^{s}\right\rangle _{\mathfrak{H}}}W_{\varphi}\left(g_{\varphi,L}^{s}+g_{\varphi,R}^{s}\right)W_{\pi}\left(g_{\pi,L}^{s}+g_{\pi,R}^{s}\right)\nonumber \\
 & = & \!\!\!\!\mathrm{e}^{i\mathrm{Im}\left\langle f,f^{s}\right\rangle _{\mathfrak{H}}}\mathrm{e}^{i\mathrm{Im}\left\langle g_{\varphi}^{s},g_{\pi}^{s}\right\rangle _{\mathfrak{H}}}W_{\varphi}\left(g_{\varphi,L}^{s}\right)W_{\varphi}\left(g_{\varphi,R}^{s}\right)W_{\pi}\left(g_{\pi,L}^{s}\right)W_{\pi}\left(g_{\pi,R}^{s}\right)\\
 & = & \!\!\!\!\mathrm{e}^{i\mathrm{Im}\left\langle f,f^{s}\right\rangle _{\mathfrak{H}}}\mathrm{e}^{i\mathrm{Im}\left\langle g_{\varphi}^{s},g_{\pi}^{s}\right\rangle _{\mathfrak{H}}}\underset{=0}{\underbrace{\mathrm{e}^{-2i\mathrm{Im}\left\langle g_{\varphi,R}^{s},g_{\pi,L}^{s}\right\rangle _{\mathfrak{H}}}}}\underset{\in\mathcal{R}\mathcal{_{W}}}{\underbrace{W_{\varphi}\left(g_{\varphi,R}^{s}\right)W_{\pi}\left(g_{\pi,R}^{s}\right)}}\underset{\in\mathcal{R}\mathcal{_{W}}'}{\underbrace{W_{\varphi}\left(g_{\varphi,L}^{s}\right)W_{\pi}\left(g_{\pi,L}^{s}\right)}}\,.\nonumber 
\end{eqnarray}
Finally, replacing $V\left(t\right)=W_{\varphi}\left(g_{\varphi,R}^{s}\right)W_{\pi}\left(g_{\pi,R}^{s}\right)$
into \eqref{eq:rel_mod_decom_2} we arrive at \eqref{rel_mod_flow_coh}.
\end{proof}
\begin{rem}
Using the fact that the relative modular flow $\Delta_{\Phi,\Omega}^{it}$
is strongly continuous and that the relative entropy $S_{R}\left(\omega_{f}\mid\omega\right)$
is finite (see the discussion at the beginning of section \ref{subsec:Hard-case})
and hence the expression \eqref{eq:rel_ent_f} holds, it is not difficult
to show that the function $t\mapsto\left\langle \Omega,\Delta_{\Phi,\Omega}^{it}\Omega\right\rangle _{\mathcal{H}}$
is continuous differentiable. Furthermore, taking the vacuum expectation
value on the r.h.s of \eqref{rel_mod_flow_coh}, it can be proven
that the function $\alpha\left(s\right)\in C^{1}\left(\mathbb{R}\right)$. 
\end{rem}
Finally, from \eqref{eq:rel_mod_ham_p} we get the following expression
for the relative modular Hamiltonian
\begin{equation}
K_{\Phi,\Omega}=2\pi\left(\alpha'\left(0\right)\mathbf{1}+\varphi\left(h_{\varphi,R}\right)+\pi\left(h_{\pi,R}\right)+K_{1}\right)\,,\label{eq:explicit_rel_mod}
\end{equation}
where
\begin{eqnarray}
h_{\varphi,R}\left(\bar{x}\right)\!\!\!\! & := & \!\!\!\!\left.\frac{\mathrm{d}}{\mathrm{d}s}\right|_{s=0}g_{\varphi,R}^{s}\left(\bar{x}\right)=\left(x^{1}\frac{\partial^{2}F}{\left(\partial x^{0}\right)^{2}}\left(0,\bar{x}\right)+\frac{\partial F}{\partial x^{1}}\left(0,\bar{x}\right)\right)\cdot\Theta\left(x^{1}\right)\,,\\
h_{\pi,R}\left(\bar{x}\right)\!\!\!\! & := & \!\!\!\!\left.\frac{\mathrm{d}}{\mathrm{d}s}\right|_{s=0}g_{\pi,R}^{s}\left(\bar{x}\right)=\left(-x^{1}\frac{\partial F}{\partial x^{0}}\left(0,\bar{x}\right)\right)\cdot\Theta\left(x^{1}\right)\,.\label{eq:functions_explicit_rel_mod}
\end{eqnarray}
With similar arguments used above, we have that $h_{\varphi,R}\in K_{\varphi}\left(\Sigma\right)$
and $h_{\pi,R}\in K_{\pi}\left(\Sigma\right)$.\footnote{An explicit computation of the strong derivative in equation \eqref{eq:explicit_rel_mod}
shows that the vacuum vector $\Omega$, any coherent vector and any
vector of finite number of particles belong to the domain of $K_{\Psi,\Omega}$.}

Before we proceed to obtain the constant $\alpha'\left(0\right)$,
we emphasize its importance,
\begin{equation}
S_{R}\left(\omega_{f}\mid\omega\right)=\left\langle \Omega,K_{\Phi,\Omega}\Omega\right\rangle _{\mathcal{H}}=2\pi\alpha'\left(0\right)\,.\label{srel_alpha}
\end{equation}
Thus, the constant $\alpha'\left(0\right)$ gives the desired result
for the relative entropy. Regardless of the problem of computing the
value of $\alpha'\left(0\right)$, expressions (\ref{eq:explicit_rel_mod}-\ref{eq:functions_explicit_rel_mod})
gives us an explicit exact expression for the relative modular Hamiltonian
$K_{\Phi,\Omega}$ up to a constant. It is interesting to notice that
the difference $K_{\Phi,\Omega}-K_{\Omega}$ is just a linear term
on the fields operators plus a constant term. We expect that this
structure holds not just for the Rindler wedge, but for any kind of
region as long as $\Phi=W\left(f\right)\Omega$ is a coherent vector.

\subsubsection{Determination of $\alpha'\left(0\right)$ and the relative entropy\label{subsec:determ_alpha}}

As we have already explained in equation \eqref{srel_alpha}, we need
to determine the constant $\alpha'\left(0\right)$. Most of the calculation
is straightforward and we present the detailed computations in Appendix
\ref{subsec:appendix_3}. As in theorem \ref{th_mod_rel_coh}, throughout
this section we set $s:=-2\pi t$.

We start taking the vacuum expectation value on both sides in expression
\eqref{eq:concatenar},
\begin{equation}
\left\langle \Omega,\Delta_{\Psi,\Omega}^{it_{1}}\Delta_{\Psi,\Omega}^{it_{2}}\Omega\right\rangle _{\mathcal{H}}=\left\langle \Omega,\Delta_{\Psi,\Omega}^{i\left(t_{1}+t_{2}\right)}\Omega\right\rangle _{\mathcal{H}}\,,\label{eq:concatenar2}
\end{equation}
and we replace the expression \eqref{rel_mod_flow_coh} obtained for
the relative modular flow (see equations \ref{apen_conca_1}-\ref{apen_conca_2}).
Applying $\left.\frac{\mathrm{d}}{\mathrm{d}s_{1}}\right|_{s_{1}=0}=-\frac{1}{2\pi}\left.\frac{\mathrm{d}}{\mathrm{d}t_{1}}\right|_{t_{1}=0}$
on both sides of \eqref{eq:concatenar2} (equations \ref{apen_der_1}-\ref{apen_der_2}),\footnote{Analytic properties of the relative modular flow ensures that both
sides of \eqref{eq:concatenar2} are continuous differentiable functions
on $t_{1}$ and $t_{2}$. } and matching real and imaginary parts separately we get\footnote{The $\frac{\mathrm{d}}{\mathrm{d}s_{2}}$ in \eqref{eq:ec_dif_alpha}
appears because in some terms the dependance on $s_{1}$ of the expression
is through $s_{1}+s_{2}$. }

\begin{eqnarray}
\alpha'\left(s_{2}\right)-\frac{\mathrm{d}}{\mathrm{d}s_{2}}\mathrm{Im}\left\langle g_{\varphi,R}^{s_{2}},g_{\pi,R}^{s_{2}}\right\rangle _{\mathfrak{H}}\!\!\!\! & = & \!\!\!\!\alpha'\left(0\right)-\left.\frac{\mathrm{d}}{\mathrm{d}s_{1}}\right|_{s_{1}=0}\mathrm{Im}\left\langle g_{R}^{s_{1}},g_{R}^{s_{2}}\right\rangle _{\mathfrak{H}}\,,\label{eq:ec_dif_alpha}\\
\nonumber \\
\left.\frac{\mathrm{d}}{\mathrm{d}s_{1}}\right|_{s_{1}=0}\left\Vert g_{R}^{s_{1}+s_{2}}\right\Vert _{\mathfrak{H}}^{2}\!\!\!\! & = & \!\!\!\!\left.\frac{\mathrm{d}}{\mathrm{d}s_{1}}\right|_{s_{1}=0}\left\Vert g_{R}^{s_{1}}+u\left(\Lambda_{1}^{s_{1}},0\right)g_{R}^{s_{2}}\right\Vert _{\mathfrak{H}}^{2}\,,\label{eq:obvia}
\end{eqnarray}
where $g_{R}^{s}=E_{\varphi}\left(g_{\varphi,R}^{s}\right)+E_{\pi}\left(g_{\pi,R}^{s}\right)$.
The second equation is useless to determine $\alpha'\left(0\right)$,
then we concentrate in the first one which is a differential equation
for $\alpha'\left(s\right)$, with the particularity that $\alpha'\left(0\right)$
appears on it. To solve it, let us analyze the second term on the
right-hand side of equation \eqref{eq:ec_dif_alpha}. In Appendix
\ref{subsec:appendix_3} we compute 
\begin{eqnarray}
2\mathrm{Im}\left\langle g_{R}^{s_{1}},g_{R}^{s_{2}}\right\rangle _{\mathfrak{H}}\!\!\!\! & = & \!\!\!\!2\mathrm{Im}\left\langle g_{\varphi,R}^{s_{1}}+g_{\pi,R}^{s_{1}},g_{\varphi,R}^{s_{2}}+g_{\pi,R}^{s_{2}}\right\rangle _{\mathfrak{H}}\nonumber \\
 & = & \!\!\!\!\underset{:=P\left(s_{1}\right)}{\underbrace{\int_{x^{1}>0}\negthickspace f_{\varphi}\left(\bar{x}\right)f_{\pi}^{s_{1}}\left(\bar{x}\right)\:d^{d-1}x-\int_{x^{1}>0}\negthickspace f_{\varphi}^{s_{1}}\left(\bar{x}\right)f_{\pi}\left(\bar{x}\right)\:d^{d-1}x}}\nonumber \\
 &  & \!\!\!\!+\underset{:=Q\left(s_{1},s_{2}\right)}{\underbrace{\int_{x^{1}>0}\negthickspace f_{\varphi}^{s_{1}}\left(\bar{x}\right)f_{\pi}^{s_{2}}\left(\bar{x}\right)\:d^{d-1}x}}-\underset{:=R\left(s_{1},s_{2}\right)}{\underbrace{\int_{x^{1}>0}\negthickspace f_{\varphi}^{s_{2}}\left(\bar{x}\right)f_{\pi}^{s_{1}}\left(\bar{x}\right)\:d^{d-1}x}}+\gamma\left(s_{2}\right)\:\textrm{,}\label{eq:QyR}
\end{eqnarray}
The function $\gamma$ includes all the $s_{1}$-independent terms,
which they do not contribute to \eqref{eq:ec_dif_alpha}. In the same
appendix we analyze $P,\,Q,\,R$ carefully and we get

\begin{equation}
\left.\frac{\mathrm{d}P}{\mathrm{d}s_{1}}\right|_{s_{1}=0}=\int_{x^{1}\geq0}\negthickspace d^{d-1}x\,x^{1}\left.\left(\left(\frac{\partial F}{\partial x^{0}}\right)^{2}+\left(\nabla F\right)^{2}+m^{2}F^{2}\right)\right|_{x^{0}=0}=:\boldsymbol{S}\:,\label{eq:cool_relation_1}
\end{equation}
\begin{equation}
\left.\frac{\mathrm{d}}{\mathrm{d}s_{1}}\left(Q-R\right)\right|_{s_{1}=0}=-\left.\frac{\mathrm{d}}{\mathrm{d}s_{2}}\left(Q-R\right)\right|_{s_{1}=0}\,.\label{eq:cool_relation_2}
\end{equation}
Coming back to \eqref{eq:ec_dif_alpha}, we have that
\begin{eqnarray}
\alpha'\left(s_{2}\right)-\frac{\mathrm{d}}{\mathrm{d}s_{2}}\mathrm{Im}\left\langle g_{\varphi,R}^{s_{2}},g_{\pi,R}^{s_{2}}\right\rangle _{\mathfrak{H}}\!\!\!\! & = & \!\!\!\!\alpha'\left(0\right)-\left.\frac{\mathrm{d}}{\mathrm{d}s_{1}}\right|_{s_{1}=0}\mathrm{Im}\left\langle g_{R}^{s_{1}},g_{R}^{s_{2}}\right\rangle _{\mathfrak{H}}\nonumber \\
 & = & \!\!\!\!\alpha'\left(0\right)-\frac{1}{2}\left.\frac{\mathrm{d}}{\mathrm{d}s_{1}}\right|_{s_{1}=0}\left(P(s_{1})+Q\left(s_{1},s_{2}\right)-R\left(s_{1},s_{2}\right)\right)\\
 & = & \!\!\!\!\alpha'\left(0\right)-\frac{1}{2}\boldsymbol{S}+\frac{1}{2}\frac{\mathrm{d}}{\mathrm{d}s_{2}}\left(Q\left(0,s_{2}\right)-R\left(0,s_{2}\right)\right)\,.\nonumber 
\end{eqnarray}
Then, integrating this last equation with respect to $s_{2}$ we have
\begin{equation}
\alpha\left(s_{2}\right)-\mathrm{Im}\left\langle g_{\varphi,R}^{s_{2}},g_{\pi,R}^{s_{2}}\right\rangle _{\mathfrak{H}}=\alpha'\left(0\right)\,s_{2}-\frac{1}{2}\boldsymbol{S}\,s_{2}+\frac{1}{2}\left(Q\left(0,s_{2}\right)-R\left(0,s_{2}\right)\right)\,,\label{eq:alpha}
\end{equation}
where we have used $g_{\varphi,R}^{s_{2=0}}=g_{\pi,R}^{s_{2=0}}=0\Rightarrow\mathrm{Im}\left\langle g_{\varphi,R}^{s_{2}=0},g_{\pi,R}^{s_{2}=0}\right\rangle _{\mathfrak{H}}=0$,
and $Q\left(0,0\right)-R\left(0,0\right)=0$ which follows from the
definitions of $Q$ and $R$. To determine $\alpha'\left(0\right)$,
we use the KMS-condition stated in theorem \ref{thm_kms}. Using
$A=B=\boldsymbol{1}$ in equation \eqref{kms_func} and simply calling
$G\left(z\right)$ to the underlying function, we have that
\begin{eqnarray}
G\left(t\right)=\left\langle \Omega,\Delta_{\Psi,\Omega}^{it}\Omega\right\rangle _{\mathcal{H}} & \underset{t\rightarrow-i}{\longrightarrow} & G\left(-i\right)=\left\langle \Phi,\Phi\right\rangle _{\mathcal{H}}=1\,.\label{eq:periodicidad}
\end{eqnarray}
In terms of the real variable $s=-2\pi t$, the function $G\left(s\right)$
is in analytic on $\mathbb{R}+i\left(0,2\pi\right)$, and relation
\eqref{eq:periodicidad} must hold for $s\rightarrow2\pi i$. Using
\eqref{rel_mod_flow_coh}, we have that
\begin{equation}
G\left(s\right)=\mathrm{e}^{i\alpha\left(s\right)}\left\langle \Omega,\mathrm{e}^{i\varphi\left(g_{\varphi,R}^{s}\right)}\mathrm{e}^{i\pi\left(g_{\pi,R}^{s}\right)}\Omega\right\rangle _{\mathcal{H}}=\mathrm{e}^{i\alpha\left(s\right)-i\mathrm{Im}\left\langle g_{\varphi,R}^{s},g_{\pi,R}^{s}\right\rangle _{\mathfrak{H}}-\frac{1}{2}\left\Vert g_{R}^{s}\right\Vert _{\mathfrak{H}}^{2}}\,,\label{eq:kms_explicita}
\end{equation}
and hence
\begin{eqnarray}
i\alpha\left(s\right)-i\mathrm{Im}\left\langle g_{\varphi,R}^{s},g_{\pi,R}^{s}\right\rangle _{\mathfrak{H}}-\frac{1}{2}\left\Vert g_{R}^{s}\right\Vert _{\mathfrak{H}}^{2} & \underset{s\rightarrow2\pi i}{\longrightarrow} & i2n\pi\,,\quad n\in\mathbb{Z}\,.\label{eq:n_molesto}
\end{eqnarray}
Taking this into account, we come back to \eqref{eq:alpha} and write
\begin{equation}
i\alpha\left(s\right)-i\mathrm{Im}\left\langle g_{\varphi,R}^{s},g_{\pi,R}^{s}\right\rangle _{\mathfrak{H}}-\frac{1}{2}\left\Vert g_{R}^{s}\right\Vert _{\mathfrak{H}}^{2}=i\alpha'\left(0\right)\,s-\frac{i}{2}\boldsymbol{S}\,s+\frac{i}{2}\left(Q\left(0,s\right)-R\left(0,s\right)\right)-\frac{1}{2}\left\Vert g_{R}^{s}\right\Vert _{\mathfrak{H}}^{2}\,.\label{eq: tomar_limite}
\end{equation}
Before we take limit $s\rightarrow2\pi i$, we may notice that $\bar{x}^{s}=\left(-x^{1}\sinh\left(s\right),x^{1}\cosh\left(s\right),\bar{x}_{\bot}\right)\underset{s\rightarrow2\pi i}{\longrightarrow}\left(0,\bar{x}\right)$,
which informally suggests that
\begin{eqnarray}
g_{R}^{s}\underset{s\rightarrow2\pi i}{\longrightarrow}0\!\!\! & \Longrightarrow & \!\!\!\left\Vert g_{R}^{s}\right\Vert _{\mathfrak{H}}^{2}\underset{s\rightarrow2\pi i}{\longrightarrow}0\,,\\
f_{\nu}^{s}\underset{s\rightarrow2\pi i}{\longrightarrow}f_{\nu}\!\!\! & \Longrightarrow & \!\!\!Q\left(0,s\right)-R\left(0,s\right)\underset{s\rightarrow2\pi i}{\longrightarrow}0\,,\quad\textrm{where }\nu=\varphi,\pi\,.\label{eq: is_zero}
\end{eqnarray}
We prove in Appendix \ref{subsec:Analytic-continuation} that the
function
\begin{equation}
N\left(s\right):=\frac{i}{2}\left(Q\left(0,s\right)-R\left(0,s\right)\right)-\frac{1}{2}\left\Vert g_{R}^{s}\right\Vert _{\mathfrak{H}}^{2}\textrm{ ,}
\end{equation}
of the variable $s\in\mathbb{R}$, can be analytically continued on
the strip $\mathbb{R}+i\left(0,2\pi\right)$ and that $\lim_{s\rightarrow2\pi i}N\left(s\right)=0$.
Then, taking the limit $s\rightarrow2\pi i$ on \eqref{eq: tomar_limite}
we get
\begin{equation}
i2n\pi=-\alpha'\left(0\right)\,2\pi+\frac{1}{2}\boldsymbol{S}\,2\pi\,.
\end{equation}
Since $\alpha'\left(0\right),\boldsymbol{S}\in\mathbb{R}$ then it
must be $n=0$, an hence we finally get $\alpha'\left(0\right)=\frac{1}{2}\boldsymbol{S}$. 

All these together can be summarized in the following theorem which
generalizes the theorem \ref{thm_sr_coh_par}.
\begin{thm}
\textup{Let $f\in\mathcal{S}\left(\mathbb{R}^{d},\mathbb{R}\right)$
with $supp\left(f\right)$ compact. Then the relative entropy between
the coherent state $\omega_{f}$ and the vacuum $\omega$, for the
right Rindler wedge algebra $\mathcal{R}_{\mathcal{W}}$, is
\begin{equation}
S_{R}\left(\omega_{f}\mid\omega\right)=2\pi\int_{x^{1}>0}\negthickspace d^{d-1}x\,x^{1}\left.\frac{1}{2}\left(\left(\frac{\partial F}{\partial x^{0}}\right)^{2}+\left|\nabla F\right|^{2}+m^{2}F^{2}\right)\right|_{x^{0}=0},\label{sr_coh_final}
\end{equation}
where $F\left(x\right)=\int_{\mathbb{R}^{d}}\Delta\left(x-y\right)\,f\left(y\right)d^{d}y$
. In addition, formula \eqref{sr_coh_final} only depends in the behavior
of $f$ in $\mathbb{R}^{d}-\mathcal{W}'$.}
\end{thm}

\section*{Acknowledgments}

We acknowledge discussions with Stefan Hollands and Pedro Martinez. We also
thank Thomas Faulkner for his useful comments on an earlier version
of the manuscript. This work was partially supported by CONICET, CNEA
and Universidad Nacional de Cuyo, Argentina. H.C. acknowledges support
from an It From Qubit grant of the Simons Foundation.

\appendix

\section{Appendix\label{sec:Appendix}}

\subsection{Sobolev spaces\label{subsec:appendix_Sobolev}}

For the definition and properties of Sobolev spaces, we follow \cite{Evans}.
Here we adapt the notation to our convenience.

Consider the test function space $\mathcal{D}\left(\mathbb{R}^{n}\right):=C_{c}^{\infty}\left(\mathbb{R}^{n}\right)\varsubsetneq\mathcal{S}\left(\mathbb{R}^{n}\right)$
of smooth and compactly supported functions, with its usual topology.
The $n$-dimensional complex \textit{Sobolev space} of order $\alpha\in\mathbb{R}$
is defined as 
\begin{equation}
H^{\alpha}\left(\mathbb{R}^{n}\right):=\left\{ f\in\mathcal{D}'\left(\mathbb{R}^{n}\right)\,:\,\hat{f}\left(\bar{p}\right)\omega_{\bar{p}}^{\alpha}\in L^{2}\left(\mathbb{R}^{n}\right)\right\} \,,
\end{equation}
where $\omega_{\bar{p}}=\sqrt{\bar{p}^{2}+1}$ and $\hat{f}\left(\bar{p}\right):=\left(2\pi\right)^{-\frac{n}{2}}\int_{\mathbb{R}^{n}}f\left(\bar{x}\right)\mathrm{e}^{-i\bar{p}\cdot\bar{x}}d^{n}x$
is the usual Fourier transform. From the definition follows that $H^{0}\left(\mathbb{R}^{n}\right)=L^{2}\left(\mathbb{R}^{n}\right)$
and $H^{\alpha}\left(\mathbb{R}^{n}\right)\subset H^{\alpha'}\left(\mathbb{R}^{n}\right)$
if $\alpha>\alpha'$.

The Sobolev space $H^{\alpha}\left(\mathbb{R}^{n}\right)$ is a Hilbert
space under the inner product 
\begin{equation}
\left\langle f,g\right\rangle _{H^{\alpha}}:=\left\langle \hat{f}\omega_{\bar{p}}^{\alpha},\hat{g}\omega_{\bar{p}}^{\alpha}\right\rangle _{L^{2}}=\int_{\mathbb{R}^{n}}d^{n}p\,\hat{f}\left(\bar{p}\right)^{*}\hat{g}\left(\bar{p}\right)\omega_{\bar{p}}^{2\alpha}\,.
\end{equation}
Furthermore, for $f\in H^{\alpha}\left(\mathbb{R}^{n}\right)$ we
have that $\left\Vert f\right\Vert _{H^{\alpha'}}\leq\left\Vert f\right\Vert _{H^{\alpha}}$
if $\alpha>\alpha'$, and hence the natural injections $H^{\alpha}\left(\mathbb{R}^{n}\right)\hookrightarrow H^{\alpha'}\left(\mathbb{R}^{n}\right)$
for $\alpha>\alpha'$ are continuous. We also have that the set $C^{\infty}\left(\mathbb{R}^{n}\right)\subset\mathcal{S}\left(\mathbb{R}^{n}\right)$
is dense in $H^{\alpha}\left(\mathbb{R}^{n}\right)$.

When $\alpha=k\in\mathbb{N}_{0}$, there is also another useful equivalent
characterization of the Sobolev spaces in term of weak derivatives\footnote{The weak derivative of an element of $\mathcal{D}'\left(\mathbb{R}^{n}\right)$
is its usual derivative in the distributional sense.}

\begin{equation}
H^{k}\left(\mathbb{R}^{n}\right)=\left\{ f\in\mathcal{D}'\left(\mathbb{R}^{n}\right)\,:\,D^{\mu}f\in L^{2}\left(\mathbb{R}^{n}\right)\,,\textrm{ for all }\left|\mu\right|\leq k\right\} \,.\label{eq: sobolev_2}
\end{equation}
It is useful to introduce a new norm in $H^{k}\left(\mathbb{R}^{n}\right)$
as 
\begin{equation}
\left\Vert f\right\Vert '_{H^{k}}:=\left(\sum_{\left|\mu\right|\leq k}\int_{\mathbb{R}^{n}}d^{n}x\left|D^{\mu}f\left(x\right)\right|^{2}\right)^{\frac{1}{2}}\,,
\end{equation}
which is equivalent to the former norm $\left\Vert \cdot\right\Vert {}_{H^{k}}$.

The real Sobolev spaces $H^{\alpha}\left(\mathbb{R}^{n},\mathbb{R}\right)$
are defined in a similar manner as above, but restricting to real-valued functions.\\

In general, it is easier to calculate the usual pointwise derivatives
rather than the weak derivatives. Then, the following lemma states
sufficient conditions for both notions of derivatives coincide. Before
we formulate it, we need to introduce the notions of \textit{$C^{k}$}-piecewise
function. 
\begin{defn}
Let $U\subset\mathbb{R}^{n}$ open, $f\in L_{loc}^{1}\left(U\right)$
and $k\in\mathbb{N}_{0}$. We say that $f$ is a \textit{$C^{k}$-piecewise}
function iff there exists a finite family of pairwise disjoint open
sets $\left\{ \Omega_{j}\right\} _{j=1,\ldots,J}\subset U$ such that
\end{defn}
\begin{enumerate}
\item $\bigcup_{j=1}^{J}\overline{\Omega}_{j}=\overline{U}$. 
\item $f\in C^{k}\left(\Omega_{j}\right)$ for all $j=1,\ldots,J$. 
\item For all $j=1,\ldots,J$, $\forall x_{0}\in\partial\Omega_{j}$ and
for all multi-index $\left|\alpha\right|\leq k$, the $\lim_{x\rightarrow x_{0}}\left.D^{\alpha}f\left(x\right)\right|_{\Omega_{j}}$
exist and are finite (where $D^{\alpha}$ is the usual multiorder
pointwise derivative). 
\end{enumerate}
We denote by $C_{t}^{k}\left(U\right)$ the set of \textit{$C^{k}$}-piecewise
functions on $U$.\\

Now, we formulate the lemma that ensures that weak derivatives and
pointwise derivatives coincide. 
\begin{lem}
\label{lem:derivatives}

\textup{Let $U\subset\mathbb{R}^{n}$ be open and $f\in C^{0}\left(U\right)\cap C_{t}^{1}\left(U\right)$.
Then the (first order) weak derivatives of $f$ coincides with the
usual pointwise derivatives. } 
\end{lem}
\begin{proof}
Since $f\in C^{0}\left(U\right)\cap C_{t}^{1}\left(U\right)$ we have
that $f$ is locally Lipschitz continuous on $U$ (see corollary 4.1.1
on \cite{scholtes}). Then we have that $f$ is locally absolute continuous
on $U$, and of course $f\in L_{loc}^{1}\left(U\right)$. Then $f$
is weakly differentiable and the (first order) weak and pointwise
derivatives of $f$ coincide a.e. 
\end{proof}
Now, using the above lemma and the alternative definition (eq. \ref{eq: sobolev_2})
for the Sobolev space $H^{1}\left(\mathbb{R}^{n}\right)$, the proof
in lemma \ref{par:lemma-1} is trivial.

\subsection{Calculation of $\left\langle \Omega,W\left(f_{R}\right)^{*}K_{1}W\left(f_{R}\right)\Omega\right\rangle $\label{subsec:appendix_coherent}}

Take $f_{R}\in\mathcal{S}\left(\mathbb{R}^{d},\mathbb{R}\right)$
and for simplicity call $f:=f_{R}$. Then 

\begin{eqnarray}
\left\langle \Omega,W\left(f\right)^{*}K_{1}W\left(f\right)\Omega\right\rangle _{\mathcal{H}}\!\!\!\! & = & \!\!\!\!\left\langle \mathrm{e}^{-\frac{\left\Vert f\right\Vert _{\mathfrak{H}}^{2}}{2}}\sum_{n=0}^{\infty}\frac{if^{\otimes n}}{\sqrt{n!}},K_{1}\mathrm{e}^{-\frac{\left\Vert f\right\Vert _{\mathfrak{H}}^{2}}{2}}\sum_{n=0}^{\infty}\frac{if^{\otimes n}}{\sqrt{n!}}\right\rangle _{\mathcal{H}}\\
 & = & \!\!\!\!\mathrm{e}^{-\left\Vert f\right\Vert _{\mathfrak{H}}^{2}}\left\langle \sum_{n=0}^{\infty}\frac{if^{\otimes n}}{\sqrt{n!}},K_{1}\sum_{n=0}^{\infty}\frac{if^{\otimes n}}{\sqrt{n!}}\right\rangle _{\mathcal{H}}=\mathrm{e}^{-\left\Vert f\right\Vert _{\mathfrak{H}}^{2}}\sum_{n=0}^{\infty}\frac{\left(-i\right)^{n}\left(i\right)^{n}}{n!}\left\langle f^{\otimes n},K_{1}f^{\otimes n}\right\rangle _{\mathfrak{H}^{\otimes n}}\nonumber \\
 & = & \!\!\!\!\mathrm{e}^{-\left\Vert f\right\Vert _{\mathfrak{H}}^{2}}\sum_{n=0}^{\infty}\frac{n}{n!}\left\langle f^{\otimes n},\left(k_{1}f\right)\otimes f^{\otimes n-1}\right\rangle _{\mathfrak{H}^{\otimes n}}=\mathrm{e}^{-\left\Vert f\right\Vert _{\mathfrak{H}}^{2}}\sum_{n=1}^{\infty}\frac{1}{\left(n-1\right)!}\left\langle f,k_{1}f\right\rangle _{\mathfrak{H}}\left\langle f,f\right\rangle _{\mathfrak{H}}^{n-1}\nonumber \\
 & = & \!\!\!\!\mathrm{e}^{-\left\Vert f\right\Vert _{\mathfrak{H}}^{2}}\left\langle f,k_{1}f\right\rangle _{\mathfrak{H}}\sum_{n=1}^{\infty}\frac{1}{\left(n-1\right)!}\left\langle f,f\right\rangle _{\mathfrak{H}}^{n-1}=\mathrm{e}^{-\left\Vert f\right\Vert _{\mathfrak{H}}^{2}}\left\langle f,k_{1}f\right\rangle _{\mathfrak{H}}\mathrm{e}^{\left\Vert f\right\Vert _{\mathfrak{H}}^{2}}=\left\langle f,k_{1}f\right\rangle _{\mathfrak{H}}.
\end{eqnarray}

\subsection{Calculation of $\left\langle f_{R},k_{1}f_{R}\right\rangle _{\mathfrak{H}}$\label{subsec:appendix_2}}

Take $f_{R}\in\mathcal{S}\left(\mathbb{R}^{d},\mathbb{R}\right)$
and for simplicity call $f:=f_{R}$. Then 
\begin{eqnarray}
\left\langle f,k_{1}f\right\rangle _{\mathfrak{H}}\!\!\!\! & = & \!\!\!\!\mathrm{Re}\left\langle f,k_{1}f\right\rangle _{\mathfrak{H}}=\mathrm{Re}\left(-i\left.\frac{\mathrm{d}}{\mathrm{d}s}\right|_{s=0}\left\langle f,\mathrm{e}^{ik_{1}s}f\right\rangle _{\mathfrak{H}}\right)=\left.\frac{\mathrm{d}}{\mathrm{d}s}\right|_{s=0}\mathrm{Im}\left\langle f,u\left(\Lambda_{1}^{s},0\right)f\right\rangle _{\mathfrak{H}}=\nonumber \\
 & = & \!\!\!\!\left.\frac{\mathrm{d}}{\mathrm{d}s}\right|_{s=0}\mathrm{Im}\left\langle f,f_{\left(\Lambda_{1}^{s},0\right)}\right\rangle _{\mathfrak{H}}=\left.\frac{\mathrm{d}}{\mathrm{d}s}\right|_{s=0}\mathrm{Im}\left\langle f,f^{s}\right\rangle _{\mathfrak{H}}\label{eq:evqm_1}
\end{eqnarray}
where we have defined $f^{s}=f_{\left(\Lambda_{1}^{s},0\right)}$.
As we explained in section \ref{subsec:Relation-between-algebras},
there exist functions $f_{\varphi},\,f_{\pi},\,f_{\varphi}^{s},\,f_{\pi}^{s}\in\mathcal{S}\left(\mathbb{R}^{d-1},\mathbb{R}\right)$
such that
\begin{eqnarray}
E\left(f\right)=E_{\varphi}\left(f_{\varphi}\right)+E_{\pi}\left(f_{\pi}\right) & \textrm{ and } & E\left(f^{s}\right)=E_{\varphi}\left(f_{\varphi}^{s}\right)+E_{\pi}\left(f_{\pi}^{s}\right)\:\textrm{.}
\end{eqnarray}
Replacing these in \eqref{eq:evqm_1} we get
\begin{eqnarray}
\left\langle f,k_{1}f\right\rangle _{\mathfrak{H}}\!\!\!\! & = & \!\!\!\!\left.\frac{\mathrm{d}}{\mathrm{d}s}\right|_{s=0}\mathrm{Im}\left\langle f_{\varphi}+f_{\pi},f_{\varphi}^{s}+f_{\pi}^{s}\right\rangle _{\mathfrak{H}}=\left.\frac{\mathrm{d}}{\mathrm{d}s}\right|_{s=0}\left(\mathrm{Im}\left\langle f_{\varphi},f_{\pi}^{s}\right\rangle _{\mathfrak{H}}+\mathrm{Im}\left\langle f_{\pi},f_{\varphi}^{s}\right\rangle _{\mathfrak{H}}\right)\nonumber \\
 & = & \!\!\!\!\left.\frac{\mathrm{d}}{\mathrm{d}s}\right|_{s=0}\left(\frac{1}{2}\int_{\mathbb{R}^{d-1}}f_{\varphi}\left(\bar{x}\right)\,f_{\pi}^{s}\left(\bar{x}\right)d^{d-1}x-\mathrm{\frac{1}{2}}\int_{\mathbb{R}^{d-1}}f_{\varphi}^{s}\left(\bar{x}\right)\,f_{\pi}\left(\bar{x}\right)d^{d}x\right)\,,\label{eq:evqm_2}
\end{eqnarray}
where we have used the relations \eqref{eq:rel_prod_int} in the second line
and \eqref{eq:img_pe} in the last line. From the Poincaré invariance
of the distribution $\Delta\left(x\right)$ we have that
\begin{equation}
F^{s}\left(x\right)=\int_{\mathbb{R}^{d}}\Delta\left(x-y\right)\,f^{s}\left(x\right)d^{d}y\,\textrm{,}
\end{equation}
where we have defined $F^{s}:=F_{\left(\Lambda_{1}^{s},0\right)}$.
Then
\begin{eqnarray}
f_{\varphi}\left(\bar{x}\right)\!\!\!\! & := & \!\!\!\!-\frac{\partial F}{\partial x^{0}}\left(0,\bar{x}\right)\,,\\
f_{\pi}\left(\bar{x}\right)\!\!\!\! & := & \!\!\!\!F\left(0,\bar{x}\right)\,,\\
f_{\varphi}^{s}\left(\bar{x}\right)\!\!\!\! & := & \!\!\!\!-\cosh\left(s\right)\frac{\partial F}{\partial x^{0}}\left(\bar{x}^{s}\right)+\sinh\left(s\right)\frac{\partial F}{\partial x^{1}}\left(\bar{x}^{s}\right)\,,\\
f_{\pi}^{s}\left(\bar{x}\right)\!\!\!\! & := & \!\!\!\!F\left(\bar{x}^{s}\right)\,,
\end{eqnarray}
being $\bar{x}^{s}:=\left(-x^{1}\sinh\left(s\right),x^{1}\cosh\left(s\right),x_{\bot}\right)$,
and hence 
\begin{eqnarray}
\left.\frac{\mathrm{d}}{\mathrm{d}s}\right|_{s=0}f_{\varphi}^{s}\left(\bar{x}\right)\!\!\!\! & := & \!\!\!\!x^{1}\frac{\partial^{2}F}{\left(\partial x^{0}\right)^{2}}\left(0,\bar{x}\right)+\frac{\partial F}{\partial x^{1}}\left(0,\bar{x}\right)\,,\\
\left.\frac{\mathrm{d}}{\mathrm{d}s}\right|_{s=0}f_{\pi}^{s}\left(\bar{x}\right)\!\!\!\! & := & \!\!\!\!-x^{1}\frac{\partial F}{\partial x^{0}}\left(0,\bar{x}\right)\:\textrm{.}
\end{eqnarray}
Replacing such expressions in \eqref{eq:evqm_2}, using the equation
of motion for $F$ and doing an integration by parts, we finally get
\begin{equation}
\left\langle f,k_{1}f\right\rangle _{\mathfrak{H}}=\int_{\mathbb{R}^{d-1}}d^{d-1}x\,x^{1}\,\frac{1}{2}\left.\left(\left(\frac{\partial F}{\partial x^{0}}\right)^{2}+\left|\nabla F\right|^{2}+m^{2}F^{2}\right)\right|_{x^{0}=0}\textrm{.}\label{apend_sr_easy}
\end{equation}

\subsection{Explicit computations of section \ref{subsec:determ_alpha}\label{subsec:appendix_3}}

Defining $g_{R}^{s}=E_{\varphi}\left(g_{\varphi,R}^{s}\right)+E_{\pi}\left(g_{\pi,R}^{s}\right)\in\mathfrak{H}$,
\begin{eqnarray}
\left\langle \Omega,\Delta_{\Psi,\Omega}^{it_{1}}\Delta_{\Psi,\Omega}^{it_{2}}\Omega\right\rangle _{\mathcal{H}}\!\!\!\! & = & \!\!\!\!\left\langle \Omega,\mathrm{e}^{i\alpha\left(s_{1}\right)}W_{\varphi}\left(g_{\varphi,R}^{s_{1}}\right)W_{\pi}\left(g_{\pi,R}^{s_{1}}\right)\Delta_{\Omega}^{it_{1}}\mathrm{e}^{i\alpha\left(s_{2}\right)}W_{\varphi}\left(g_{\varphi,R}^{s_{2}}\right)W_{\pi}\left(g_{\pi,R}^{s_{2}}\right)\Delta_{\Omega}^{it_{1}}\Omega\right\rangle _{\mathcal{H}}\nonumber \\
 & = & \!\!\!\!\mathrm{e}^{i\alpha\left(s_{1}\right)+i\alpha\left(s_{2}\right)}\left\langle \Omega,W_{\varphi}\left(g_{\varphi,R}^{s_{1}}\right)W_{\pi}\left(g_{\pi,R}^{s_{1}}\right)\mathrm{e}^{is_{1}K_{1}}W_{\varphi}\left(g_{\varphi,R}^{s_{2}}\right)W_{\pi}\left(g_{\pi,R}^{s_{2}}\right)\Omega\right\rangle _{\mathcal{H}}\nonumber \\
 & = & \!\!\!\!\mathrm{e}^{i\alpha\left(s_{1}\right)+i\alpha\left(s_{2}\right)-i\mathrm{Im}\left\langle g_{\varphi,R}^{s_{1}},g_{\pi,R}^{s_{1}}\right\rangle -i\mathrm{Im}\left\langle g_{\varphi,R}^{s_{2}},g_{\pi,R}^{s_{2}}\right\rangle _{\mathfrak{H}}}\left\langle \Omega,W\left(g_{R}^{s_{1}}\right)\mathrm{e}^{is_{1}K_{1}}W\left(g_{R}^{s_{2}}\right)\Omega\right\rangle _{\mathcal{H}}\label{apen_conca_1}\\
 & = & \!\!\!\!\mathrm{e}^{i\alpha\left(s_{1}\right)+i\alpha\left(s_{2}\right)-i\mathrm{Im}\left\langle g_{\varphi,R}^{s_{1}},g_{\pi,R}^{s_{1}}\right\rangle -i\mathrm{Im}\left\langle g_{\varphi,R}^{s_{2}},g_{\pi,R}^{s_{2}}\right\rangle _{\mathfrak{H}}}\left\langle \Omega,W\left(g_{R}^{s_{1}}\right)\mathrm{e}^{is_{1}K_{1}}W\left(g_{R}^{s_{2}}\right)\mathrm{e}^{-is_{1}K_{1}}\Omega\right\rangle _{\mathcal{H}}\nonumber \\
 & = & \!\!\!\!\mathrm{e}^{i\alpha\left(s_{1}\right)+i\alpha\left(s_{2}\right)-i\mathrm{Im}\left\langle g_{\varphi,R}^{s_{1}},g_{\pi,R}^{s_{1}}\right\rangle -i\mathrm{Im}\left\langle g_{\varphi,R}^{s_{2}},g_{\pi,R}^{s_{2}}\right\rangle _{\mathfrak{H}}}\left\langle \Omega,W\left(g_{R}^{s_{1}}\right)W\left(u\left(\Lambda_{1}^{s_{1}}\right)g_{R}^{s_{2}}\right)\Omega\right\rangle _{\mathcal{H}}\nonumber \\
 & = & \!\!\!\!\mathrm{e}^{i\alpha\left(s_{1}\right)+i\alpha\left(s_{2}\right)-i\mathrm{Im}\left\langle g_{\varphi,R}^{s_{1}},g_{\pi,R}^{s_{1}}\right\rangle -i\mathrm{Im}\left\langle g_{\varphi,R}^{s_{2}},g_{\pi,R}^{s_{2}}\right\rangle _{\mathfrak{H}}-i\mathrm{Im}\left\langle g_{R}^{s_{1}},u\left(\Lambda_{1}^{s_{1}}\right)g_{R}^{s_{2}}\right\rangle _{\mathfrak{H}}}\left\langle \Omega,W\left(g_{R}^{s_{1}}+u\left(\Lambda_{1}^{s_{1}}\right)g_{R}^{s_{2}}\right)\Omega\right\rangle _{\mathcal{H}}\nonumber \\
 & = & \!\!\!\!\mathrm{e}^{i\alpha\left(s_{1}\right)+i\alpha\left(s_{2}\right)-i\mathrm{Im}\left\langle g_{\varphi,R}^{s_{1}},g_{\pi,R}^{s_{1}}\right\rangle -i\mathrm{Im}\left\langle g_{\varphi,R}^{s_{2}},g_{\pi,R}^{s_{2}}\right\rangle _{\mathfrak{H}}-i\mathrm{Im}\left\langle g_{R}^{s_{1}},u\left(\Lambda_{1}^{s_{1}}\right)g_{R}^{s_{2}}\right\rangle _{\mathfrak{H}}-\frac{1}{2}\left\Vert g_{R}^{s_{1}}+u\left(\Lambda_{1}^{s_{1}}\right)g_{R}^{s_{2}}\right\Vert _{\mathfrak{H}}^{2}}\,,\nonumber 
\end{eqnarray}
and
\begin{eqnarray}
\left\langle \Omega,\Delta_{\Psi,\Omega}^{i\left(t_{1}+t_{2}\right)}\Omega\right\rangle _{\mathcal{H}}\!\!\!\! & = & \!\!\!\!\left\langle \Omega,\mathrm{e}^{i\alpha\left(s_{1}+s_{2}\right)}W_{\varphi}\left(g_{\varphi,R}^{s_{1}+s_{2}}\right)W_{\pi}\left(g_{\pi,R}^{s_{1}+s_{2}}\right)\Delta_{\Omega}^{i\left(t_{1}+t_{2}\right)}\Omega\right\rangle _{\mathcal{H}}\nonumber \\
 & = & \!\!\!\!\mathrm{e}^{i\alpha\left(s_{1}+s_{2}\right)}\left\langle \Omega,W_{\varphi}\left(g_{\varphi,R}^{s_{1}+s_{2}}\right)W_{\pi}\left(g_{\pi,R}^{s_{1}+s_{2}}\right)\Omega\right\rangle _{\mathcal{H}}\nonumber \\
 & = & \!\!\!\!\mathrm{e}^{i\alpha\left(s_{1}+s_{2}\right)-i\mathrm{Im}\left\langle g_{\varphi,R}^{s_{1}+s_{2}},g_{\pi,R}^{s_{1}+s_{2}}\right\rangle _{\mathfrak{H}}}\left\langle \Omega,W\left(g_{R}^{s_{1}+s_{2}}\right)\Omega\right\rangle _{\mathcal{H}}\label{apen_conca_2}\\
 & = & \!\!\!\!\mathrm{e}^{i\alpha\left(s_{1}+s_{2}\right)-i\mathrm{Im}\left\langle g_{\varphi,R}^{s_{1}+s_{2}},g_{\pi,R}^{s_{1}+s_{2}}\right\rangle _{\mathfrak{H}}-\frac{1}{2}\left\Vert g_{R}^{s_{1}+s_{2}}\right\Vert _{\mathfrak{H}}^{2}}\,.\nonumber 
\end{eqnarray}
Taking $\left.\frac{\mathrm{d}}{\mathrm{d}s_{1}}\right|_{s_{1}=0}$
on both expressions above we obtain,
\begin{eqnarray}
\left.\frac{\mathrm{d}}{\mathrm{d}s_{1}}\right|_{s_{1}=0}\left\langle \Omega,\Delta_{\Psi,\Omega}^{it_{1}}\Delta_{\Psi,\Omega}^{it_{2}}\Omega\right\rangle _{\mathcal{H}}\!\!\!\! & = & \!\!\!\!i\alpha'\left(0\right)-i\underset{=0}{\underbrace{\left.\frac{\mathrm{d}}{\mathrm{d}s_{1}}\right|_{s_{1}=0}\mathrm{Im}\left\langle g_{\varphi,R}^{s_{1}},g_{\pi,R}^{s_{1}}\right\rangle _{\mathfrak{H}}}}-i\left.\frac{\mathrm{d}}{\mathrm{d}s_{1}}\right|_{s_{1}=0}\mathrm{Im}\left\langle g_{R}^{s_{1}},u\left(\Lambda_{1}^{s_{1}}\right)g_{R}^{s_{2}}\right\rangle _{\mathfrak{H}}\nonumber \\
 &  & \!\!\!\!-\frac{1}{2}\left.\frac{\mathrm{d}}{\mathrm{d}s_{1}}\right|_{s_{1}=0}\left\Vert g_{R}^{s_{1}}+u\left(\Lambda_{1}^{s_{1}}\right)g_{R}^{s_{2}}\right\Vert _{\mathfrak{H}}^{2}\label{apen_der_1}\\
 & = & \!\!\!\!i\alpha'\left(0\right)-i\left.\frac{\mathrm{d}}{\mathrm{d}s_{1}}\right|_{s_{1}=0}\mathrm{Im}\left\langle g_{R}^{s_{1}},g_{R}^{s_{2}}\right\rangle _{\mathfrak{H}}-\frac{1}{2}\left.\frac{\mathrm{d}}{\mathrm{d}s_{1}}\right|_{s_{1}=0}\left\Vert g_{R}^{s_{1}}+u\left(\Lambda_{1}^{s_{1}}\right)g_{R}^{s_{2}}\right\Vert _{\mathfrak{H}}^{2}\,,\nonumber 
\end{eqnarray}
and
\begin{eqnarray}
\left.\frac{\mathrm{d}}{\mathrm{d}s_{1}}\right|_{s_{1}=0}\left\langle \Omega,\Delta_{\Psi,\Omega}^{i\left(t_{1}+t_{2}\right)}\Omega\right\rangle _{\mathcal{H}}\!\!\!\! & = & \!\!\!\!i\alpha'\left(s_{2}\right)-i\left.\frac{\mathrm{d}}{\mathrm{d}s_{1}}\right|_{s_{1}=0}\mathrm{Im}\left\langle g_{\varphi,R}^{s_{1}+s_{2}},g_{\pi,R}^{s_{1}+s_{2}}\right\rangle _{\mathfrak{H}}-\frac{1}{2}\left.\frac{\mathrm{d}}{\mathrm{d}s_{1}}\right|_{s_{1}=0}\left\Vert g_{R}^{s_{1}+s_{2}}\right\Vert _{\mathfrak{H}}^{2}\nonumber \\
 & = & \!\!\!\!i\alpha'\left(s_{2}\right)-i\frac{\mathrm{d}}{\mathrm{d}s_{2}}\mathrm{Im}\left\langle g_{\varphi,R}^{s_{2}},g_{\pi,R}^{s_{2}}\right\rangle _{\mathfrak{H}}-\frac{1}{2}\left.\frac{\mathrm{d}}{\mathrm{d}s_{1}}\right|_{s_{1}=0}\left\Vert g_{R}^{s_{1}+s_{2}}\right\Vert _{\mathfrak{H}}^{2}\:\textrm{.}\label{apen_der_2}
\end{eqnarray}
Matching real and imaginary parts of these two last expressions, we
arrive to formulas \eqref{eq:ec_dif_alpha} and \eqref{eq:obvia}.

Expressions \eqref{eq:QyR} follows from
\begin{eqnarray}
2\mathrm{Im}\left\langle g_{R}^{s_{1}},g_{R}^{s_{2}}\right\rangle _{\mathfrak{H}}\!\!\!\! & = & \!\!\!\!2\mathrm{Im}\left\langle g_{\varphi,R}^{s_{1}}+g_{\pi,R}^{s_{1}},g_{\varphi,R}^{s_{2}}+g_{\pi,R}^{s_{2}}\right\rangle _{\mathfrak{H}}\nonumber \\
 & = & \!\!\!\!\int_{\Sigma}\!d^{d-1}x\,g_{\varphi}^{s_{1}}\left(\bar{x}\right)g_{\pi}^{s_{2}}\left(\bar{x}\right)-\int_{\Sigma}\!d^{d-1}x\,g_{\varphi}^{s_{2}}\left(\bar{x}\right)g_{\pi}^{s_{1}}\left(\bar{x}\right)\nonumber \\
 & = & \!\!\!\!\int_{\Sigma}\!d^{d-1}x\left(f_{\varphi}^{s_{1}}\left(\bar{x}\right)-f_{\varphi}\left(\bar{x}\right)\right)\left(f_{\pi}^{s_{2}}\left(\bar{x}\right)-f_{\pi}\left(\bar{x}\right)\right)-\int_{\Sigma}\!d^{d-1}x\left(f_{\varphi}^{s_{2}}\left(\bar{x}\right)-f_{\varphi}\left(\bar{x}\right)\right)\left(f_{\pi}^{s_{1}}\left(\bar{x}\right)-f_{\pi}\left(\bar{x}\right)\right)\nonumber \\
 & = & \underset{:=P\left(s_{1}\right)}{\!\!\!\!\underbrace{\int_{\Sigma}\!d^{d-1}x\,f_{\varphi}\left(\bar{x}\right)f_{\pi}^{s_{1}}\left(\bar{x}\right)-\int_{\Sigma}\!d^{d-1}x\,f_{\varphi}^{s_{1}}\left(\bar{x}\right)f_{\pi}\left(\bar{x}\right)}}\nonumber \\
 &  & \!\!\!\!+\underset{:=Q\left(s_{1},s_{2}\right)}{\underbrace{\int_{\Sigma}\!d^{d-1}x\,f_{\varphi}^{s_{1}}\left(\bar{x}\right)f_{\pi}^{s_{2}}\left(\bar{x}\right)}}-\underset{:=R\left(s_{1},s_{2}\right)}{\underbrace{\int_{\Sigma}\!d^{d-1}x\,f_{\varphi}^{s_{2}}\left(\bar{x}\right)f_{\pi}^{s_{1}}\left(\bar{x}\right)}}+\gamma\left(s_{2}\right)\,,
\end{eqnarray}
where the function $\gamma$ includes all the $s_{1}$-independent
terms.

The function $P\left(s_{1}\right)$ is essentially the same as \eqref{eq:evqm_2}
in Appendix \ref{subsec:appendix_2}, with the difference that now
the integration is over the region $\Sigma=\left\{ \bar{x}\in\mathbb{R}^{d-1}\,:\,x^{1}\geq0\right\} $
instead of the whole $\mathbb{R}^{d-1}$. Despite this, the final
result is the same and hence we get\footnote{Following the computation of \eqref{eq:evqm_2} in Appendix \ref{subsec:appendix_2},
there now appears a boundary term after the integration by parts. Fortunately,
this term vanishes since the integrand is $0$ at the boundary of
$\Sigma$. }
\[
\left.\frac{\mathrm{d}P}{\mathrm{d}s_{1}}\right|_{s_{1}=0}=\int_{\Sigma}\!d^{d-1}x\,x^{1}\left.\left(\left(\frac{\partial F}{\partial x^{0}}\right)^{2}+\left(\nabla F\right)^{2}+m^{2}F^{2}\right)\right|_{x^{0}=0}=:\boldsymbol{S}\,.
\]
Now we explicitly obtain the relations \eqref{eq:cool_relation_2}.
Indeed,
\begin{eqnarray}
\left.\frac{\mathrm{d}R}{\mathrm{d}s_{1}}\right|_{s_{1}=0}\!\!\!\! & = & \!\!\!\!\left.\frac{\mathrm{d}}{\mathrm{d}s_{1}}\right|_{s_{1}=0}\int_{\Sigma}\!d^{d-1}x\,\left(-\cosh\left(s_{2}\right)\frac{\partial F}{\partial x^{0}}\left(\bar{x}^{s_{2}}\right)+\sinh\left(s_{2}\right)\frac{\partial F}{\partial x^{1}}\left(\bar{x}^{s_{2}}\right)\right)F\left(\bar{x}^{s_{1}}\right)\nonumber \\
 & = & \!\!\!\!\int_{\Sigma}\!d^{d-1}x\,\left(-\cosh\left(s_{2}\right)\frac{\partial F}{\partial x^{0}}\left(\bar{x}^{s_{2}}\right)+\sinh\left(s_{2}\right)\frac{\partial F}{\partial x^{1}}\left(\bar{x}^{s_{2}}\right)\right)\left(-x^{1}\frac{\partial F}{\partial x^{0}}\left(\bar{x}\right)\right)\nonumber \\
 & = & \!\!\!\!\int_{\Sigma}\!d^{d-1}x\,\left(-\frac{\partial F}{\partial x^{0}}\left(\bar{x}\right)\right)\left(-x^{1}\cosh\left(s_{2}\right)\frac{\partial F}{\partial x^{0}}\left(\bar{x}^{s_{2}}\right)+x^{1}\sinh\left(s_{2}\right)\frac{\partial F}{\partial x^{1}}\left(\bar{x}^{s_{2}}\right)\right)\nonumber \\
 & = & \!\!\!\!\frac{\mathrm{d}}{\mathrm{d}s_{2}}\int_{\Sigma}\!d^{d-1}x\,\left(-\frac{\partial F}{\partial x^{0}}\left(\bar{x}\right)\right)F\left(\bar{x}^{s_{2}}\right)\nonumber \\
 & = & \!\!\!\!\left.\frac{\mathrm{d}}{\mathrm{d}s_{2}}\right|_{s_{1}=0}\int_{\Sigma}\!d^{d-1}x\,\left(-\cosh\left(s_{1}\right)\frac{\partial F}{\partial x^{0}}\left(\bar{x}^{s_{1}}\right)+\sinh\left(s_{1}\right)\frac{\partial F}{\partial x^{1}}\left(\bar{x}^{s_{1}}\right)\right)F\left(\bar{x}^{s_{2}}\right)\nonumber \\
 & = & \!\!\!\!\left.\frac{\mathrm{d}}{\mathrm{d}s_{2}}\right|_{s_{1}=0}\int_{\Sigma}\!d^{d-1}x\,f_{\varphi}^{s_{1}}\left(\bar{x}\right)f_{\pi}^{s_{2}}\left(\bar{x}\right)=\left.\frac{\mathrm{d}Q}{\mathrm{d}s_{2}}\right|_{s_{1}=0}.\label{eq:RaQ}
\end{eqnarray}
Similarly we start with
\begin{eqnarray*}
\left.\frac{\mathrm{d}Q}{\mathrm{d}s_{1}}\right|_{s_{1}=0}\!\!\!\! & = & \!\!\!\!\left.\frac{\mathrm{d}}{\mathrm{d}s_{1}}\right|_{s_{1}=0}\int_{\Sigma}\!d^{d-1}x\,f_{\varphi}^{s_{1}}\left(\bar{x}\right)f_{\pi}^{s_{2}}\left(\bar{x}\right)\\
 & = & \!\!\!\!\left.\frac{\mathrm{d}}{\mathrm{d}s_{1}}\right|_{s_{1}=0}\int_{\Sigma}\!d^{d-1}x\,\left(-\cosh\left(s_{1}\right)\frac{\partial F}{\partial x^{0}}\left(\bar{x}^{s_{1}}\right)+\sinh\left(s_{1}\right)\frac{\partial F}{\partial x^{1}}\left(\bar{x}^{s_{1}}\right)\right)F\left(\bar{x}^{s_{2}}\right)\\
 & = & \!\!\!\!\int_{\Sigma}\!d^{d-1}x\,\left(x^{1}\frac{\partial^{2}F}{\left(\partial x^{0}\right)^{2}}\left(\bar{x}\right)+\frac{\partial F}{\partial x^{1}}\left(\bar{x}\right)\right)F\left(\bar{x}^{s_{2}}\right)\\
 & = & \!\!\!\!\int_{\Sigma}\!d^{d-1}x\,\left(x^{1}\left(\nabla^{2}-m^{2}\right)F\left(\bar{x}\right)+\frac{\partial F}{\partial x^{1}}\left(\bar{x}\right)\right)F\left(\bar{x}^{s_{2}}\right)\,.
\end{eqnarray*}
First we integrate the Laplacian term by parts,
\begin{eqnarray*}
\left.\frac{\mathrm{d}Q}{\mathrm{d}s_{1}}\right|_{s_{1}=0}\!\!\!\! & = & \!\!\!\!-\int_{\Sigma}\!d^{d-1}x\,x^{1}m^{2}F\left(\bar{x}\right)F\left(\bar{x}^{s_{2}}\right)-\int_{\Sigma}\!d^{d-1}x\,x^{1}\nabla_{\bot}F\left(\bar{x}\right)\cdot\nabla_{\bot}F\left(\bar{x}^{s_{2}}\right)\\
 &  & \!\!\!\!-\int_{\Sigma}\!d^{d-1}x\,x^{1}\frac{\partial F}{\partial x^{1}}\left(\bar{x}\right)\left(-\sinh\left(s_{2}\right)\frac{\partial F}{\partial x^{0}}\left(\bar{x}^{s_{2}}\right)+\cosh\left(s_{2}\right)\frac{\partial F}{\partial x^{1}}\left(\bar{x}^{s_{2}}\right)\right).
\end{eqnarray*}
After a second integration by parts we get
\begin{eqnarray*}
\left.\frac{\mathrm{d}Q}{\mathrm{d}s_{1}}\right|_{s_{1}=0}\!\!\!\! & = & \!\!\!\!\int_{\Sigma}\!d^{d-1}x\,x^{1}F\left(\bar{x}\right)\left(\nabla_{\bot}^{2}-m^{2}\right)F\left(\bar{x}^{s_{2}}\right)+\int_{\Sigma}\!d^{d-1}x\,F\left(\bar{x}\right)\left(-\sinh\left(s_{2}\right)\frac{\partial F}{\partial x^{0}}\left(\bar{x}^{s_{2}}\right)+\cosh\left(s_{2}\right)\frac{\partial F}{\partial x^{1}}\left(\bar{x}^{s_{2}}\right)\right)\\
 &  & \!\!\!\!+\int_{\Sigma}\!d^{d-1}x\,x^{1}F\left(\bar{x}\right)\left(\sinh^{2}\left(s_{2}\right)\frac{\partial^{2}F}{\left(\partial x^{0}\right)^{2}}\left(\bar{x}^{s_{2}}\right)-2\sinh\left(s_{2}\right)\cosh\left(s_{2}\right)\frac{\partial^{2}F}{\partial x^{0}\partial x^{1}}\left(\bar{x}^{s_{2}}\right)+\cosh^{2}\left(s_{2}\right)\frac{\partial^{2}F}{\left(\partial x^{1}\right)^{2}}\left(\bar{x}^{s_{2}}\right)\right).
\end{eqnarray*}
Now we form a Laplacian term in the first line and we use the equation
of motion for $F$,
\begin{eqnarray*}
\left.\frac{\mathrm{d}Q}{\mathrm{d}s_{1}}\right|_{s_{1}=0}\!\!\!\! & = & \!\!\!\!\int_{\Sigma}\!d^{d-1}x\,x^{1}F\left(\bar{x}\right)\frac{\partial^{2}F}{\left(\partial x^{0}\right)^{2}}\left(\bar{x}^{s_{2}}\right)+\int_{\Sigma}\!d^{d-1}x\,F\left(\bar{x}\right)\left(-\sinh\left(s_{2}\right)\frac{\partial F}{\partial x^{0}}\left(\bar{x}^{s_{2}}\right)+\cosh\left(s_{2}\right)\frac{\partial F}{\partial x^{1}}\left(\bar{x}^{s_{2}}\right)\right)\\
 &  & \!\!\!\!+\int_{\Sigma}\!d^{d-1}x\,x^{1}F\left(\bar{x}\right)\left(\sinh^{2}\left(s_{2}\right)\frac{\partial^{2}F}{\left(\partial x^{0}\right)^{2}}\left(\bar{x}^{s_{2}}\right)-2\sinh\left(s_{2}\right)\cosh\left(s_{2}\right)\frac{\partial^{2}F}{\partial x^{0}\partial x^{1}}\left(\bar{x}^{s_{2}}\right)+\sinh^{2}\left(s_{2}\right)\frac{\partial^{2}F}{\left(\partial x^{1}\right)^{2}}\left(\bar{x}^{s_{2}}\right)\right).
\end{eqnarray*}
Finally, a straightforward computation shows that
\begin{eqnarray}
\left.\frac{\mathrm{d}Q}{\mathrm{d}s_{1}}\right|_{s_{1}=0}\!\!\!\! & = & \!\!\!\!\int_{\Sigma}\!d^{d-1}x\,\frac{\mathrm{d}}{\mathrm{d}s_{2}}\left(-\cosh\left(s_{2}\right)\frac{\partial F}{\partial x^{0}}\left(\bar{x}^{s_{2}}\right)+\sinh\left(s_{2}\right)\frac{\partial F}{\partial x^{1}}\left(\bar{x}^{s_{2}}\right)\right)F\left(\bar{x}\right)\nonumber \\
 & = & \!\!\!\!\left.\frac{\mathrm{d}}{\mathrm{d}s_{2}}\right|_{s_{1}=0}\int_{\Sigma}\!d^{d-1}x\,\left(-\cosh\left(s_{2}\right)\frac{\partial F}{\partial x^{0}}\left(\bar{x}^{s_{2}}\right)+\sinh\left(s_{2}\right)\frac{\partial F}{\partial x^{1}}\left(\bar{x}^{s_{2}}\right)\right)F\left(\bar{x}^{s_{1}}\right)\nonumber \\
 & = & \!\!\!\!\left.\frac{\mathrm{d}}{\mathrm{d}s_{2}}\right|_{s_{1}=0}\int_{\Sigma}\!d^{d-1}x\,f_{\varphi}^{s_{2}}\left(\bar{x}\right)f_{\pi}^{s_{1}}\left(\bar{x}\right)=\left.\frac{\mathrm{d}R}{\mathrm{d}s_{2}}\right|_{s_{1}=0}.\label{eq:QaR}
\end{eqnarray}
Using \eqref{eq:RaQ} and \eqref{eq:QaR} we arrive at \eqref{eq:cool_relation_2}.

\subsection{Analytic continuation for $N\left(s\right)$\label{subsec:Analytic-continuation}}

In order to show that formulas \eqref{eq: is_zero} hold, we need
to explicitly show the analytic continuation for the function 
\begin{equation}
N\left(s\right)=\frac{i}{2}\left(Q\left(0,s\right)-R\left(0,s\right)\right)-\frac{1}{2}\left\Vert g_{R}^{s}\right\Vert _{\mathfrak{H}}^{2}\textrm{ ,}\label{eq:N(s)}
\end{equation}
or more specifically, we need to show that there exists a continuous
function $\tilde{N}:\mathbb{R}+i\left[0,2\pi\right]\rightarrow\mathbb{C}$,
analytic on $\mathbb{R}+i\left(0,2\pi\right)$ such that
\begin{equation}
\tilde{N}\left(s+i0\right)=N\left(s\right)\textrm{ .}
\end{equation}
To begin with, we notice that

\begin{eqnarray}
\frac{i}{2}Q\left(0,s\right)\!\!\!\! & = & \!\!\!\!\frac{i}{2}\int_{x^{1}>0}d^{d-1}x\,f_{\varphi}\left(\bar{x}\right)f_{\pi}^{s}\left(\bar{x}\right)=i\,\mathrm{Im}\left\langle f_{\varphi,R},f_{\pi,R}^{s}\right\rangle _{\mathfrak{H}}\,,\\
\frac{i}{2}R\left(0,s\right)\!\!\!\! & = & \!\!\!\!\frac{1}{2}\int_{x^{1}>0}d^{d-1}x\,f_{\varphi}^{s}\left(\bar{x}\right)f_{\pi}\left(\bar{x}\right)=i\,\mathrm{Im}\left\langle f_{\varphi,R}^{s},f_{\pi,R}\right\rangle _{\mathfrak{H}}\,,
\end{eqnarray}
where the above expressions make sense regardless of $f_{\pi,R}^{s}\notin\mathfrak{H}$.
This is because 
\begin{equation}
\left\langle f_{\varphi,R},f_{\pi,R}^{s}\right\rangle _{\mathfrak{H}}=\int_{\mathbb{R}^{d-1}}\frac{d^{d-1}p}{2\omega_{\bar{p}}}\hat{f}_{\varphi,R}\left(\bar{p}\right)^{*}i\omega_{\bar{p}}\hat{f}_{\pi,R}^{s}\left(\bar{p}\right)=\frac{i}{2}\left\langle \hat{f}_{\varphi,R},\hat{f}_{\pi,R}^{s}\right\rangle _{L^{2}}\,,
\end{equation}
which is convergent. The problem involving scalar products of split
functions $f_{\varphi,R}^{s}$ and $f_{\pi,R}^{s}$ happens only when
we try to compute the scalar product of two sharply cut test functions
of the momentum operator, e.g. 
\begin{equation}
\left\langle f_{\text{\ensuremath{\pi}},R},f_{\pi,R}^{s}\right\rangle _{\mathfrak{H}}=\int_{\mathbb{R}^{d-1}}\frac{d^{d-1}p}{2\omega_{\bar{p}}}\left(i\omega_{\bar{p}}\hat{f}_{\pi,R}\left(\bar{p}\right)\right)^{*}i\omega_{\bar{p}}\hat{f}_{\pi,R}^{s}\left(\bar{p}\right)=\frac{1}{2}\left\langle \hat{f}_{\pi,R},\hat{f}_{\pi,R}^{s}\right\rangle _{H^{\frac{1}{2}}}\,,\label{eq: scalar_pr_bad}
\end{equation}
which is in general divergent. Such divergency comes from the noncontinuity of the function $f_{\pi,R}\left(\bar{x}\right)=f_{\pi}\left(\bar{x}\right)\Theta\left(x^{1}\right)$
at $x^{1}=0$. To overcome this difficulty we introduce a family of
smooth functions (for $\epsilon>0$) 
\begin{equation}
f_{\varphi,R}^{\varepsilon}\left(\bar{x}\right):=f_{\varphi}\left(\bar{x}\right)\Theta_{\varepsilon}\left(x^{1}\right)\quad\mathrm{and}\quad f_{\pi,R}^{\varepsilon}\left(\bar{x}\right):=f_{\pi}\left(\bar{x}\right)\Theta_{\varepsilon}\left(x^{1}\right)\,,
\end{equation}
where $\Theta_{\varepsilon}\in C^{\infty}\left(\mathbb{R}\right)$
is a regularized Heaviside function such that
\begin{equation}
\Theta_{\varepsilon}\left(t\right)=\begin{cases}
0 & \textrm{if }t\leq\frac{\varepsilon}{2}\\
1 & \textrm{if }t\geq\varepsilon
\end{cases}\textrm{ .}\label{eq:step_smooth}
\end{equation}
Then
\begin{eqnarray}
f_{\varphi,R}^{\varepsilon}\left(\bar{x}\right)\underset{\epsilon\rightarrow0^{+}}{\longrightarrow}f_{\varphi,R}\left(\bar{x}\right) & \textrm{ and } & f_{\pi,R}^{\varepsilon}\left(\bar{x}\right)\underset{\epsilon\rightarrow0^{+}}{\longrightarrow}f_{\pi,R}\left(\bar{x}\right)\,,\label{eq: conv_funcs}
\end{eqnarray}
where the above convergence must be in a sense that we specify
opportunely below. Before we get into such convergence issues, we
notice that $f_{\varphi,R}^{\varepsilon},f_{\pi,R}^{\varepsilon}\in\mathcal{S}\left(\mathbb{R}^{d-1},\mathbb{R}\right)$
and hence the scalar product \eqref{eq: scalar_pr_bad} is now well
defined. Then we define the function
\begin{equation}
N^{\epsilon}\left(s\right):=i\,\mathrm{Im}\left\langle f_{\varphi,R}^{\epsilon},f_{\pi,R}^{s,\epsilon}\right\rangle _{\mathfrak{H}}-i\,\mathrm{Im}\left\langle f_{\varphi,R}^{s,\epsilon},f_{\pi,R}^{\epsilon}\right\rangle _{\mathfrak{H}}-\frac{1}{2}\left\Vert g_{R}^{s,\epsilon}\right\Vert _{\mathfrak{H}}^{2}\,,\label{eq:Ne(s)}
\end{equation}
which is just the regularized version of \eqref{eq:N(s)}. In the
next subsection we show that $N^{\epsilon}\left(s\right)\rightarrow N\left(s\right)$
when $\epsilon\rightarrow0^{+}$. Expression \ref{eq:Ne(s)} can be
rewritten as
\begin{eqnarray}
\hspace{-1cm}N^{\epsilon}\left(s\right)\!\!\!\! & = & \!\!\!\!i\,\mathrm{Im}\left\langle f_{\varphi,R}^{\epsilon},f_{\pi,R}^{s,\epsilon}\right\rangle _{\mathfrak{H}}-i\,\mathrm{Im}\left\langle f_{\varphi,R}^{s,\epsilon},f_{\pi,R}^{\epsilon}\right\rangle _{\mathfrak{H}}-\frac{1}{2}\left\Vert g_{R}^{s,\epsilon}\right\Vert _{\mathfrak{H}}^{2}\nonumber \\
 & = & \!\!\!\!i\,\mathrm{Im}\left\langle f_{\varphi,R}^{\epsilon},f_{\pi,R}^{s,\epsilon}\right\rangle _{\mathfrak{H}}+i\,\mathrm{Im}\left\langle f_{\pi,R}^{\epsilon},f_{\varphi,R}^{s,\epsilon}\right\rangle _{\mathfrak{H}}-\frac{1}{2}\left\langle f_{R}^{\epsilon}-f_{R}^{s,\epsilon},f_{R}^{\epsilon}-f_{R}^{s,\epsilon}\right\rangle _{\mathfrak{H}}\nonumber \\
 & = & \!\!\!\!\left\langle f_{R}^{\epsilon},f_{R}^{s,\epsilon}\right\rangle _{\mathfrak{H}}-\frac{1}{2}\left\langle f_{R}^{\epsilon},f_{R}^{\epsilon}\right\rangle _{\mathfrak{H}}-\frac{1}{2}\left\langle f_{R}^{s,\epsilon},f_{R}^{s,\epsilon}\right\rangle _{\mathfrak{H}}=\left\langle f_{R}^{-\frac{s}{2},\epsilon},f_{R}^{\frac{s}{2},\epsilon}\right\rangle _{\mathfrak{H}}-\left\langle f_{R}^{\epsilon},f_{R}^{\epsilon}\right\rangle _{\mathfrak{H}}\nonumber \\
 & = & \!\!\!\!\int_{\mathbb{R}^{d-1}}\!\!\frac{d^{d-1}p}{2\omega_{\bar{p}}}\left[\left(\hat{f}_{\varphi,R}^{-\frac{s}{2},\epsilon}+i\omega_{\bar{p}}\hat{f}_{\pi,R}^{-\frac{s}{2},\epsilon}\right)^{*}\left(\hat{f}_{\varphi,R}^{\frac{s}{2},\epsilon}+i\omega_{\bar{p}}\hat{f}_{\pi,R}^{\frac{s}{2},\epsilon}\right)-\left(\hat{f}_{\varphi,R}^{\epsilon}+i\omega_{\bar{p}}\hat{f}_{\pi,R}^{\epsilon}\right)^{*}\left(\hat{f}_{\varphi,R}^{\epsilon}+i\omega_{\bar{p}}\hat{f}_{\pi,R}^{\epsilon}\right)\right],\label{eq:estudiar_conver}
\end{eqnarray}
where in the penultimate line we have used that $f_{R}^{s_{1}+s_{2},\epsilon}=u\left(\Lambda_{1}^{s_{2}}\right)f_{R}^{s_{1},\epsilon}$
for all $s_{1},s_{2}\in\mathbb{R}$. For a moment, let assume that
this last expression converges to
\begin{equation}
N\left(s\right)=\int_{\mathbb{R}^{d-1}}\frac{d^{d-1}p}{2\omega_{\bar{p}}}\left[\left(\hat{f}_{\varphi,R}^{-\frac{s}{2}}+i\omega_{\bar{p}}\hat{f}_{\pi,R}^{-\frac{s}{2}}\right)^{*}\left(\hat{f}_{\varphi,R}^{\frac{s}{2}}+i\omega_{\bar{p}}\hat{f}_{\pi,R}^{\frac{s}{2}}\right)-\left(\hat{f}_{\varphi,R}+i\omega_{\bar{p}}\hat{f}_{\pi,R}\right)^{*}\left(\hat{f}_{\varphi,R}+i\omega_{\bar{p}}\hat{f}_{\pi,R}\right)\right]\,,\label{eq: estudiar_anal}
\end{equation}
when $\epsilon\rightarrow0^{+}$. We prove this in the next subsection.
The second term of the above integrand is independent on $s$ and
hence its analytic continuation is trivial. Let us then focus on the
first term. Using the Poincaré covariance and causality of the Klein-Gordon
equation, it is not difficult to show that

\begin{equation}
\hat{f}_{\varphi,R}^{s}\left(\bar{p}\right)+i\omega_{\bar{p}}\hat{f}_{\pi,R}^{s}\left(\bar{p}\right)=\hat{f}_{\varphi,R}\left(\Lambda_{1}^{s}\bar{p}\right)+i\,\Lambda_{1}^{s}\omega_{\bar{p}}\,\hat{f}_{\pi,R}\left(\Lambda_{1}^{s}\bar{p}\right)\,,\label{eq:cortar_boostear}
\end{equation}
where $\Lambda_{1}^{s}\bar{p}=\left(p^{1}\cosh\left(s\right)-\omega_{\bar{p}}\sinh\left(s\right),\bar{p}_{\bot}\right)$
and $\Lambda_{1}^{s}\omega_{\bar{p}}=\omega_{\bar{p}}\cosh\left(s\right)-p^{1}\sinh\left(s\right)$.
Then, the first integrand term of \eqref{eq: estudiar_anal} becomes
\begin{eqnarray}
 &  & \!\!\!\!\left(\hat{f}_{\varphi,R}^{-\frac{s}{2},\epsilon}\left(\bar{p}\right)+i\omega_{\bar{p}}\hat{f}_{\pi,R}^{-\frac{s}{2},\epsilon}\left(\bar{p}\right)\right)^{*}\left(\hat{f}_{\varphi,R}^{\frac{s}{2},\epsilon}\left(\bar{p}\right)+i\omega_{\bar{p}}\hat{f}_{\pi,R}^{\frac{s}{2},\epsilon}\left(\bar{p}\right)\right)\nonumber \\
 & = & \!\!\!\!\int_{\mathbb{R}^{2\left(d-1\right)}}\negthickspace d^{d-1}x\,d^{d-1}y\left(f_{\varphi,R}\left(\bar{x}\right)-i\omega_{\bar{p}}f_{\pi,R}\left(\bar{x}\right)\right)\left(f_{\varphi,R}\left(\bar{y}\right)+i\omega_{\bar{p}}f_{\pi,R}\left(\bar{y}\right)\right)\mathrm{e}^{i\Lambda^{-\frac{s}{2}}\left(\bar{p}\right)\cdot\bar{x}}\mathrm{e}^{-i\Lambda^{\frac{s}{2}}\left(\bar{p}\right)\cdot\bar{y}}\,,\label{eq:pi_ext}
\end{eqnarray}
where $-i\Lambda^{\frac{s}{2}}\left(\bar{p}\right)\cdot\bar{y}=-i\left(-\sinh\left(\frac{s}{2}\right)\omega_{\bar{p}}+\cosh\left(\frac{s}{2}\right)p^{1}\right)y^{1}-i\bar{p}_{\bot}\cdot\bar{y}_{\bot}$,
and equivalently for $i\Lambda^{-\frac{s}{2}}\left(\bar{p}\right)\cdot\bar{x}$.
Then 
\begin{eqnarray}
 & \!\!\!\!\!\!\!\!\!\!\!\!\!\!\!\! & \!\!\!\!-i\left(-\sinh\left(\frac{s}{2}\right)\omega_{\bar{p}}+\cosh\left(\frac{s}{2}\right)p^{1}\right)y^{1}\nonumber \\
 & \!\!\!\!\!\!\!\!\!\!\!\!\!\!\!\!\underset{s\rightarrow s+i\sigma}{\longrightarrow} & \!\!\!\!-i\left(-\sinh\left(\frac{s+i\sigma}{2}\right)\omega_{\bar{p}}+\cosh\left(\frac{s+i\sigma}{2}\right)p^{1}\right)y^{1}\nonumber \\
 & \!\!\!\!\!\!\!\!\!\!\!\!\!\!\!\!= & \!\!\!\!-i\left(-\sinh\left(\frac{s}{2}\right)\omega_{\bar{p}}+\cosh\left(\frac{s}{2}\right)p^{1}\right)y^{1}\cos\left(\frac{\sigma}{2}\right)-\underset{\geq m}{\underbrace{\left(\cosh\left(\frac{s}{2}\right)\omega_{\bar{p}}-\sinh\left(\frac{s}{2}\right)p^{1}\right)}}y^{1}\sin\left(\frac{\sigma}{2}\right)\textrm{ ,}
\end{eqnarray}
where the second term provides an exponential dumping in equation
\eqref{eq:pi_ext} when $\sigma\in\left(0,2\pi\right)$ because $supp\left(f_{\varphi,R}\right),$
$supp\left(f_{\text{\ensuremath{\pi}},R}\right)\subset\Sigma$. Equivalently
it can be shown that $i\Lambda^{-\frac{s}{2}}\left(\bar{p}\right)\cdot\bar{x}$
also provides an exponential dumping for $\sigma\in\left(0,2\pi\right)$.
Hence we have that
\begin{equation}
\tilde{N}\left(s+i\sigma\right)\textrm{ is an analytic function for }s+i\sigma\in\mathbb{R}+i\left(0,2\pi\right)\textrm{. }
\end{equation}
Looking at expressions \eqref{eq: estudiar_anal} and \eqref{eq:pi_ext},
it is easy to determine that 
\begin{equation}
\lim_{\sigma\rightarrow2\pi^{-},s=0}\tilde{N}\left(s+i\sigma\right)=0\,.
\end{equation}

\subsubsection{Convergence of $N^{\epsilon}\left(s\right)$}

In order to show that expression \eqref{eq: estudiar_anal} holds,
we need to prove the following two limits
\begin{eqnarray}
\!\!\!\!\!\!\!\!\!\!\!\!\!\!\!\!N^{\epsilon}\left(s\right)\!\!\!\! & \underset{\epsilon\rightarrow0^{+}}{\longrightarrow} & \!\!\!\!N\left(s\right)=\frac{i}{2}\left(Q\left(0,s\right)-R\left(0,s\right)\right)-\frac{1}{2}\left\Vert g_{R}^{s}\right\Vert _{\mathfrak{H}}^{2}\,,\label{eq:limite1}\\
\!\!\!\!\!\!\!\!\!\!\!\!\!\!\!\!N^{\epsilon}\left(s\right)\!\!\!\! & \underset{\epsilon\rightarrow0^{+}}{\longrightarrow} & \!\!\!\!\int_{\mathbb{R}^{d-1}}\frac{d^{d-1}p}{2\omega_{\bar{p}}}\left[\left(\hat{f}_{\varphi,R}^{-\frac{s}{2}}+i\omega_{\bar{p}}\hat{f}_{\pi,R}^{-\frac{s}{2}}\right)^{*}\left(\hat{f}_{\varphi,R}^{\frac{s}{2}}+i\omega_{\bar{p}}\hat{f}_{\pi,R}^{\frac{s}{2}}\right)-\left(\hat{f}_{\varphi,R}+i\omega_{\bar{p}}\hat{f}_{\pi,R}\right)^{*}\left(\hat{f}_{\varphi,R}+i\omega_{\bar{p}}\hat{f}_{\pi,R}\right)\right].\label{eq: limite2}
\end{eqnarray}
To do this, we must be precise in which sense the functions $f_{\varphi,R}^{s,\epsilon},f_{\pi,R}^{s,\epsilon}$
converge in \eqref{eq: conv_funcs}. To begin we choose the following
smooth step function \eqref{eq:step_smooth}
\begin{equation}
\Theta_{\varepsilon}\left(t\right)=\begin{cases}
0 & \textrm{if }t\leq\frac{\varepsilon}{2}\,,\\
\left[\exp\left(\frac{\epsilon\left(t-\frac{3\epsilon}{4}\right)}{\left(t-\frac{3\epsilon}{4}\right)^{2}-\left(\frac{\epsilon}{4}\right)^{2}}\right)+1\right]^{-1} & \textrm{if }\frac{\varepsilon}{2}<t<\varepsilon\,,\\
1 & \textrm{if }t\geq\varepsilon\textrm{ .}
\end{cases}\label{eq: theta_reg}
\end{equation}
First we focus on the limit \eqref{eq: limite2}. Looking back
to \eqref{eq:estudiar_conver}, we can rewrite the r.h.s. of that
expression as
\begin{eqnarray}
N^{\epsilon}\left(s\right)\!\!\!\! & = & \!\!\!\!\left\langle f_{R}^{-\frac{s}{2},\epsilon},f_{R}^{\frac{s}{2},\epsilon}\right\rangle _{\mathfrak{H}}-\left\langle f_{R}^{\epsilon},f_{R}^{\epsilon}\right\rangle _{\mathfrak{H}}=\left\langle f_{R}^{-\frac{s}{2},\epsilon},f_{R}^{\frac{s}{2},\epsilon}-f_{R}^{\epsilon}+f_{R}^{\epsilon}\right\rangle _{\mathfrak{H}}-\left\langle f_{R}^{\epsilon},f_{R}^{\epsilon}\right\rangle _{\mathfrak{H}}\nonumber \\
 & = & \!\!\!\!\left\langle f_{\varphi,R}^{-\frac{s}{2},\epsilon},f_{\varphi,R}^{\frac{s}{2},\epsilon}-f_{\varphi,R}^{\epsilon}\right\rangle _{\mathfrak{H}}+\left\langle f_{\varphi,R}^{-\frac{s}{2},\epsilon},f_{\pi,R}^{\frac{s}{2},\epsilon}-f_{\pi,R}^{\epsilon}\right\rangle _{\mathfrak{H}}+\left\langle f_{\varphi,R}^{-\frac{s}{2},\epsilon}-f_{\varphi,R}^{\epsilon},f_{\varphi,R}^{\epsilon}\right\rangle _{\mathfrak{H}}\nonumber \\
 &  & \!\!\!\!+\left\langle f_{\pi,R}^{-\frac{s}{2},\epsilon}-f_{\pi,R}^{\epsilon},f_{\varphi,R}^{\epsilon}\right\rangle _{\mathfrak{H}}+\left\langle f_{\pi,R}^{-\frac{s}{2},\epsilon},f_{\varphi,R}^{\frac{s}{2},\epsilon}-f_{\varphi,R}^{\epsilon}\right\rangle _{\mathfrak{H}}+\left\langle f_{\varphi,R}^{-\frac{s}{2},\epsilon}-f_{\varphi,R}^{\epsilon},f_{\pi,R}^{\epsilon}\right\rangle _{\mathfrak{H}}\\
 &  & \!\!\!\!+\underset{\varoast}{\underbrace{\left\langle f_{\pi,R}^{-\frac{s}{2},\epsilon},f_{\pi,R}^{\frac{s}{2},\epsilon}-f_{\pi,R}^{\epsilon}\right\rangle _{\mathfrak{H}}}}+\underset{\varoast}{\underbrace{\left\langle f_{\pi,R}^{-\frac{s}{2},\epsilon}-f_{\pi,R}^{\epsilon},f_{\pi,R}^{\epsilon}\right\rangle _{\mathfrak{H}}}}\,.\label{eq:analizar_conver}
\end{eqnarray}
It is not difficult to see that
\begin{eqnarray}
f_{\varphi,R}^{s,\epsilon}\underset{\epsilon\rightarrow0^{+}}{\longrightarrow}f_{\varphi,R}^{s} & \textrm{ and } & f_{\pi,R}^{s,\epsilon}\underset{\epsilon\rightarrow0^{+}}{\longrightarrow}f_{\pi,R}^{s}\,,\quad\textrm{in }L^{2}\left(\mathbb{R}^{d-1}\right)\,,\label{eq:conv_l2}
\end{eqnarray}
which implies that all terms in \eqref{eq:analizar_conver} are convergent,
except perhaps those pointed by $\varoast$. Now we concentrate in
those remaining terms, e.g.
\begin{eqnarray}
\left\langle f_{\pi,R}^{-\frac{s}{2},\epsilon},f_{\pi,R}^{\frac{s}{2},\epsilon}-f_{\pi,R}^{\epsilon}\right\rangle _{\mathfrak{H}}\!\!\!\! & = & \!\!\!\!\frac{1}{2}\int_{\mathbb{R}^{d-1}}d^{d-1}p\,\hat{f}_{\pi,R}^{-\frac{s}{2},\epsilon}\left(\bar{p}\right)\left(\hat{f}_{\pi,R}^{-\frac{s}{2},\epsilon}\left(\bar{p}\right)-\hat{f}_{\pi,R}^{\epsilon}\left(\bar{p}\right)\right)\omega_{\bar{p}}\,.\label{eq:conv_casi_final}
\end{eqnarray}
The convergence of \eqref{eq:conv_casi_final} is guaranteed by the
fact that 
\begin{eqnarray}
f_{\pi,R}^{\epsilon}-f_{\pi,R}^{-\frac{s}{2},\epsilon} & \underset{\epsilon\rightarrow0^{+}}{\longrightarrow} & f_{\pi,R}-f_{\pi,R}^{-\frac{s}{2}}\quad\textrm{in }H^{1}\left(\mathbb{R}^{d-1}\right)\,,\label{eq:conv_final}\\
 & \Downarrow\nonumber \\
\left(\hat{f}_{\pi,R}^{\epsilon}-\hat{f}_{\pi,R}^{-\frac{s}{2},\epsilon}\right)\omega_{\bar{p}} & \underset{\epsilon\rightarrow0^{+}}{\longrightarrow} & \left(\hat{f}_{\pi,R}-\hat{f}_{\pi,R}^{-\frac{s}{2}}\right)\omega_{\bar{p}}\quad\textrm{in }L^{2}\left(\mathbb{R}^{d-1}\right)\,.\label{eq:conv_final_bis}
\end{eqnarray}
In order to probe \eqref{eq:conv_final} we remember that $f_{\pi,R}\left(\bar{x}\right)-f_{\pi,R}^{s}\left(\bar{x}\right)=g_{\pi}^{s}\left(\bar{x}\right)\Theta\left(x^{1}\right)$
with $g_{\pi}^{s}\in\mathcal{S}\left(\mathbb{R}^{d-1},\mathbb{R}\right)$
and $\left.g_{\pi}^{s}\right|_{x^{1}=0}=0$. Then the following lemma
ensures \eqref{eq:conv_final}. 
\begin{lem}
\textup{Let $g\in\mathcal{S}\left(\mathbb{R}^{n}\right)$ with $\left.g\right|_{x^{1}=0}=0$,
$g_{R}\left(\bar{x}\right)=g\left(\bar{x}\right)\Theta\left(x^{1}\right)$
and $g_{R}^{\epsilon}\left(\bar{x}\right)=g\left(\bar{x}\right)\Theta_{\epsilon}\left(x^{1}\right)$
with $\Theta_{\epsilon}$ as \eqref{eq: theta_reg}. Then $g_{R}\in H^{1}\left(\mathbb{R}^{n}\right)$
and $g_{R}^{\epsilon}\underset{\epsilon\rightarrow0^{+}}{\longrightarrow}g_{R}$
in $H^{1}\left(\mathbb{R}^{n}\right)$.} 
\end{lem}
\begin{proof}
The fact that $g_{R}\in H^{1}\left(\mathbb{R}^{n}\right)$ is guaranteed
by lemma \ref{par:lemma-1}. Then we prove the convergence for $n=1$.
The generalization to $n>1$ is straightforward. Since $g_{R}$ and $g_{R}^{\epsilon}$
satisfy the hypothesis of the lemma \ref{lem:derivatives}, their
weak derivatives coincide with theirs pointwise derivatives and hence
\begin{eqnarray}
\!\!\!\!\!\!\!\!\!\!\!\!\!\!\left\Vert g_{R}^{\epsilon}-g_{R}\right\Vert '^{2}{}_{H^{1}}\!\!\!\! & = & \!\!\!\!\int_{-\infty}^{+\infty}dx\left|g\left(x\right)\Theta_{\epsilon}\left(x\right)-g\left(x\right)\Theta\left(x\right)\right|^{2}+\int_{-\infty}^{+\infty}dx\left|\partial_{x}\left[g\left(x\right)\Theta_{\epsilon}\left(x\right)-g\left(x\right)\Theta\left(x\right)\right]\right|^{2}\nonumber \\
 & \leq & \!\!\!\!\int_{-\infty}^{+\infty}dx\left|g\left(x\right)\right|^{2}\left|\Theta_{\epsilon}\left(x\right)-\Theta\left(x\right)\right|^{2}+\int_{-\infty}^{+\infty}dx\left|g'\left(x\right)\right|^{2}\left|\Theta_{\epsilon}\left(x\right)-\Theta\left(x\right)\right|^{2}\nonumber \\
 &  & \!\!\!\!+\int_{-\infty}^{+\infty}dx\left|g\left(x\right)\right|^{2}\left|\Theta'_{\epsilon}\left(x\right)\right|^{2}+2\int_{-\infty}^{+\infty}dx\left|g\left(x\right)\right|\left|g'\left(x\right)\right|\left|\Theta_{\epsilon}\left(x\right)-\Theta\left(x\right)\right|\left|\Theta'_{\epsilon}\left(x\right)\right|\nonumber \\
 & \leq & \!\!\!\!\int_{\frac{\epsilon}{2}}^{\epsilon}dx\left(\left|g\left(x\right)\right|^{2}+\left|g'\left(x\right)\right|^{2}\right)+\int_{\frac{\epsilon}{2}}^{\epsilon}dx\left|g\left(x\right)\right|^{2}\left|\Theta'_{\epsilon}\left(x\right)\right|^{2}+2\int_{\frac{\epsilon}{2}}^{\epsilon}dx\left|g\left(x\right)\right|\left|g'\left(x\right)\right|\left|\Theta'_{\epsilon}\left(x\right)\right|\,.
\end{eqnarray}
We notice that since $g\in C^{\infty}\left(\mathbb{R}\right)$ and
$g\left(0\right)=0$, by the Taylor theorem we have that $g\left(x\right)=g'\left(0\right)x+r\left(x\right)x$
with $r\left(x\right)\underset{x\rightarrow0}{\longrightarrow}0$
and $r\in C^{\infty}\left(\mathbb{R}\right)$. We also have that $\max_{x\in\mathbb{R}}\left|\Theta'_{\epsilon}\left(x\right)\right|=\frac{4}{\epsilon}$,
which follows from the definition of that function . Then using the
above properties and assuming $0<\varepsilon\leq1$,
\begin{eqnarray}
\left\Vert g_{R}^{\epsilon}-g_{R}\right\Vert '^{2}{}_{H^{1}}\!\!\!\! & \leq & \!\!\!\!\underset{x\in\left[0,1\right]}{\max}\left(\left|g\left(x\right)\right|^{2}+\left|g'\left(x\right)\right|^{2}\right)\int_{\frac{\epsilon}{2}}^{\epsilon}dx+\underset{x\in\left[0,1\right]}{\max}\left|g'\left(0\right)+r\left(x\right)\right|^{2}\frac{16}{\epsilon^{2}}\int_{\frac{\epsilon}{2}}^{\epsilon}dx\,x^{2}\nonumber \\
 &  & \!\!\!\!+\underset{x\in\left[0,1\right]}{\max}\left|g'\left(x\right)\right|\underset{x\in\left[0,1\right]}{\max}\left|g'\left(0\right)+r\left(x\right)\right|\frac{8}{\epsilon}\int_{\frac{\epsilon}{2}}^{\epsilon}dx\,x\nonumber \\
 & \leq & \!\!\!\!\underset{x\in\left[0,1\right]}{\max}\left(\left|g\left(x\right)\right|^{2}+\left|g'\left(x\right)\right|^{2}\right)\frac{\epsilon}{2}+\underset{x\in\left[0,1\right]}{\max}\left|g'\left(0\right)+r\left(x\right)\right|^{2}\frac{14}{3}\epsilon\nonumber \\
 &  & \!\!\!\!+\underset{x\in\left[0,1\right]}{\max}\left|g'\left(x\right)\right|\underset{x\in\left[0,1\right]}{\max}\left|g'\left(0\right)+r\left(x\right)\right|3\epsilon\underset{\epsilon\rightarrow0^{+}}{\longrightarrow}0\,.
\end{eqnarray}
\end{proof}
Then we have that all terms in \eqref{eq:analizar_conver} converge.
By continuity of the scalar product, the limit of \eqref{eq:analizar_conver}
is just this same expression but evaluated at $\epsilon=0$, which
coincides with the l.h.s of \eqref{eq:conv_l2}.

We use the same arguments to prove the limit \eqref{eq:limite1}.
The first two terms of \eqref{eq:Ne(s)} are convergent due to \eqref{eq:conv_l2},
and the remaining term is also convergent due to \eqref{eq:conv_final}
and \eqref{eq:conv_final_bis}. Then by continuity of the scalar product
we have that
\begin{equation}
N^{\epsilon}\left(s\right)\underset{\epsilon\rightarrow0^{+}}{\longrightarrow}N\left(s\right)=\frac{i}{2}\left(Q\left(0,s\right)-R\left(0,s\right)\right)-\frac{1}{2}\left\Vert g_{R}^{s}\right\Vert _{\mathfrak{H}}^{2}\,.
\end{equation}
Finally, expression \eqref{eq: estudiar_anal} holds.

\bibliographystyle{utphys}
\bibliography{Bibliography}

\end{document}